\documentclass[AFST,Unicode]{cedram}

\usepackage{amsmath, latexsym, amscd, amssymb}
\usepackage{graphicx}
\usepackage{mathabx}
\usepackage{bm} 
\usepackage{alltt} 

\usepackage{soul}

\usepackage{xcolor}
\usepackage{array}
\usepackage{arydshln}
\usepackage{cmap}
\usepackage{hyperref}

\usepackage{enumerate}
\usepackage{booktabs}

\def\currentjournaltitle{}
\def\currentjournaltitleX{}
\def\currentjournalshorttitle{}
\makeatletter
\def\article@logo{%
  \set@logo{}%
}
\makeatother

\def \mmedskip {\vskip 4pt }

\def\rmd {\mkern 2mu {\rm d}}
\def\rme {\mkern 2mu {\rm e}}
\def\rmi {\mkern 2mu {\rm i}}

\def \virg		{\,\raise 2pt\hbox{\rm,}\, }

\renewcommand{\N}{\mathbb{N}}
\renewcommand{\Q}{\mathbb{Q}}
\renewcommand{\C}{\mathbb{C}} 
\renewcommand{\R}{\mathbb{R}}
\renewcommand{\Z}{\mathbb{Z}}

\newcommand{\sB}{\mathcal{B}}

\newcommand{\sD}{\mathcal{D}}

\newcommand{\sH}{\mathcal{H}}
\newcommand{\sI}{\mathcal{I}}

\newcommand{\sL}{\mathcal{L}}

\newcommand{\sK}{\mathcal{K}}
\newcommand{\sN}{\mathcal{N}}

\newcommand{\sP}{\mathcal{P}}

\newcommand{\pL}{\mathfrak{L}}

\newcommand{\ugo}{{\underline{\omega}}}

\newcommand{\abs}[1]{\mathopen|#1\mathclose|}

\newbox \cadre
\newdimen\htcadre
\newdimen\dpcadre

\def \shortboxit (#1,#2,#3,#4)#5%
{\setbox\cadre=\hbox{\kern #1 \relax {#5}\kern #2}%
 \htcadre=\ht\cadre 
\advance \htcadre by #3 
  \dpcadre=\dp\cadre 
 \advance \dpcadre by #4
  \setbox\cadre=\hbox{\vrule height \htcadre 
                                         depth  \dpcadre
                                         \box\cadre\vrule}
   \setbox\cadre=\vbox{\hrule\box\cadre\hrule}%
    \advance \htcadre by .4pt%
    \ht\cadre= \htcadre 
     \advance \dpcadre by .4pt
     \dp\cadre= \dpcadre  
     \box\cadre\relax
      }

\def\longboxit [#1](#2,#3,#4,#5)#6%
{\setbox\cadre=\hbox{\shortboxit({#2},{#3},{#4},{#5}){#6}}%
\htcadre=\ht\cadre        \advance \htcadre by -#1 
  \setbox\cadre=\hbox{\box\cadre \vrule height \htcadre width #1}%
\htcadre=\wd\cadre        \advance \htcadre by -#1 
  \setbox\cadre=\vbox{\hbox{\box\cadre} 
  \offinterlineskip
   \hbox{\kern #1\vbox{\hrule width \htcadre depth #1}}}%
\box\cadre\relax}                       

\def \choixboxit
{\ifx \token [  
\let \next = \longboxit 
\else 
\let \next = \shortboxit \fi  
\next }

\def\boxit {\futurelet \token \choixboxit }

\def \ensemble #1{\lbrace #1\rbrace}

\def\pv 		{\hskip 1.5pt plus .5pt minus 0.3pt
   		\string;\hskip 1pt plus .5pt minus 0.8pt}

\newcommand{\resp}{resp. }
\newcommand{\cf}{{\it cf.} }
\newcommand{\ie}{{\it i.e.}, }

\newcommand{\diag}{{\rm diag} }
\newcommand{\ga }{\alpha }
\newcommand{\gb }{\beta }
\newcommand{\gc }{\gamma }
\newcommand{\gd }{\delta }

\newcommand{\gve }{\varepsilon }
\newcommand{\gvf }{\varphi }

\newcommand{\gz }{\zeta }
\newcommand{\gl}{\lambda}

\newcommand{\go }{\omega }

\newcommand{\gt}{\theta}

\newcommand{\gC }{\Gamma }
\newcommand{\gD }{\Delta }

\newcommand{\gL}{\Lambda}

\newcommand{\hY}{\widehat Y}

\newcommand{\hy}{\widehat y}

\newcommand{\hH}{\widehat H}
\newcommand{\hK}{\widehat K}

\newcommand{\hh}{\widehat h}

\newcommand{\tY}{\widetilde Y}

\newcommand{\tH}{\widetilde{\sl H}}
\newcommand{\tK}{\widetilde{\sl K}}
\newcommand{\ty}{\widetilde y}

\newcommand{\tf}{\widetilde f}
\newcommand{\tg}{\widetilde g}
\newcommand{\tth}{\widetilde h}

\newcommand{\tq}{{\widetilde q}}

\def\tildebrack #1#2%
{{\mkern 5mu 
\scriptstyle\widetilde
{\mkern-2mu \textstyle #1}\mathstrut^{\mkern 3mu [#2]}%
}}

\def\hatbrack #1#2%
{{\mkern 5mu 
\scriptstyle\widehat
{\mkern-2mu \textstyle #1}\mathstrut^{\mkern 3mu [#2]}%
}}

\newtheorem{lemma}{Lemma}[section]

\newtheorem{cor}[lemma]{Corollary}
\newtheorem{pro}[lemma]{Proposition}
\newtheorem{dfn}[lemma]{Definition}
\newtheorem{hyp}[lemma]{Hypothesis}

\theoremstyle{remark}
\newtheorem{rem}[lemma]{Remark}
\newtheorem{exa}[lemma]{Example}

\newtheoremstyle{algorithm}{\smallskipamount}{\smallskipamount}{\normalfont}{\parindent}{\normalfont\scshape}{\pointir}{0pt}{}
\theoremstyle{algorithm}
\newtheorem{algorithm}{Algorithm}

\newcommand\omicron{\mathrm o}
\newcommand{\mathd}{\mathrm{d}}
\newcommand\codestar{\texttt}
\newcommand{\ddx}[1]{\frac{\mathd}{\mathd #1}}
\newcommand{\diff}[2]{\frac{\mathd^{#2}}{\mathd \, #1^{#2}}}
\providecommand{\tmat}[2]{\mathbf{T}_{#1, #2}}
\newcommand{\tmatstar}[2]{\mathbf{T}^{\ast}_{#1, #2}}
\newcommand{\bmat}[1]{\mathbf{B}_{#1}}
\newcommand{\bfmat}[1]{\mathbf{L}_{#1}}
\newcommand{\mmat}[1]{\mathbf{M}_{#1}}

\newcommand{\QQbar}{\bar{\mathbb{Q}}}
\newcommand{\localbasis}[1]{\hat{\mathcal{Y}}_{#1}}
\newcommand{\dir}{\underline{{\omega}}}
\newcommand{\Balls}{\mathbb{C}_{\bullet}}
\newcommand{\CComp}{\mathbb{C}_{\circlearrowleft}}
\newcommand{\bttlt}{<^{\uparrow}}

\newenvironment{enumerate*}{\begin{enumerate}}{\end{enumerate}}

\title[Rigorous computation of Stokes multipliers]{Rigorous numerical computation of the Stokes multipliers for linear differential equations\\ with single level one}
\author{\firstname{Michèle} \lastname{Loday-Richaud}}
\address{Université Paris-Saclay, CNRS, Laboratoire de mathématiques d'Orsay, 91405, Orsay, France.}
\email{michele.loday@universite-paris-saclay.fr}

\author{\firstname{Marc} \lastname{Mezzarobba}}
\address{
CNRS-LIX
1 rue Honoré d'Estienne d'Orves -
Bâtiment Alan Turing
CS35003 -
91120 Palaiseau, France}
\email{marc@mezzarobba.net}

\thanks{MM's work was supported in part by ANR grants ANR-20-CE48-0014-02 (NuSCAP) and ANR-22-CE48-0016 (NODE)}

\author{\firstname{Pascal} \lastname{Remy}}
\address{Laboratoire de Mathématiques de Versailles, UVSQ (Paris-Saclay) \& CNRS (UMR 8100), 45 avenue des Etats-Unis, 78035 Versailles cedex, France}
\email{pascal.remy@uvsq.fr}

\altkeywords{équations différentielles linéaires, multiplicateurs de Stokes, sommabilité, algorithmes, logiciels, calcul numérique rigoureux}

\usepackage{tikz}

\begin{document}

\begin{abstract}
  We describe a practical algorithm for computing the Stokes multipliers of a linear differential equation with polynomial coefficients at an irregular singular point of single level one. The algorithm follows a classical approach based on Borel summation and numerical ODE solving, but avoids a large amount of redundant work compared to a direct implementation. It applies to differential equations of arbitrary order, with no genericity assumption, and is suited to high-precision computations. In addition, we present an open-source implementation of this algorithm in the SageMath computer algebra system and illustrate its use with several examples. Our implementation supports arbitrary-precision computations and automatically provides rigorous error bounds. The article assumes minimal prior knowledge of the asymptotic theory of meromorphic differential equations and provides an elementary introduction to the linear Stokes phenomenon that may be of independent interest.
\end{abstract}

\maketitle

\section*{Introduction}\label{intro}

The aim of the present article is to describe a practical algorithm for the numerical computation of the \emph{Stokes matrices} of a linear differential equation~$Dy (x) = 0$ with coefficients in~$\mathbb{C} [x]$, under the assumption that~$D$ has a \emph{single level equal to one} (Definition~\ref{defpure1}, p.~\pageref{defpure1}) at the singular point of interest.

The Stokes phenomenon is a discontinuity phenomenon, long observed by physicists interested in numerical computations, in the asymptotic behaviour of analytic functions as the variable approaches a singularity from a varying direction. In the case of solutions of linear differential equations, the phenomenon originates in the occurrence in the asymptotic expansions of exponential factors that alternate between being flat in some directions and exponentially large in others. At the boundaries of a sector of flatness (\emph{Stokes directions}), the asymptotic expansion of a given solution can suddenly change albeit the solution itself shows no discontinuity. For the analysis of the Stokes phenomenon, though, it is more appropriate to look at the bisectors of the sectors of flatness, called \emph{anti-Stokes} or \emph{singular} directions. Understanding the Stokes phenomenon amounts to characterising the gaps between analytic solutions having the same asymptotic expansion on sectors that overlap around an anti-Stokes direction. Once appropriately normalised\footnote{In older literature, the name \emph{Stokes matrix} is used for matrices connecting arbitrary pairs of analytic fundamental solutions asymptotic to the same formal fundamental solution. Here, in accordance with modern usage, we reserve it for those matrices representing the components of the Stokes cocycle mentioned below.}, these gaps are described by a finite number of constant matrices called Stokes matrices.

Starting with Stokes' own work on the Airy equation and until the 1970's, the study of the Stokes phenomenon focused on the calculation of Stokes multipliers (nontrivial entries of the Stokes matrices) for differential equations with explicit solutions, typically given by integral formulae. This mainly concerns the family of hypergeometric equations, including the generalised hypergeometric equations, whose solutions are expressible in term of Barnes-Mellin integrals~\cite{DM89}. Such an approach is possible thanks to the fact that the Mellin transformation translates these equations into first-order difference equations.

From a theoretical perspective, the linear Stokes phenomenon was fully understood with the cohomological approach initiated by Sibuya and Malgrange \cite{Ma79}, which resulted in its characterisation by the Stokes cocycle~\cite{L-R94,L-R16,BV83}. In this line of work, the Stokes matrices are viewed primarily as the \emph{transcendental local invariants} in the classification of linear differential equations by meromorphic gauge transformations. The introduction of the Stokes cocycle makes it clear how to choose a non-redundant family of Stokes matrices---one for each anti-Stokes direction---in full generality. The cohomological framework also led to a new approach to the determination of Stokes multipliers, resulting essentially from the comparison between the isomorphism theorem of Malgrange-Sibuya and its infinitesimal version based on the Cauchy-Heine integral. The Stokes multipliers are then obtained as limits of rapidly converging linearly recurrent sequences. This method was only developed for linear equations with a single level, either of small order or satisfying genericity assumptions~\cite{JLP76II,MartR82,LS97,L-R90}. (See~\cite{LS97} for additional references.)

In parallel, the development in the 1980's of several theories of summability and of the theory of resurgence~\cite{Ecalle1981,Ecalle1981a,Ec85} provided another access to Stokes multipliers. In this approach each Stokes matrix results from the comparison of the sums of a given formal solution on both sides of an anti-Stokes direction. At a singular point of a single level one, the classic Borel-Laplace summation method suffices to obtain integral formulae for the sums that can be used to compute the Stokes multipliers. An explicit formula for the Stokes multipliers of a linear differential system with the single level~1 is given in~\cite[Thm.~4.3]{LR11} (see also~\cite{Guzzetti2016}). In the case of multiple levels, a more powerful theory of summation is required. Écalle's acceleration method~\cite{Ec85,Braaksma1991,Ecalle1992} was the first such theory amenable to effective computations. A variant that has also been used in practice is Balser's method of iterated Laplace integrals~\cite{Balser1992}.

Starting in the early 1990's, Thomann and several coauthors explored the application of summation methods in numerical computations with divergent solutions of linear differential equations. They showed how the Borel-Laplace method and its extensions can be implemented by combining techniques from symbolic computation with classical numerical methods and used to evaluate the sums of divergent solutions \cite{Thomann1990,Naegele1995,Thomann1995,JungNaegeleThomann1996,FauvetThomann2005} or compute the Stokes matrices~\cite{FJT09}. Their work is, to the best of our knowledge, the only one that went beyond an ad-hoc implementation dedicated to a specific example. However, due in part to the choice of numerical methods, it is mostly limited to machine-precision computations on equations of small order with sufficiently generic Stokes values and local exponents. In addition, unfortunately, their software is not publicly available.

Van~der~Hoeven~\cite{vdH2007,vdH16} gave a complete algorithm for computing sums of solutions of linear differential equations with polynomial coefficients, based on Écalle's acceleration method. His algorithm is optimised for high-precision computations; in particular, it can compute the Stokes multipliers of a fixed equation with an error of at most~$2^{- p}$ in time $\mathrm{O} (p \log^{3 + \omicron (1)} p)$. In addition, it yields rigorous bounds on the difference between the numerical result and the mathematical value\footnote{This is in contrast with typical numerical methods, which guarantee at best that the computed result converges to the true value at a certain rate when parameters such as a step size or a numerical working precision tend to zero or infinity. Nevertheless, most of the other algorithms discussed above could in principle be adapted to produce rigorous error bounds as well. Richard-Jung~\cite{Richard-Jung2011} presented an example of the computation of Stokes multipliers with rigorous error bounds.}. This algorithm is the most complete algorithmic treatment to date of the summation of divergent solutions of linear differential equations, but it has not been fully implemented\footnote{A~prototype implementation limited to the case of singular points of a single level~$1$ and due to the second author is available as part of the \codestar{ore\_algebra} package mentioned below.}.

In the present work, we focus on the computation of Stokes multipliers, as opposed to the more general problem of evaluating the sums, and we restrict ourselves to singular points with a single level equal to one. In this setting, we give a detailed description and a complete implementation of an algorithm for computing the Stokes matrices of a linear differential equation with coefficients in~$\bar{\mathbb{Q}} [x]$. There is no additional restriction; in particular, the algorithm applies to equations of arbitrary order and makes no assumption on the geometry of Stokes values or the shape of indicial polynomials.

The method is very similar to that of Fauvet, Richard-Jung and Tho\-mann~\cite{FJT09}, but our algorithm is more complete in that it can handle arbitrary degeneracies with no user intervention, and produces rigorous error bounds. It differs from van der Hoeven's method in that the numerical computation of integral transforms of holonomic functions is replaced by expressions in terms of special functions. This leads to a marginally better computational complexity estimate at high precision for the special case we consider. A more important difference in practice, compared to all existing approaches, is that we pay attention to avoiding redundancies that occur when computing several Stokes multipliers of the same equation.

Like virtually all algorithms for the problem of computing Stokes multipliers, the one we develop mixes exact and approximate computations. Roughly speaking, the geometry of Stokes values and the structure of local exponents at singular points in the Borel plane (see Sec.\ \ref{BorelSol}~and~\ref{sec:Stokes-Borel}) must be determined exactly because one needs to decide when Stokes values are aligned or local exponents differ by integers. The rest of the computation is done numerically but, as already mentioned, with rigorous error bounds. While, in principle, one could labor to determine a priori how accurately to perform each step of the computation in order for the final result to satisfy a given error tolerance, it is difficult to obtain sharp error bounds this way while keeping the calculations manageable. Instead, and following standard practice, we content ourselves with keeping track of the accumulated error over the course of the computation and returning an error bound along with the result. This also allows us to refer to existing work~\cite{Johansson2023,Mezzarobba2019} for bounding the error resulting from the key analytic steps of the algorithm (namely numerical integration of differential equations and the evaluation of special functions). In case the error bound is insufficient for the application, it is always possible to run the algorithm again at higher precision, and doing so eventually produces arbitrarily precise results.

Returning error bounds increases confidence in the computed results but also has algorithmic applications based on differential Galois theory, via the Ramis density theorem.
For instance, checking that one of the Stokes multipliers of one of the variational equations of a Hamiltonian system is nontrivial suffices to prove that the system is not completely integrable with meromorphic first integrals~\cite{Morales-RuizRamis2010,BoucherWeil2007,Ramis2024}.
Criteria of this kind are available for a number of other problems such as certifying that the differential operator~$D$ is irreducible or that it is the minimal-order operator annihilating a certain function.
Importantly, though, such local criteria typically translate into sufficient conditions only.
Full characterizations---of irreducibility, for instance---in the same style, and complete algorithms based on these characterizations~\cite{vdH2007,Goyer2025}, are possible given a representation of the \emph{global} differential Galois group through monodromy and Stokes automorphisms at all singular points \emph{expressed in the same basis}.
Moving the Stokes automorphisms computed by our algorithm (which are expressed in a basis that depends on the singular point) to the same point is a connection problem that effective summation methods can serve to solve, and thus our algorithm, even when it applies, does not eliminate the need for such methods.

We have implemented our algorithm using the SageMath computer algebra system. Our implementation is part of the \codestar{ore\_algebra} package, available at \codestar{\href{https://www.github.com/mkauers/ore\_algebra/}{https://www.github.com/mkauers/ore\_algebra/}}. It heavily relies on the \textsc{Flint} library\footnote{See \href{https://flintlib.org}{https://flintlib.org/}. More specifically, we mainly use \textsc{Flint}'s modules for real and complex ball arithmetic, formerly distributed as a separate library called \textsc{Arb}~\cite{Johansson2017}.}, and on pre-existing code by the second author for the computation of connection matrices between regular singular points~\cite{Mezzarobba2016}.

\paragraph*{Outline}

This article is organised as follows.

In Part~\ref{part:theory}, we develop a first, `theoretical' algorithm that reduces the numerical computation of Stokes matrices to the numerical solution of connection problems between regular singular points. The algorithm amounts to comparing the 1-sums of a formal fundamental solution to on either side of each singular direction. For the convenience of the non-expert reader, and also because we could not find a fully adequate reference, we recall in that part most of the necessary background, including the notion of equation of single level~1, the Borel--Laplace summation method, the definition of Stokes matrices based on lateral sums, and the expression of the Stokes multipliers in terms of the connection constants of the Borel transform of the differential operator~$D$.

The regular singular connection problems we are left with are much easier than the general ones alluded to above. Solving them numerically reduces to computing the analytic continuation of solutions of differential equations given by convergent series expansions at singular points, or in other words to solving generalised differential initial value problems. Even so, it is the most computationally expensive stage of the method in practice.
In Part~\ref{part:algorithm} of the article, we present a more practical variant of the algorithm from Part~\ref{part:theory}, designed to minimize the use of numerical analytic continuation and be relatively easy to implement. We give a detailed description of the algorithm in terms of operations readily available in today's computer algebra libraries. We also introduce our implementation and illustrate its use on several examples.

\paragraph*{Acknowledgments}

We are grateful to Bruno Salvy for enabling this collaboration by initiating the contact between us.
We also thank him, as well as Fredrik Johansson, Jean-Pierre Ramis, and Anne Vaugon, for valuable comments and discussions.

\part{Underlying theory}
\label{part:theory}

All along this paper we consider a linear differential equation
\begin{equation}\label{Eqinitiale}
 Dy(x)=0
\end{equation}
with polynomial coefficients in the complex variable~$x$, which we study at an irregular singular point that we position at~$x=0$.
For reasons that will become clear later (\cf Sec.~\ref{eqBorel}, p.~\pageref{eqBorel}), we write the operator~$D$ in terms of the derivation~$\partial=x^2 \rmd/\rmd x$ and, dividing by a convenient power of~$x$ if necessary, we assume that~$D$ is of the form
\begin{equation}\label{Opinitial}
D\equiv  D\Big(\frac{1}{x}, \partial\Big) = a_n\Big(\frac{1}{x}\Big) \partial^n +a_{n-1}\Big(\frac{1}{x}\Big) \partial^{n-1} + \cdots + a_1\Big(\frac{1}{x}\Big) \partial +a_0\Big(\frac{1}{x}\Big)
\end{equation}
where the coefficients
\begin{equation}\label{coef}
 a_\ell(X) = \sum_{j=0}^\nu a_{\ell,j}\, X^j
\end{equation}
for~${\ell= 0,\dots, n}$ are polynomials of degree~$\leq \nu$.
We assume that $D$~has order~$n$, that is,
$a_n(X)\not\equiv 0$,
and at least one of the coefficients~$a_\ell$ has degree~$\nu$.

Starting from Sec.~\ref{preparedsol} we additionally assume the restrictive condition that the singular point at the origin is of single level one (\cf Definition~\ref{defpure1}, p.~\pageref{defpure1}, Hypothesis~\ref{hyp2}, p.~\pageref{hyp2}).

\section{The equation in the Laplace plane} \label{eqLaplace} 

The plane~$\C$ with coordinate~$x$ where  the initial equation~(\ref{Eqinitiale}) is considered is commonly called  the \emph{Laplace plane}, recalling thus that it is the image by a Laplace transformation of the \emph{Borel plane} to be considered in the next section.
It is well-known that the equation $D\,y = 0$ admits $n$~linearly independent formal solutions of the form
\begin{equation} \label{eq:individual-sol}
  \Bigl( \sum_{j=0}^{s-1} \tilde g_{i,j}(x) \ln^j x \Bigr) x^{\mathcal L_i} \rme^{q_i(1/x)}.
\end{equation}
Here the term \emph{formal} refers to the fact that the series~$\tg_{i,j}(x)$ may be divergent.
Such a family is collectively called a \emph{formal fundamental solution} and one can choose it in the form
\begin{equation}\label{ffs}
\begin{bmatrix}
 \tf_1(x) &\tf_2(x)& \cdots & \tf_n(x)
\end{bmatrix}
x^{\frak{L}} \, \rme^{Q(1/x)}
\end{equation}
where the~$\hspace{-.5mm}\tf_j(x)$'s are (usually divergent) power series of~$x$, $\frak{L}$ is a {constant}~ma\-trix in Jordan form and~$Q(1/x)=\diag(q_1(1/x),\dots,q_n(1/x))$ with~${q_j(1/x)}$'s that are polynomials in~$1/x$ or in  fractional powers of~$1/x$ with no constant term. These polynomials are called the \emph{determining polynomials of}~$D$. The matrix~$\frak{L}$ is called the matrix of \emph{exponents of formal monodromy}.

\subsection{Newton polygons and single level~1} \label{NewtonLevel1}

The Newton polygon of the operator~$D$ at~0 is defined as follows  \cite[Sec.~3.3.3.1]{L-R16}: one defines the weight~$w$ of a monomial by~$w(x^{-j} \partial^\ell)= \ell-j$ and the weight of~$D$ as the smallest weight of a non-zero monomial in~$D$. With these non-zero monomials one associates marked points of coordinates~${(\ell, \ell-j)}$ in the half-plane~$\R_{\geq 0} \times \R$ and second quadrants based at these marked points. The Newton polygon~$\sN(D)$ of~$D$ is the convex hull of this family of quadrants. The Newton polygon may contain a horizontal side, one or more sides with positive slopes, and always ends with a vertical line of abscissa~$n$ which we do not consider as a side. One can prove that the slopes of the sides (including~$0$) are the degrees of the determining polynomials (including the trivial polynomial~$q\equiv0$). Moreover, if the side of slope~$p$ has a horizontal length equal to~$r$ then there are~$r$ determining polynomials of degree~$p$.

We now explain how to calculate the exponents (\ie the diagonal terms of~$\frak{L}$) and the determining polynomials~$q_j$.
For simplicity---since we are interested in this case only---\emph{we suppose from now  that the operator is  of level~one}, which means that the largest slope of its Newton polygon is equal to~1.
This is a weaker condition than being of single level~1 (Hypothesis~\ref{hyp2}, p.~\pageref{hyp2}). One sees on the Newton polygon~$\sN(D)$ that $D$ is of  level~1 if and only if there exists~$j, 0\leq j<n$ such that
$a_{n,\nu} \, a_{j,\nu} \neq 0$.

Like for differential equations with constant coefficients the exponents of the solutions \eqref{eq:individual-sol} without exponential factor are the roots of an algebraic equation called the \emph{indicial equation} and built as follows. Suppose that the horizontal side of~$\sN(D)$ has length~$k_0$ and lies at ordinate~$h_0$; then, the operator restricted to this horizontal side reads~${ \sum_{\ell=0}^{k_0} a_{\ell, \ell-h_0} (1/x)^{\ell-h_0} \partial^\ell  }$. The indicial equation is obtained by writing that the monomial~$x^\sL$ is a solution of this restricted operator and it reads
\begin{equation}\label{indicialeq}
\sP(\sL)\equiv \sum_{\ell=0}^{k_0} a_{\ell, \ell-h_0} [\sL]_{\ell^+} =0
\end{equation}
with the following notation.

\begin{nota}\label{notalambdaj}
 For~$\gl\in\C$ and~$j\in\Z_{>0}$ we denote
\begin{equation*}
\begin{aligned}
 [\gl]_{j^-} =& \gl(\gl-1)\cdots(\gl-j+1) \quad \text{($j$ factors decreasing by~1 from $\gl$)};\\
 [\gl]_{j^+} =& \gl(\gl+1)\cdots(\gl+j-1)  \quad \text{($j$ factors increasing by~1 from $\gl$)}.\\
\end{aligned}
\end{equation*}
\end{nota}

Saying that the operator is of level~$1$ is equivalent to saying that the highest degree of its determining polynomials is equal to~$1$. However, even determining polynomials of degree~$1$ might contain fractional powers of~$1/x$. If the horizontal length  of the side of slope~1 is~$k_1$, there are~$k_1$ determining polynomials of the form~$q_j(1/x) =-\ga_j/x +o(1/x)$, the other~$n-k_1$ having no term of degree~1 (\ie $\ga_j=0$).
The coefficients~$\ga_j$, zero or not, are the (possibly multiple) roots of the \emph{characteristic polynomial} associated with the side of slope~1 and obtained as follows. One considers the operator~$D$ restricted to  the side of slope~1, that is,  the operator~${\sum_{j=0}^{k_1} a_{n-j, r} (1/x)^{r} \, \partial^{n-j}}$ for a suitable~$r$. When the side of slope~1 of the Newton polygon~$\sN(D)$ is on the first bisector one has~$r=0$, and in general $\sN(x^rD)$ has a side of slope~1 on the first bisector of the axes.
The \emph{characteristic equation} (of level~1) then reads
\begin{equation}\label{charpol}
 {P(X)\equiv\sum_{j=0}^{k_1} a_{n-j, r} X^{n-j}=0}.
\end{equation}
It admits~$X=0$ as a root of multiplicity~$n-k_1$. The roots of~$P(X)$ are called~\emph{Stokes values} (of level 1) of~$D$; the non-zero ones are called \emph{characteristic roots}.
In what follows, when we refer to `the' characteristic polynomial, characteristic roots, or Stokes values of~$D$, we always mean the characteristic polynomial of level~$1$ and its roots.

\begin{dfn}\label{defpure1}
The operator~$D$ is said to be \emph{of single level~1} (or of \emph{pure} level~1) if all its nonzero determining polynomials~$q_j(1/x)$ are monomials of degree~1 in~$1/x$.
\end{dfn}

Equivalently, $D$ is of single level~$1$ if \emph{all non-zero differences~$q_\ell-q_j$ among the determining polynomials of~$D$ are  of degree~$1$}. Indeed, suppose that there is a determining polynomial~${q(1/x)=\ga/x+\gb/{x^{r/s}}+\cdots}$ with~$r/s\in\Q$, $0<r/s<1$ and~$\gb\neq 0$. Since the change of variable~${x \leftarrow x\rme^{2\pi \rmi}}$ leaves the operator~$D$ invariant, there must also be a determining polynomial of the form ${\tq(1/x)=\ga/x+ \gb\rme^{-2r\pi\rmi/s}/{x^{r/s}}+\cdots}\cdot$ The difference~$q(1/x)-\tq(1/x)$ is a polynomial of degree~$r/s$ with~$0<r/s<1$, hence the operator~$D$ is not of single level. The converse assertion is straightforward.

As the terminology suggests, if $D$~is of single level~$1$, its Newton polygon~$\sN(D)$ has a unique side of slope~$1$ in addition to a possible horizontal side, however this condition is not sufficient.

\subsection{The operators~$D_{[\ga]}$}

To characterise operators of single level~1 it is useful to introduce the operators~$D_{[\ga]}$ deduced from~$D$ by the change of variable~$y=\rme^{-\ga/x} \, z$ with~$\ga$ a Stokes value (of level~1). The operator~$D_{[\ga]}$ has the elementary following properties.\label{propertiesDalpha}

\begin{enumerate}
 \item \label{item:shift-dop} One has $D_{[\ga]}(1/x, \partial) =D(1/x, \partial+\ga)$, and hence
  \begin{equation*}
\begin{aligned}
 D_{[\ga]}& = \textstyle a_n \big(\frac{1}{x}\big)(\partial+\ga)^n+a_{n-1}\big(\frac{1}{x}\big)(\partial+\ga)^{n-1}+\cdots+ a_0\big(\frac{1}{x}\big) \\
&= \textstyle A_n\big(\frac{1}{x}\big)\partial^n + A_{n-1}\big(\frac{1}{x}\big)\partial^{n-1} +\cdots +  A_1\big(\frac{1}{x}\big) \partial + A_0\big(\frac{1}{x}\big)\\
\end{aligned}
\end{equation*}
where, for $\ell=0,1,\dots, n$ and denoting by~$ \binom{j}{\ell}$
 the binomial coefficients,  the new coefficients~${A_\ell \big(\frac{1}{x}\big)= \sum_{j=0}^\nu A_{\ell,j} \frac{1}{x^j}}$ read

 \begin{equation}\label{Aell}
 \begin{split}
 \textstyle
 A_\ell \big(\frac{1}{x}\big)= a_n \big(\frac{1}{x}\big)  \binom{n}{\ell}\ga^{n-\ell} + a_{n-1}\big(\frac{1}{x}\big)  \binom{n-1}{\ell}  \ga^{n-1-\ell} +\cdots\\
 \textstyle
 \dots+ a_{\ell+1}\big(\frac{1}{x}\big)  \binom{\ell+1}{\ell} \ga^1+ a_\ell\big(\frac{1}{x}\big)  \binom{\ell}{\ell}.
 \end{split}
 \end{equation}
 In particular, $A_n=a_n$.

 \smallskip
\item For~$\ell=0,\dots, n$,  one has 
$
 A_{\ell,\nu} = \sum_{j=\ell}^n a_{j,\nu}  \binom{j}{\ell} \ga^{j-\ell} =  P^{(\ell)}(\ga)\,/\,(\ell\,!)
$
where~$P^{(\ell)}$ is the~$\ell^{th}$ derivative of the characteristic polynomial~$P$ of~$D$ (\cf equation~\eqref{charpol}, p.~\pageref{charpol}). 
\smallskip
\item \label{item:mult}
If~$\ga$ has multiplicity~$k$ as a Stokes value, then~$A_{k,\nu}\neq 0$.
\smallskip
\item
 Therefore, $D_{[\ga]}$ has order~$n$ and degree~$\nu$.
\smallskip
\item If~$D$ is of pure level~1 so is~$D_{[\ga]}$ for any~$\ga$.

 \smallskip
 \item If~$P(X)$ is the characteristic polynomial  of~$D$ then that of $D_{[\ga]}$ is
 $$P_{[\ga]}(X)= P(\ga+X).$$ And indeed, if the polynomials~$q_j(1/x)$'s are the determining polynomials of~$D$, the  determining polynomials of~$D_{[\ga]}$ are the polynomials~$q_j(1/x)+\ga/x$.
\smallskip
\item The exponents~$\sL$ of the solutions of~$D$ with exponential part~$\rme^{-\ga/x}$ are the solutions of the indicial equation (\cf \eqref{indicialeq}) of~$D_{[\ga]}$.
\end{enumerate}

\begin{pro}\label{CNSsingle}
 Let~$D$ be an operator of order~$n$ and level~1 at~0. We exclude the trivial case\footnote{With a unique characteristic root~$\ga$ of multiplicity~$n$ the change of variable~${y=\rme^{-\ga/x}z}$ changes~$D$ into an operator with a regular singular point at~0 if~$D$ is of single level~1 or an operator of level~$<1$ if~$D$ is of level~$1$ but not of single level~1.}
 where~$D$ has a unique characteristic root of multiplicity~$n$.
 The operator~$D$ is of pure level~1 if and only if the following equivalent properties hold:
\begin{enumerate}
 \item \label{item:CNSsingle:1} For any Stokes value~$\ga$ $($including~$\ga=0)$ of~$D$  of multiplicity~$k$, the Newton polygon~$\sN(D_{[\ga]})$ has a horizontal side of length~$k$ and a side of slope~$1$ with horizontal length~$n-k$. After a vertical translation of height~$\nu$ it looks as depicted on Figure~\ref{fig-dessin1bis} below.
 \begin{figure}[!ht]
 \begin{center}
 \includegraphics[width=0.5\textwidth]{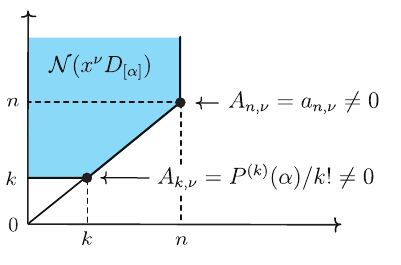}
 \vspace{-5mm}\caption{Newton polygon of $x^\nu D_{[\ga]}$\quad}
 \label{fig-dessin1bis}
 \end{center}
 \end{figure}

 \item \label{item:CNSsingle:2} For any Stokes value~$\ga$ $($including~$\ga=0)$ of~$D$  of multiplicity~$k$, the coefficients of~$D_{[\ga]}$ satisfy the relations
 \begin{equation*}
\begin{array}{l}
 A_{\ell,j}=0 \quad\textrm{for all}\quad \ell, j \quad\textrm{satisfying}\quad    \ell -  j<k -\nu, \\
 \noalign{\smallskip}
 A_{n,\nu}\, A_{k,\nu}\neq 0. \\
\end{array}
\end{equation*}
\end{enumerate}
\end{pro}

\begin{proof}
(\ref{item:CNSsingle:1}) The condition for~$\ga=0$ means that the Newton polygon~$\sN(D)$ of~$D$ has no slope other than $0$~and~$1$. Hence, the determining polynomials of~$D$ are all of degree~1. We have to prove that they  are monomials. Reasoning by contradiction, suppose one of these determining polynomials reads~$q(1/x)=-\ga/x +\gb/x^r+\cdots$ with~$\gb\neq 0$    and~$r<1$ (necessarily rational).
Then, the polynomial~$q(1/x)+\ga/x=\gb/x^r+\cdots$ is a determining polynomial of~$D_{[\ga]}$ of degree~$r$.
This implies that~$\sN(D_{[\ga]})$ has a  side with slope~$r<1$. Hence  a contradiction with  assumption~{(\ref{item:CNSsingle:1})} asserting that the only possible slopes for~$\sN(D_{[\ga]})$ are $0$~and~$1$. Thus the operator~$D$ is of single level~1.  The converse assertion is straightforward.

(\ref{item:CNSsingle:2}) The conditions translate the form of each Newton polygon~$\sN(D_{[\ga]})$ on the coefficients of~$D_{[\ga]}$.
\end{proof}

\begin{hyp}\label{hyp2}
 From here on, we assume that~$D$ \rm{is of single level~1}. 
\end{hyp}

\subsection{Prepared fundamental solution}\label{preparedsol} 

In full generality, the matrices~$\frak{L}$ and~$Q$ appearing in the formal fundamental solution~\eqref{ffs} of ${D\,y=0}$ do not commute \cite{BJL79}, \cite[Cor.~2.5]{L-R01}.
Under Hypothesis~\ref{hyp2}, however, the determining polynomials~$q_j$ are unramified, so that the matrices $\mathfrak{L}$~and~$Q$ commute.
In this case one can choose the formal fundamental solution in the form
\begin{equation}\label{fundsol}
\begin{aligned}
Y(x) &=
\begin{bmatrix}
\ty_1(x)&\ty_2(x)&\cdots&\ty_n(x)\\
\end{bmatrix} \\
& =  \begin{bmatrix} \tf_1(x) & \tf_2(x)&\cdots &\tf_n(x)\end{bmatrix} \, x^{\pL}\, \rme^{Q(1/x)} 
\end{aligned}
\end{equation}
 where
\begin{itemize}
 \item $Q(1/x) =\diag\big(q_1(1/x),\,  \dots,\, q_n(1/x)\big)$ with determining polynomials~$q_j$ of the form $q_j(1/x)=-\ga_j/x$;
 
\item the matrix $\pL$ of the exponents of formal monodromy is a constant matrix in Jordan form commuting with~$Q$; 
 
\item  the~${\tf_j(x)}$'s are power series in~$x$;

\item the regular part~$\tth(x)=\begin{bmatrix} \tf_1(x) & \tf_2(x)&\cdots &\tf_n(x)\end{bmatrix} \, x^{\pL}$ has valuation~$>0$, in the sense that for all monomials~$x^\beta \log^\rho x$ appearing in the entries of~$\tth(x)$ one has $\Re(\beta) > 0$.
\end{itemize}

The general form is classical and we stick to it for the theoretical discussion that follows.
Note, though, that in the second part of the article we will switch to a slightly different choice of fundamental solution (see Rem.~\ref{remprepared} and Sec.~\ref{sec:local-bases}) for consistency with the implementation.

The valuation condition can always be satisfied with a change of variable of type~$y=x^{-r} \, z$ for a convenient integer~$r$.  Such a transformation changes equation~(\ref{Eqinitiale}), p.~\pageref{Eqinitiale}, into an equation of the same form and one can check that it preserves the Stokes matrices we want to compute (\cf \eqref{StokesFormula1}, p.~\pageref{StokesFormula1}). 

In terms of the classification of differential equations, the matrices $\mathfrak{L}$~and~$Q$ provide a complete set of \emph{formal invariants}. These invariants  are obtained by solving algebraic equations and they determine the anti-Stokes directions (Definition~\ref{aSd}, p.~\pageref{aSd}).
In contrast, the Stokes phenomenon (\emph{i.e.}, the \emph{meromorphic invariants}) is determined from  the divergent series~${\tf_j(x)}$ and, except in very simple cases, it results from transcendental calculations.

We additionally suppose that the solutions are so arranged  that the matrix~$Q(1/x)$ splits into diagonal  blocks
\begin{equation*}
  {Q(1/x) = q_1(1/x) I_{k_1} \oplus q_2(1/x) I_{k_2} \oplus \cdots \oplus q_N(1/x) I_{k_N}}
\end{equation*}
with pairwise distinct $q_\ell$'s and $q_\ell=-\ga_\ell/x$.
We denote by
\begin{equation} \label{blocQ}
\left\lbrace
\begin{array}{ccl}
\sI_1&=&\{1,2,\dots k_1\}\\
\sI_2&=&\{k_1+1, k_1+2,\dots,k_1+k_2\}\\
  \cdots& &\\
 \sI_N&=&\ensemble{k_1+\cdots+k_{N-1}+1, \ \dots \ , k_1+k_2+\dots +k_N=n}\\ 
\end{array}
\right.
\end{equation}
the sets of row indices corresponding to each block of~$Q(1/x)$:
the solutions with index~$i\in\sI_\ell$ are those with exponential part~$\rme^{q_\ell}$ (one also says that they are associated with~$q_\ell$ or with~$\ga_\ell$).
We further denote
\begin{equation*}
 {\sK_0=0,\ \ \sK_1=k_1, \ \ \sK_2=\sK_1+k_2, \ \ \dots, \ \ \sK_N=\sK_{N-1}+k_N=n}
\end{equation*}
and split the fundamental solution~$Y(x)$ according to the splitting of~$Q$, obtaining thus the block decomposition
\begin{equation*}
 Y(x) = 
\begin{bmatrix}
 Y_{i}(x)
\end{bmatrix}_{1\leq i\leq N}
\quad \textrm{with}\quad Y_{i}(x)= 
\begin{bmatrix}
\ty_{ u}(x)
\end{bmatrix}_{u\in \sI_i}=
\begin{bmatrix}
\ty_{ u}(x)
\end{bmatrix}_{\sK_{i-1}+1\leq u\leq \sK_i}.
\end{equation*}
We can now arrange the solutions with same determining polynomial~$q$
so that  the matrix~$\pL$ splits into a sum of Jordan blocks (\cf Notation~\ref{Jordan} below)
\begin{equation*}
  {\pL=(\sL_1 I_{m_1} +J_{m_1}) \oplus (\sL_2 I_{m_2} +J_{m_2}) \oplus\cdots
  \oplus (\sL_{m_M} I_{m_M}+ J_{m_M}})
\end{equation*}
which refines the splitting of~$Q$ into blocks. 

\begin{nota}\label{Jordan}
 We denote by~$I_r$ the identity matrix of dimension~$r$ and by~$J_r$ the upper nilpotent Jordan matrix of dimension~$r$ (with a super-diagonal of ones and zeroes elsewhere).
\end{nota}

\begin{rem}[generic case] \label{gencase}
 Quite often the situation is not as general as the~one considered above. A common case---classically called the generic case---is that all Stokes values~$\ga$ are simple ($k=1$). Then, the matrix~$Q$ splits into diagonal blocks of size~1:
 \[ {Q(1/x)=[q_1(1/x)]\oplus [q_2(1/x)] \oplus \cdots \oplus[q_n(1/x)]}, \]
 the matrix~$\pL$ is diagonal, and hence the formal solutions in the expanded form~\eqref{eq:individual-sol} involve no logarithms.
The characteristic equation of~$D_{[\ga]}$ reads (\cf equation~\eqref{charpol}, p.~\pageref{charpol})
\begin{equation*}
P_{[\ga]}(X)=P(\ga+X)= A_{n,\nu}X^n +A_{n-1,\nu}X^{n-1} +\cdots + A_{1,\nu}X^1
\end{equation*}
and the conditions in Proposition~\ref{CNSsingle} reduce to~$A_{0,\nu}=0$ and~${A_{1,\nu} A_{n,\nu}\neq 0}$.
The indicial equation associated with the Stokes value~$\ga$ reads
\[ \sP(\sL)\equiv  A_{0,\nu-1} + A_{1,\nu} \sL= 0, \]
corresponding to a unique exponent
$\sL=-{A_{0,\nu-1}}/{A_{1,\nu}}$.
\end{rem}

A fundamental property for the algorithm to be developed below is that under Hypothesis~\ref{hyp2}
the series~$\tf_j(x)$ appearing in the fundamental solution~$(\ref{fundsol})$ are all {1-\emph{summable}}~\cite[Thm.~5.2.5]{L-R16}.
One obtains their sums by applying the so-called Borel--Laplace method, as we discuss in the next two sections.

\section{The equation in the Borel plane} \label{eqBorel} 

The Borel plane is a second copy of~$\C$, with coordinate~$\xi$, to which the equation~\eqref{Eqinitiale} is  sent by means of a Borel transformation. A Laplace transformation sends the Borel plane back to the Laplace plane.

As much as possible we use Latin letters in the Laplace plane and Greek letters in the Borel plane.

\subsection{Borel transformations}\label{PrelimBorel}

We will consider two kinds of Borel transformations:

--- The \emph{Borel transformation $\sB$ for operators} is the morphism from the ring of linear differential operators~${\sum_{j,\ell}c_{\ell,j}\frac{1}{x^j}\, \partial^\ell}$ `in the Laplace plane' to
that of differential operators~${\sum_{j,\ell} \gc_{\ell,j} \xi^j \frac{\rmd}{\rmd \xi^\ell}}$ `in the Borel plane' defined by
\begin{equation} \label{BorelOper}
 \sB\Big(\frac{1}{x}\Big)= \frac{\rmd}{\rmd \xi} \quad\textrm{and}\quad \sB(\partial)= \xi \quad \textrm{(operator of multiplication by } \xi\textrm{) } 
\end{equation}
Observe that an operator $D=D\big(\frac{1}{x}, \partial\big)$ with polynomial coefficients in~$1/x$
is changed into an operator~${\gD=D( \frac{\rmd}{\rmd \xi}, \xi)}$ with polynomial coefficients in~$\xi$.
This property is the main reason why we chose to write our operators in terms of $1/x$~and~$\partial$ (positive powers of~$x$ would lead to integro-differential operators).
If~$D$ has order~$n$ and degree~$\nu$ in~$1/x$ then~$\gD$ has order~$\nu$ and degree~$n$ in~$\xi$.

--- The \emph{(functional) Borel transformation}~$\sB_\gt$  depends on a direction~$\gt$ and applies to  actual functions. Written at~0 (as opposed to infinity) it reads 
\begin{equation}\label{anaBorel}
 \sB_\gt(f(x))(\xi) = \frac{1}{2 \pi \rmi} \int_{\gc_\gt} f(x) \rme^{\xi/x} \frac{\rmd x}{x^2}
\end{equation}
where the contour~$\gc_\gt$ is a `Hankel contour at~$0$', \cf Figure~\ref{fig-dessin-Hankel}.
Equivalently, one has $\sB_\gt\big(f(x)\big)(\xi) = \frac{1}{2\pi\rmi} \int_{\gC_{(-\gt+\pi)}}
f(1/u)\rme^{\xi u} \mathrm du$
where the integral is now along the classical Hankel contour $\gC_{ (- \gt+\pi)}$ around the half-line~$d_{(-\gt + \pi)}$ from $\arg u=-\gt-\pi$ to $\arg u = -\gt+\pi$.\\
We require that the integral converge in the Lebesgue sense.
\begin{figure}[th]
\begin{center}
\includegraphics[width=0.2\textwidth]{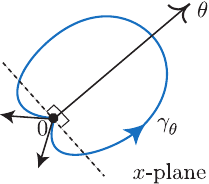}
\vspace{-5mm}\caption{A Hankel contour at~0}
\label{fig-dessin-Hankel}
\end{center}
\end{figure}

The following elementary lemmas are useful in the theory that follows.
For more details on Borel transformations we refer to~\cite[Chap.~5]{MSau16},~\cite[Sec.~5.3]{L-R16} and  \cite{Cost09} among many references.

\begin{lemma} \label{dico}
  For any~$\gl\in\C$ and $\xi \in \C \setminus \{0\}$, the Borel transform of monomials~$x^\lambda \log^m x$ is given by the following formulae in which the Borel integrals converge on the open half-plane~$\Re( \xi \rme^{-\rmi \gt} )>0$ bisected by~$d_\gt$.
 \begin{enumerate} 
 \item\label{item:dico:simple} If~$\gl\in \C \setminus \Z_{\leq 0}$, then  $\displaystyle \sB_\gt(x^\gl)(\xi) =\frac{ \xi^{\gl-1}}{\gC(\gl)}\not\equiv 0$ for any direction~$\gt\in\C$.
 If~$\gl\in\Z_{\leq 0}$ then%
\footnote{However, the definition of the Borel transform is classically extended by introducing Dirac masses in this case, see Remark~\ref{Dirac} below.}
 $\displaystyle \sB_\gt(x^\gl)(\xi) =\frac{ \xi^{\gl-1}}{\gC(\gl)}\equiv 0$.

 \item\label{item:dico:main}
 For  any ${m\in\Z_{\geq 1}}$ and any direction~$\gt$, one has
\begin{equation} \label{eq:dico-main}
  \sB_\gt \big(x^\gl\,\ln^{m} x\big)(\xi) = \frac{\rmd^{m}}{\rmd \gl^{m}} \sB_\gt(x^{\gl})(\xi).
\end{equation}
 
 \item\label{item:dico:special}
 The case $m=1$ is worth stating explicitly.
 For any direction~$\gt$, one has
\begin{equation*}
 \displaystyle \sB_\gt (x^\gl\,\ln x)(\xi) =
\xi^{\gl-1}  \Big( \frac{\ln\xi}{\gC(\gl)} + \Big(\frac{1}{\gC}\Big)'(\lambda) \Big).
\end{equation*}
In particular, when~$\gl=p$ is a positive integer the formula reads
\[
\sB_\gt(x^p\ln x)(\xi) = \frac{ \xi^{p-1}}{(p-1)!} \Big( \ln \xi  +\gc-\sH_{p-1} \Big)
\]
 where~$\gc$ is the Euler constant and $\sH_n = 1 + \dots + 1/n$ is the harmonic sum of order~$n$.
 For nonpositive integer~$\lambda = -p$ one has
 \[ \sB_\gt\Big(\frac{1}{x^p} \ln x\Big) =(-1)^p \frac{p!}{\xi^{p+1}}, \]
and in particular $\displaystyle \sB_\gt(\ln x)(\xi)= \frac{1}{\xi}$.
\end{enumerate}
\end{lemma}

\begin{proof}
(\ref{item:dico:simple}) By definition, $2\pi\hspace{-1pt} \rmi \sB_\gt(x^\gl)(\xi) \hspace{-1pt} =\int_{\gc_\gt} x^\gl \rme^{\xi/x} \frac{\rmd x}{x^2}$.

Setting~${u=1/x}$ and denoting as before by $\gC_{(-\gt +\pi)}$ the Hankel contour around the half-line~$d_{(-\gt + \pi)}$ run  from~$\arg u=-\gt - \pi$ to~$\arg u=-\gt +\pi$ we obtain
\begin{equation*}
\sB_\gt(x^\gl)(\xi) \hspace{-1pt} = \frac{1}{ {2\pi\hspace{-1pt} \rmi }}\int_{\gC_{(-\gt +\pi)}} u^{-\gl}\rme^{\xi u}\rmd u.
\end{equation*}
 This expression is valid for all~$\xi$ in the open half-plane~${\Re( \xi \rme^{-\rmi\gt})\hspace{-1pt} >\hspace{-2pt}0}$, that is, the open  half-plane bisected by~$d_{\gt}$. 

Suppose first that~$\gt=0$. The contour is then~$\gC_\pi$ run from $\arg(u)=-\pi$ to~$\arg(u)=+\pi$. The function  $u^{-\gl}$ is taken with its principal determination~$u^\gl=\rme^{\gl \ln u}$ on $-\pi<\arg (u)<\pi$ and, for~$\xi>0$ (\ie $\arg(\xi)=0$), we can apply~\cite[(4.8.1), p.~296]{Dieu} to obtain  the relation 
\begin{equation*}
\sB_0(x^\gl)(\xi)= \frac{1}{2\pi\rmi}\int_{\gC_{\pi}} u^{-\gl} 
\rme^{\xi u}\rmd u = \frac{{\xi^{\gl-1}}}{2\pi\rmi}
\int_{\gC_{\pi}} v^{-\gl} \rme^v \rmd v= 
{\xi^{\gl-1}}\frac{1}{\gC(\gl)}\cdot
\end{equation*}
If $\arg(\xi)=\gt'\in (-\pi/2, \pi/2)$ then 
$\int_{\gC_{\pi+\gt'}} v^{-\gl} \rme^v \rmd v = \int_{\gC_\pi}v^{-\gl} \rme^v \rmd v$ by Cauchy's residues theorem, hence the formula is valid for all $\xi$ satisfying~$\Re(\xi)>0$.

Suppose now~$\gt\neq \pi \mod 2\pi$, so that the half-planes bisected by~$d_{\gt}$ and by~$d_0$ have a non-empty intersection (if $\gt=\pi$, one has to proceed in two steps). Choose $\xi $ in the intersection of these two half-planes so that both the integral defining $\sB_\gt(x^\gl)(\xi)$ with contour~$\gC_{(-\gt+\pi)}$ and the same integral with contour~$\gC_\pi$  are convergent.
Cauchy's residues theorem shows again that  they  are equal.
Hence, $\sB_0(x^\gl)$ and $\sB_\gt(x^\gl)$ are analytic continuations of each other and we can say that $\sB_\gt(x^\gl)(\xi) =\frac{1}{\gC(\gl) }\xi^{\gl-1}$ for all~$\gt, -\pi<\arg(\gt)\leq +\pi$ and all~$\xi$ satisfying~$\Re(\xi \rme^{-\rmi\gt})>0$.
In particular, $\sB_\gt(x^\gl)(\xi) $ is identically~0 if and only if 
$\gl\in \Z_{\leq 0}$.

(\ref{item:dico:main}) The formula is a direct application of Lebesgue's derivation theorem.

(\ref{item:dico:special}) In  the general formula above, in order to calculate~$(1/\gC)' = -\gC'/\gC^2$ we use  the identity (\cf~\cite[(4.5.2), p.~294]{Dieu}):
\begin{equation}\label{Gamma'/Gamma}
\frac{\gC'(z)}{\gC(z)} =-\gc-\frac{1}{z} + \sum_{s\geq 1}\frac{z}{s(z+s)} \quad \textrm{for all }\ z\in\C\setminus\Z_{\leq0}.
\end{equation}

If~$\gl=p\in\Z_{>0}$ the partial fraction decomposition~$ \displaystyle \frac{z}{s(z+s)} = \frac{1}{s}-\frac{1}{z+s} $ gives the~result.

If now~${\gl=-p\in\Z_{\leq 0}}$, then~$1/\gC(-p)=0$ and
the first term~${\xi^{\gl-1} \ln \xi /{\gC(\gl)}}$ of the general formula vanishes. To obtain the value of the second term we write, using~\eqref{Gamma'/Gamma},
\begin{equation*}
 \frac{\gC'(z-p)}{\gC(z-p)} = -\gc -\frac{1}{z-p} + \sum_{s\geq 1} \frac{z-p}{s(z-p+s)}
 = \frac{1}{z} \big(-1+O(z)\big).
\end{equation*}
One has
\begin{equation*}
 \frac{1}{\gC(z-p)} =(-1)^p \frac{[p-z]_{p^-}}{\gC(z)} 
 =(-1)^p \gC(1+p)z\big(1 +O(z)\big),
\end{equation*}
hence $\displaystyle -\frac{\gC'(z-p)}{\gC^2(z-p)} = (-1)^{p} \gC(1+p) +O(z)$ and the given formula follows.
\end{proof}

As the lemma shows, the image of monomials
\begin{equation} \label{eq:formal-Borel}
  \sB_\gt \big(x^\gl\,\ln^{m} x\big)(\xi) = \frac{\rmd^{m}}{\rmd \gl^{m}} \frac{ \xi^{\gl-1}}{\gC(\gl)}
  \qquad \text{($\lambda \notin \Z_{\leq 0}$ or $m \geq 1$)}
\end{equation}
is independent of the direction~$\gt$.
This formula extends naturally to polynomials and formal series with no constant or polar part.
The resulting map
\[ \sB : x^{1+\lambda} \, \mathbb C[[x]][\log x] \to \xi^{\lambda} \, \mathbb C[[\xi]][\log \xi] \qquad (\lambda \in \C \setminus \Z_{< 0}) \]
is called the
\emph{formal Borel transformation}.
It is classically extended to constant and polar terms as explained in Remark~\ref{Dirac} below.
Later in the article we will also use the `naive' extension~$\tilde{\mathcal B}$ obtained by dropping the condition $\lambda \notin \Z_{\leq 0}$ from~\eqref{eq:formal-Borel}, that is, by putting $\tilde {\mathcal B}(x^{-p}) = 0$ for $p \in \N$.

 \begin{lemma}\label{idBorel} Given a series~$f(x)$ with no constant or polar terms, but possibly with non integer powers of~$x$, and a linear differential operator~$D$ as before, we consider~$\sB_\gt$ as defined in formula~$(\ref{anaBorel})$. Then,
  \begin{equation}\label{B(opf)}
 \sB_\gt\big(D(f)\big)(\xi)= \sB(D) \cdot \big(\sB_\gt(f)(\xi)\big).
\end{equation}
\end{lemma}
Here we denote with a dot the action of the operator~$\sB(D)$ on the function~${\sB_\gt(f)(\xi)}$. We allow ourselves to omit it further on.
\begin{proof}
 It suffices to consider the case of~$f(x)=x^\gl$, $\gl\notin\Z_{\leq 0}$ and either $D=\partial$ or~$D=1/x$.
 Suppose~$D=\partial$. 
 Then, $\sB(f)(\xi)=\xi^{\gl-1}/\gC(\gl)$ and $\sB(D)=\xi$,
 hence $\sB(D)\cdot\sB(f)(\xi)=\xi^\gl/\gC(\gl)$ 
 and therefore,
 \[
 \sB(Df)=\gl\xi^\gl/\gC(\gl+1)= \xi^\gl/\gC(\gl)=\sB(D)\cdot\sB(f)(\xi).
 \]
 Now suppose~$D=1/x$.
 Then, $\sB(D)=\rmd/\rmd \xi$ and
 \[
 \sB(D)\cdot\sB(f)(\xi)
 = \xi^{\gl-2}/\gC(\gl-1)= \sB(x^{\gl-1})(\xi) =\sB(Df)(\xi).
 \qedhere
 \]
\end{proof}

\begin{rem} \label{Dirac}
 Being concerned with differential equations it is desirable to work with morphisms of differential algebras. Consider the differential algebra~$(\C[x],\rmd/\rmd x)$ of polynomials equipped with the standard derivation. The Borel transformation as defined so far maps monomials~$x^m$ with $m\in\Z_{>0}$ to~$\sB(x^m)=\xi^{m-1}/(m-1)!$ and one can check that 
 \[\sB(x^m x^p) = \sB(x^m) * \sB(x^p)\]
 where $*$~is the convolution product given by
$(f*g)(\xi)=\int_0^\xi f(t)g(\xi-t)\,\rmd t$.
Moreover, we have~$\sB((\rmd/\rmd x)(x^m))= \xi \sB(x^m)$.
Hence, $\sB$ maps $\C[x]\setminus\C$ into the nonunital differential algebra~$\C[\xi]$ endowed with the convolution product and the multiplication by~$\xi$ as a derivation.
To extend~$\sB$ into a morphism of differential algebras, one has to set~$\sB(1)$ to the Dirac mass~$\gd_0$ since $\gd_0$~is the unit for the convolution product. Extending~$\sB$ to polynomials in~$x$ and~$1/x$ leads to set~${\sB(1/x^m) =\gd_0^{(m)}}$, the $m$th~derivative of~$\gd_0$.
\end{rem}

\begin{rem} \label{BopVSfn}
 Beware of the difference between the Borel transform of the constant function~$1$, which is the Dirac mass~$\gd_0$, and that of the operator~$1=\operatorname{id}$, which is also the identity operator.
As an example, the Borel transform of the differential operator~${\partial +1}$ is the operator of multiplication by $\xi+1$ since~${\sB(\partial \ty+ \ty)(\xi)= (\xi+1)\hy(\xi)}$.
\end{rem}
\smallskip

\subsection{Transformed operators}

Recall the notation~${\nu=\max_{0\leq\ell\leq n} \deg a_\ell(X)}$.
For all Stokes values~$\ga$, the operator 
\begin{equation*}
 { D_{[\ga]}= A_n\Big(\frac{1}{x}\Big)\partial^n + A_{n-1}\Big(\frac{1}{x}\Big)\partial^{n-1} +\cdots + A_2\Big(\frac{1}{x}\Big)\partial^2 + A_1\Big(\frac{1}{x}\Big) \partial + A_0\Big(\frac{1}{x}\Big)}
\end{equation*}
admits as a Borel transform
\begin{equation}\label{DeltaAlpha}
  \gD_{[\ga]} = A_n\Big(\frac{\rmd}{\rmd\gz}\Big) \gz^n+ A_{n-1}\Big(\frac{\rmd}{\rmd\gz}\Big)  \, \gz^{n-1} + \cdots + A_1\Big(\frac{\rmd}{\rmd\gz}\Big) \,\gz+ A_0\Big(\frac{\rmd}{\rmd\gz}\Big)
\end{equation}
where~${A_\ell \big(\frac{\rmd}{\rmd \gz} \big) =\sum_{j=0}^\nu A_{\ell,j} \big(\frac{\rmd^j }{ \rmd \gz^j} \big)}$ for~$\ell=1,\dots,n$
(thus, one has $\gD_{[0]} = \gD$).

As the following lemma shows, studying $\gD$ at~$\xi=\ga$ is equivalent to studying~$\gD_{[\ga]}$ at~$\gz=0$.

\begin{lemma}\label{Bexp} 
\begin{enumerate}
\item \label{item:Bexp:fun} Given a point $\ga\in\C$ and an analytic function~$h(x)$ with $h(0) = 0$, the Borel transform satisfies
\begin{equation*}
 \sB\big( h(x) \rme^{-\ga/x}\big)(\xi) = \sB\big(h(x)\big)(\xi-\ga).
\end{equation*}
\item \label{item:Bexp:op} Let ${\gz=\xi-\ga}$ be the local variable at~$\xi=\ga$. Then, 
\begin{equation*}
 \gD_{[\ga]}\big(y\big)(\gz) = \gD\big(y\big)(\ga+\gz).
\end{equation*}
\end{enumerate}
\end{lemma}

\begin{proof}
 (\ref{item:Bexp:fun}) One has
\begin{equation*}
 \sB\big(h(x) \rme^{-\ga/x}\big)(\xi) = \int_{\gc_0} h(x) \rme^{(\xi -\ga)/x}\frac{\rmd x}{x^2} =\sB\big(h(x)\big)(\xi-\ga).
\end{equation*}

(\ref{item:Bexp:op}) Applying a Borel transformation to  equality~${D_{[\ga]}\big(\frac{1}{x},\partial\big) = D\big(\frac{1}{x}, \partial+\ga\big)}$ (item~\ref{item:shift-dop} p.~\pageref{propertiesDalpha})      we obtain 
 $\gD_{[\ga]}\big(\frac{\rmd}{\rmd\gz}, \gz\big) = \gD\big(\frac{\rmd}{\rmd\gz}, \gz+\ga\big)= \gD\big(\frac{\rmd}{\rmd\xi}, \xi\big)$ for~$\xi=\gz+\ga$. Hence, the result.
\end{proof}

To define the map~$X \mapsto \sB(X \rme^{-\ga/x}\big)$ on the algebra~$\C[[x]]$ of formal power series one has to set~$\sB\big(1 \rme^{-\ga/x}\big)(\xi)=\gd_\ga$, the Dirac mass at~$\ga$. To extend it to formal Laurent series or to the polar part of meromorphic functions, according to Lemma~\ref{Bexp}, one has to set~${\sB(x^{-p} \rme^{-\ga/x}\big)=\frac{\rmd^p}{\rmd x^p}\sB\big(1\rme^{-\ga/x}\big)  =  \gd_\ga^{(p)}}$ for all~$p\in\Z_{\leq 0}$.

\begin{pro}\label{SingDelta} The operator~$\gD$ satisfies the following properties:
\begin{enumerate}
 \item\label{item:SingDelta:sing} The singular points of~$\gD$ are the Stokes values~$\ga$ of~$D$ at~0;
 \item\label{item:SingDelta:reg} All singular points of~$\gD$ are regular singular points;
 \item\label{item:SingDelta:ind} Given~${\gD_{[\ga]}}$ as in \eqref{DeltaAlpha}, the indicial polynomial of~$\gD_{[\ga]}$ at~$\gz=0$ reads
(\cf Notation~\ref{notalambdaj}, p.~\pageref{notalambdaj})
\begin{equation*}
 \Pi_{[\ga]}(\gl) = \sum_{\substack{\ell-j=k-\nu\\ 0\leq \ell\leq n\\0\leq j\leq \nu}} A_{\ell,j}\, [\ell+\gl]_{j^-}.
\end{equation*}
\end{enumerate}
 \end{pro}

\begin{proof}
(\ref{item:SingDelta:sing}) The singular points are the zeros of the leading term of
\begin{multline*}
\quad a_{n,\nu}\frac{\rmd^\nu}{\rmd \xi^\nu} \xi^n + \cdots + a_{1,\nu}\frac{\rmd^\nu}{\rmd \xi^\nu} \xi +a_{0,\nu} \frac{\rmd^\nu}{\rmd \xi^\nu} = \\
(a_{n,\nu}\xi^n +\cdots+a_{1,\nu}\xi+a_{0,\nu})\frac{\rmd^\nu}{\rmd\xi^\nu} +  \text{terms of lower  order.}
\end{multline*}
Thus, the coefficient of~${\rmd^\nu}/{\rmd\xi^\nu}$ is equal to the characteristic polynomial~$P(\xi)$ of~$D$ (\cf~\eqref{charpol}). Its roots, that is, the Stokes values of~$D$ at~0, are then also the singular points of~$\gD$. 

(\ref{item:SingDelta:reg}) Given a Stokes value~$\ga$ of multiplicity~$k$, let us prove that the Newton polygon of~$\gD_{[\ga]}$ reduces to a single, horizontal edge (\cf Figure~\ref{fig-dessin2bis}, p.~\pageref{fig-dessin2bis}). We know from Proposition~\ref{CNSsingle}, p.~\pageref{CNSsingle}, that~${A_{\ell,j}=0}$ for all~${(\ell, j)}$ satisfying~${\ell-j < k-\nu}$ and~${A_{k,\nu} A_{n,\nu}\neq 0}$. Thus, all marked points of the Newton polygon~$\sN(\gD_{[\ga]})$ are located above the horizontal line at height~$k-\nu$. The order of~$\gD_{[\ga]} $ is equal to~$\nu$ and the point~$(\nu, k-\nu)$ is a marked point since~${A_{k,\nu} \neq 0}$. Hence the result.

(\ref{item:SingDelta:ind}) The terms of weight~$k-\nu$ in~$\gD_{[\ga]}$ are exactly the terms~${A_{\ell,j} \frac{\rmd^j}{\rmd \gz^j}\gz^\ell}$ satisfying~${\ell-j = k-\nu}$. One obtains ${\Pi_{[\ga]}(\gl)}$ by writing that the operator restricted to its terms with weight~$k-\nu$  vanishes when applied to~${\gz^\gl}$.
\end{proof}

\subsection{Solutions of the transformed equation}\label{BorelSol}
From the previous results we know that the Borel transform of any formal solution of $Dy=0$ at~$x=0$ with exponential part~${\exp(-\ga/x)}$ is a solution of~$\gD_{[\ga]}\hy=0$ at~$\gz=0$ (or equivalently, of~$\gD \hy=0$ at~$\xi=\ga$) free of exponentials.

In this section we suppose that~$\ga$  \emph{is a Stokes value of multiplicity~$k$}: there exist $k$~linearly independent solutions of~${Dy=0}$ with exponential part~${\exp(-\ga/x)}$, and their Borel transforms are solutions of~$\gD_{[\ga]} \hy=0$ at~$\gz=0$. However, the operator~$D$ has order~$n$ while the operator~$\gD_{[\ga]}$ has order~$\nu$. We know that~$k$ satisfies~$k\leq n$, but $k$~may be smaller than, equal to or larger than~$\nu$. We now look at  what happens in each case, starting with the easiest case~$k=\nu$.

\paragraph{Case $k=\nu$}  

The equation~$Dy=0$ has~$k$ linearly independent solutions attached to~$\ga$, say~${\tth_1(x) \rme^{-\ga/x}, \dots , \tth_k(x) \rme^{-\ga/x}}$ as in the prepared fundamental solution in Sec.~\ref{preparedsol}, p.~\pageref{preparedsol}. Equivalently, ${\tth_1(x),\dots,\tth_k(x)}$ is a maximal family of linearly independent solutions of~$D_{[\ga]} y=0$ attached to the Stokes value~$0$. The exponents of~$D_{[\ga]}$ at~$0$ are the~$k$ roots of the indicial polynomial~$
{ \sP_{[\ga]}(\sL) = \sum_{\ell=0}^k A_{\ell,\ell} \,[\sL]_{\ell^+}}
$.

The Borel transforms~${\hh_1(\gz),\dots,\hh_k(\gz)}$  of~${\tth_1(x), \dots , \tth_k(x)}$ are linearly independent. Since~$k=\nu$, they provide a fundamental solution of~$\gD_{[\ga]} \hy =0$ at~$\gz =0$. 

We have seen that the exponents of~${\hh_1,\dots,\hh_k}$ are those of~${\tth_1(x), \dots , \tth_k(x)}$ decreased by~1. And indeed, the indicial polynomials~$\Pi_{[\ga]}(\gl)$ and~${ \sP_{[\ga]}(\sL)}$ satisfy the relation
\begin{equation} \label{eq:ind-keqnu}
\Pi_{[\ga]}(\gl) = \sum_{\ell=0}^\nu A_{\ell,\ell} [\gl +\ell]_{\ell^-} =\sum_{\ell=0}^\nu A_{\ell,\ell} [\gl+1]_{\ell^+} = \sP_{[\ga]}(\gl+1).
\end{equation}
Moreover, taking into account Lemma~\ref{dico}, p.~\pageref{dico} one can observe that the solutions~${\tth_1(x) \rme^{-\ga/x}, \dots , \tth_\nu(x) \rme^{-\ga/x}}$ contain logarithms at the same powers as their Borel transforms.

\paragraph{Case ${ k<\nu}$ }

As in the previous case, to the $k$ linearly independent solutions~${\tth_1(x) \rme^{-\ga/x}, \dots , \tth_k(x) \rme^{-\ga/x}}$ of the equation~${Dy=0}$ attached to~$\ga$, there correspond $k$~linearly independent solutions  of~${\gD_{[\ga]}\hy=0}$ given by
\[{\hh_1=\sB(\tth_1),\dots, \hh_k=\sB(\tth_k)}.\]
However, the equation~${\gD_{[\ga]}\hy=0}$ has~$\nu>k$ linearly independent solutions. Where do the~$\nu-k$ extra solutions come from?

Alongside~$D$, consider the differential operator
\begin{equation}\label{D1}
 D_1\, =\, \partial^{\nu-k} D_{[\ga]}\,=\,\partial^{\nu-k} \sum_{\ell=0}^n A_\ell(1/x) \partial^\ell.
\end{equation}
The Newton polygon~$ \sN(x^\nu D_1)$ of~$x^\nu D_1$ is obtained by translating that of~$x^\nu D$ (\cf Figure~\ref{fig-dessin1bis}, p.~\pageref{fig-dessin1bis}) by~$\nu-k$ units along each axis. It has a horizontal side of length~$\nu$ when~$\sN(x^\nu D)$ has a horizontal side of length~$k$.
The indicial polynomial of~$D_1$ is~$\sP_1(\sL) = [\sL-1]_{(\nu-k)^-} \, \sP_{[\ga]}(\sL)$.
The exponents of~$D_1$ are hence those of~$D$ and the extra numbers~${1,2,\dots, \nu-k}$.
Since these exponents all have positive real parts, the discussion in the case $k=\nu$ above applies to~$D_1$.

The space of solutions attached to~$0$ of the equation~$D_1 y=0$ is generated by, on the one hand, the solutions ${\tth_1(x), \dots , \tth_k(x)}$ of~$D_{[\ga]} y=0$,
and, on the other hand, solutions of each of the inhomogeneous equations
\begin{equation} \label{eq:ind-kltnu}
  {D_{[\ga]}\ y=1,\ \ \ D_{[\ga]} \ y=1/x,\ \ \dots,\ \ \  D_{[\ga]}\ y=1/x^{\nu-k-1}}.
\end{equation}
These solutions can be viewed as \emph{microsolutions} of~$D\, y = 0$ in the sense of Malgrange~\cite{MalgrangeBrasov};
see in particular Theorem~2.2 and Example~(2.4) in that reference for a related problem.
The Borel transform~$\gD_1$ of $D_1$ is ${\gz^{\nu-k} \gD_{[\ga]}}$, hence the equations~${\gD_1 \hy =0}$ and~${\gD_{[\ga]} \hy =0}$ have the same series solutions. The indicial equation of~$\gD_1$   reads~
\begin{equation*}
 \Pi_1(\gl)= [\gl]_{(\nu-k)^-} \Pi_{[\ga]}(\gl)\quad \textrm{ with } \quad {\Pi_{[\ga]}(\gl)= \sP_{[\ga]}(\gl+1)}
\end{equation*}
 and we can check that~${\Pi_1(\gl) = \sP_1(\gl+1)}$. Thus, the exponents of~$\gD_1$ are, on the one hand, the roots of~$\Pi_{[\ga]}(\gl)$, that is, the exponents of~$D_{[\ga]}$ shifted by~$-1$ (which we call the \emph{free exponents} since they can in principle take any complex value)\label{par:free-exponents}, and, on the other hand, the integers~${0,1,\dots,\nu-k-1}$ (which we call the \emph{trivial exponents}\label{par:trivial-exponents}).
The Newton polygon~$\sN(\gD_1)$ is the Newton polygon~$\sN(\gD_{[\ga]})$ translated vertically to the horizontal axis, see Figure~\ref{fig-dessin2bis}.

\begin{figure}[ht]
\begin{center}
\includegraphics[width=0.75\textwidth]{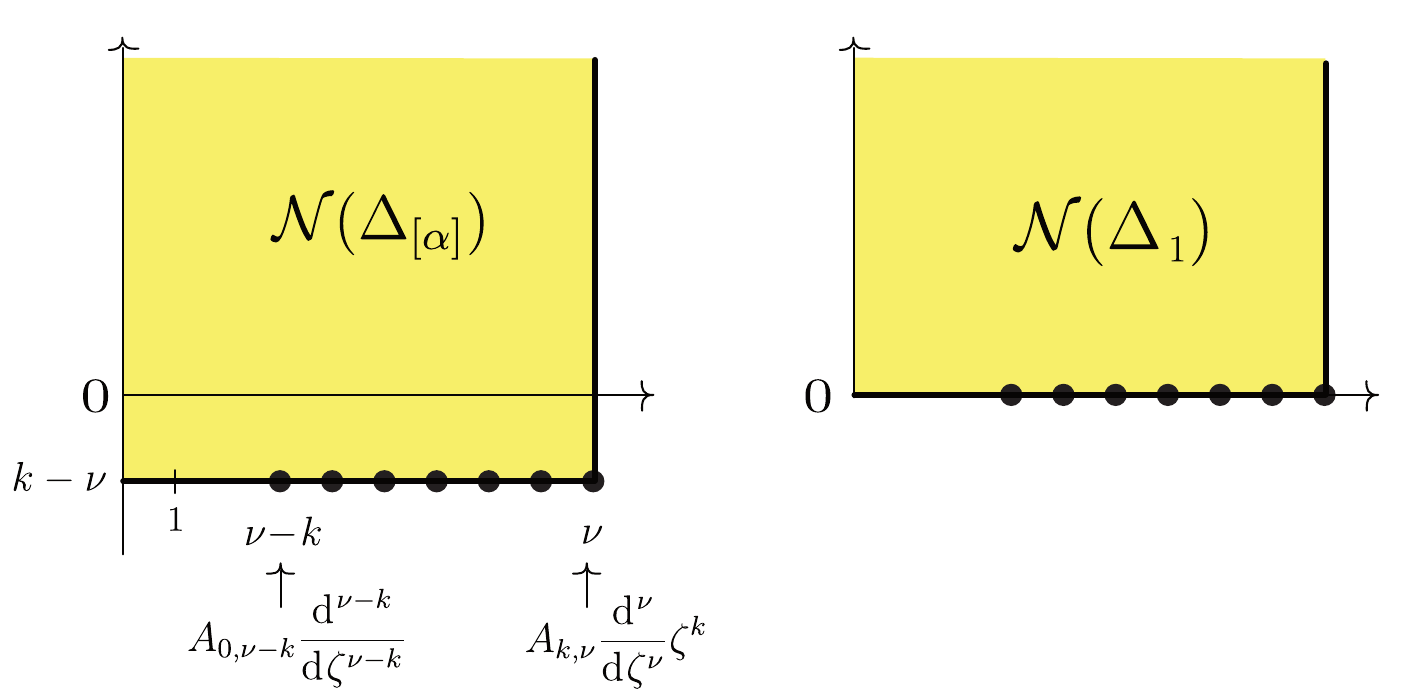}
\vspace{-2mm}\caption{The Newton polygons of $\gD_{[\ga]}$ and~$\gD_1$ when~$k< \nu$ }\vspace{-4mm}
\label{fig-dessin2bis}
\end{center}
\end{figure}

As in the case $k=\nu$, the Borel transforms~${\hh_1,\dots, \hh_k}$ of ${\tth_1, \dots, \tth_k}$ are linearly independent solutions of $\gD_{[\ga]} \hy = 0$.
Their exponents are the free exponents of~$\gD_{[\ga]}$ and they contain logarithms at the same power as~${\tth_1, \dots, \tth_k}$.

What about the solutions with trivial exponents?
Since trivial exponents differ from each other (and possibly from some of the free exponents) by integers, one may expect these solutions to involve logarithms.
More precisely, consider the exponent~$\gL=\nu-k-1-s\geq 0$, so that there are~$s$ trivial exponents larger than or equal to~$\gL$. Suppose that there are~$r \geq 0$ integer roots of~$\Pi_{[\ga]}(\gl)$ (counted with multiplicity) larger than or equal to~$\gL$.
In this situation, the classical Frobenius method for constructing solutions involving logarithms yields expressions that may contain logarithms at powers up to~$r+s-1$, and generically one of them reaches this bound.
This turns out not to be the case here: in fact, solutions attached to%
\footnote{We mean here solutions in which one can factor~$\xi^\gL$, the other factor containing only series in non negative integer powers of~$\xi$ and logarithms.\label{attached}}
the exponent~$\gL$ behave `as if the other trivial exponents did not exist'.

\begin{pro} \label{logmicrosol}
With $\gL$, $r$, and~$s$ as above, a solution attached  to the  trivial  exponent~$\gL$  of equation~${\gD_{[\ga]}\hy=0}$  may contain logarithms at a power at most~$r$.
In particular, if $\Pi_{[\ga]}(\gl)$~has no integer root larger than~$\gL$,  a solution attached to the trivial exponent~$\gL$ contains no logarithm.
\end{pro}

\begin{proof}
As we have just seen,
the space of solutions of~${\gD_{[\ga]}\hy=0}$ attached to the exponent~$\gL$ is generated by the Borel transforms of the solutions attached to $\gL + 1$ of the homogeneous equation
$D_{[\ga]} y = 0$
and the Borel transforms of one solution attached to~$\gL + 1$ of each of the inhomogeneous equations~\eqref{eq:ind-kltnu} that has one.

The method of Frobenius~\cite{Poole1936,Ince44} shows that the solutions of
$D_{[\ga]} y = 0$
attached to the exponent~$\gL + 1$
involve logarithms at powers at most~$r-1$.
In homogeneous form the equation~$D_{[\ga]} y=1/x^{p}$ for $p \in \Z_{\geq 0}$ reads
${\frac{\rmd}{\rmd x} x^{p} D_{[\ga]} y=0}$
and has at most one integer exponent larger than or equal to~$\gL + 1$ in addition to those of~$D_{[\ga]}$.
Its solutions attached to~$\gL + 1$ thus involve logarithms at powers at most~$r$.
(More precisely, the indicial equation of
${\frac{\rmd}{\rmd x} x^{p} D_{[\ga]} y=0}$
is~${(\sL+k-\nu+p)\sP_{[\ga]}(\sL)=0}$,
so that the additional exponent is $\gL + 1 + s - p$
and the inhomogeneous equation has a solution attached to~$\gL + 1$ if and only if $0 \leq p \leq s$.)

By Lemma~\ref{dico}, p.~\pageref{dico}, the Borel transforms of all these solutions involve logarithms at powers at most~$r$.
\end{proof}

Let us now consider a \emph{generic} operator~$D$ of degree~$\nu>1$ in the sense of Remark~\ref{gencase}, p.~\pageref{gencase}.
From that remark we know that, for any Stokes value~$\ga$, the operator~$D_{[\ga]}$ has a unique exponent~${\sL_{[\ga]} = -{A_{0,\nu-1}}/{A_{1,\nu}}}$. The operator~$\gD_{[\ga]}$ admits the exponent~$\gL=\sL_{[\ga]}-1$ and the~$\nu-1$ trivial exponents~${0,1,2,\dots, \nu-2}$.

\begin{cor}\label{gencaseexponent}
For generic~$D$ of degree~$\nu>1$, with $\sL_{[\ga]}$ defined as above one can set:
\begin{itemize}
 \item if~$\sL_{[\ga]} \notin \Z_{>0}$, the solutions  of~$\gD_{[\ga]}\hy=0$ contain no logarithm;
 \smallskip
 \item  if~$\sL_{[\ga]} \in \Z_{>0}$, then the solutions of~$\gD_{[\ga]}\hy=0$  with  exponent~$\gl$ satisfying~${\gl\in \{0,1,2,\dots,\nu-2\}}$ and 
 ${\gl < \gL=\sL_{[\ga]}-1}$ may involve logarithms, but only at the first power.
\end{itemize}
\end{cor}

\paragraph{Case~${\nu< k}$}

In this case, the order $\nu$ of the Borel operator~$\gD_{[\ga]}$ is smaller than the number~$k$ of solutions of~$D_{[\ga]}y=0$ attached to~0.

The conditions for~$D_{[\ga]}$ to have the single level~1 (Proposition~\ref{CNSsingle}, p.~\pageref{CNSsingle}) say, in particular, that~$A_{\ell,j}=0$ as soon as~$\ell-j < k-\nu$. When~$\ell < k-\nu$ these conditions are satisfied for all~$j$, which implies~$A_\ell(1/x) \equiv 0$. Thus, one can write~${D_{[\ga]} = D_2 \,\,\partial^{k-\nu}}$ for some differential operator~$D_2$,
and the rational functions of nonpositive valuation
${1,\, 1/x,\, 1/x^2,\,\dots,\,1/x^{ k-\nu-1}}$
are solutions of~$D_{[\ga]}\, y=0$. This cannot happen with an equation prepared as in Sec.~\ref{preparedsol}, p.~\pageref{preparedsol}. To make $D_{[\ga]}$ fit the conditions of Sec.~\ref{preparedsol} one is led to make the change of variable~$y=x^{-(k-\nu)} \, z$, after which the equation falls in case~$k=\nu$.

\section{The Stokes phenomenon in the Laplace plane} \label{sec:Stokes-Laplace}

We now recall some key facts about the \emph{Stokes phenomenon} of the linear differential equation $D\, y=0$ (equation~\eqref{Eqinitiale}, p.~\pageref{Eqinitiale}), including the definition of Stokes matrices.
We still assume that the equation is prepared as in Sec.~\ref{preparedsol}, p.~\pageref{preparedsol} and that we have fixed a formal fundamental solution
 \[
 \begin{bmatrix}
\ty_1(x)&\ty_2(x)&\cdots&\ty_n(x)\\
\end{bmatrix}
=  \begin{bmatrix} \tf_1(x) & \tf_2(x)&\cdots &\tf_n(x)\end{bmatrix} \, x^{\pL}\, \rme^{Q(1/x)}.
\]

As already mentioned, describing the Stokes phenomenon amounts to measuring the gaps between solutions with same asymptotics on a common sector based at~$0$. However, there are infinitely many ways of presenting these gaps and one has to organise them in order to understand what happens and to be able to compare various results.
From a theoretical perspective, the Stokes phenomenon of equation~\eqref{Eqinitiale} is characterised by a non-abelian 1-cocycle called the \emph{Stokes cocycle} 
that can be described as a finite collection of Stokes automorphisms (\cite[Th.~II.2.1]{{L-R94}}, \cite{BV83},     \cite[Sec.~3.5.3]{L-R16}). There is a Stokes automorphism associated with each \emph{anti-Stokes direction}, that is, each direction~$\go\in\R/(2\pi \Z)$ in which there exist  determining polynomials~$q_j(1/x)$ and~$q_\ell(1/x)$ such that~$\arg(x)=\go\in\R/(2\pi \Z)$ is a direction of maximal decay for~$\exp(q_j-q_\ell)(1/x)$. In the case under consideration where~$q_j(1/x) =-\ga_j/x$ and~$q_\ell(1/x)=-\ga_\ell /x$ this means that~$-(\ga_j-\ga_\ell)\rme^{-\rmi \go}$ is real negative \cite[Th.~3.3.5 (iv), (v)]{L-R16}.

\begin{dfn}\label{aSd}
 Let $\ga_1, \dots, \ga_N$ be the Stokes values of~$D$.
 The \emph{anti-Stokes directions} are the directions~$\go$ given by
 $${\go=\arg(\ga_j-\ga_\ell)}\ \  \text{ for all } \  j\neq \ell.$$
 If $\ty$~is a solution attached to the Stokes value~$\alpha_\ell$,
 the \emph{anti-Stokes directions associated with the solution}~$\ty$ are the directions~$\go$ such that there exists~$j \neq \ell$ with~$\go=\arg(\ga_j-\ga_\ell)$.
\end{dfn}

All details of the abstract definition of Stokes automorphisms can be found in the references cited above.
For the purposes of this paper, let us simply note that the $0$-cochain from which the Stokes cocycle is built  consists of analytic functions asymptotic to~${\begin{bmatrix} \tf_1(x) & \tf_2(x)&\cdots &\tf_n(x)\end{bmatrix}}$.  One can prove that these functions coincide with the $1$-sums of the corresponding asymptotic series on each side of~the anti-Stokes directions.

More precisely, the formal series~$\tf_j(x)$ are \emph{$1$-summable} \cite[Def.~5.1.6, Th.~5.2.5]{L-R16}, with $1$-sums in all directions except, possibly, the finitely many anti-Stokes directions. In an anti-Stokes direction~$\go$ one defines \emph{two lateral `sums'}, $f_j^-(x)$ to the right of~$\go$ and~$f_j^+(x)$ to the left of~$\go$: the sum to the right (\resp~the left) is obtained by gluing together the $1$-sums of~$\tf_j$ in the directions~$\go-\gve$ (\resp~$\go+\gve$) for all small enough\footnote{It is enough that the sector~$(\go-\gve , \go+\gve)$ does not contain any anti-Stokes direction other than~$\go$. From now, this condition is assumed.}~$\gve>0$.
On a neighbourhood of~$0$ in the open half-plane bisected by~$\go$, the lateral sums are both defined and both $1$-Gevrey asymptotic to~$\tf_j(x)$. However, in general, they are not 1-sums of~$\tf_j(x)$ in the direction~$\go$, as for this they would have to be 1-Gevrey asymptotic to~$\tf_j(x)$ on a wider sector of opening~$\pi/2-\gve$. There is a nontrivial Stokes automorphism in the direction~$\go$ if~$f_j^-(x)$ and~$f_j^+(x)$ do not agree on this common domain.

\subsection{Stokes matrices}\label{Stokes matrices}

To obtain a matrix representation of a Stokes automorphism one has to look at the gap  between fundamental solutions associated to the lateral sums.
To define these analytic fundamental solutions,
the formal factors~$x^\pL$ and~$\exp Q(1/x)$ are changed into actual functions in a natural way, with the following choices:
\begin{itemize}
 \item the determination of the argument for~$\go$ and neighbouring directions: we use the principal determination~$-\pi < \ugo \leq \pi$;
 \item the value~${\rme^0=1}$ for the exponential function at~0.
\end{itemize}
The products~${ 
\begin{bmatrix}
 f_1^\pm & f_2^\pm & \cdots f_n^\pm
\end{bmatrix}    x^\pL\,\exp Q(1/x)}$
then provide two analytic fundamental solutions~${ 
\begin{bmatrix}
 y_1^\pm & y_2^\pm& \dots & y_n^\pm
\end{bmatrix}
}$
of the equation on a neighbourhood of~$0$ in the same sector~${\mathopen\vert \arg x-\ugo\mathclose\vert<\pi/2}$.
They are hence connected by a constant invertible matrix and we can state the following definition.

\begin{dfn}\label{Smatrix}
The \emph{Stokes matrix in the direction~$\ugo$} is the unique constant invertible matrix~${I+C_\ugo}$ satisfying 
\begin{equation} \label{StokesFormula1}
\begin{bmatrix}
 y_1^{-}(x)& \cdots&y_n^{-}(x)
\end{bmatrix} =
\begin{bmatrix}
y_1^{+}(x)& \cdots&y_n^{+}(x) 
\end{bmatrix} (I+C_{\ugo}), 
\end{equation} 
 on the open  half plane bisected by~$\ugo$ near~$0$.
We denote~$C_{\ugo}= 
\begin{bmatrix}
C_{\ugo}^{[j\pv\ell]}
\end{bmatrix}_{1\leq j,\ell\leq n}$.  
\end{dfn}

\label{par:trivial-entries}
One proves that the Stokes matrix~${I+C_{\ugo}}$ is \emph{unipotent }(\ie~$C_{\ugo}$  is nilpotent): 
the diagonal~entries are equal to 1; for~${j\neq \ell}$ the entry ${C_{\ugo}^{[j\pv\ell]}}$ is zero unless~${\exp(q_j-q_\ell)(1/x)=\exp(-\ga_j+\ga_\ell)/x}$ has maximal decay, in other words, unless~${( -\ga_j+\ga_\ell) \, \rme^{-\rmi\ugo}<0}$. 
Thus, for a given~$\ell$, if there is no~$\ga_j$ such that~$\arg{(\ga_j-\ga_\ell)}=\go$,  then the $\ell$th column of~$C_\ugo$ is zero.
For more details we refer to~\cite{L-R94}.

Note that the Stokes matrix depends not only on the choice of a formal fundamental solution but also on that of the determination~$\ugo$, whereas the Stokes automorphism it represents depends only on the direction~$\go$.

\begin{exa}\label{exampleJordan}
Consider the case when the matrix~$\pL$ contains a Jordan block of size~2 and a block of size~1: $ \pL = (\sL_1 I_2 +J_2) \oplus (\sL_3 I_1)$.
Necessarily~${q_1(1/x)=q_2(1/x)\hspace{-0.5mm}:=\hspace{-0.5mm}-\ga_1/x}$ and we assume that~${q_1(1/x)\neq q_3(1/x)}$ with~${q_3(1/x):=-\ga_3/x}$. A formal fundamental solution reads 
\begin{equation*} \left\lbrace
\begin{array}{l}
  \ty_1(x) = \tf_1(x)\,  x^{\sL_1} \rme^{-\ga_1/x},\\
  \noalign{\smallskip}
  \ty_2(x)= \big(\tf_2(x) +\tf_1(x) \ln x\big) \, x^{\sL_1} \rme^{-\ga_1/x},\\
 \noalign{\smallskip}
  \ty_3(x) = \tf_3(x)\, x^{\sL_3} \rme^{-\ga_3/x}.\\
\end{array}
\right.
\end{equation*}
The unique anti-Stokes direction for~$y_1(x)$ and~$y_2(x)$ is $\go_1 = \arg(\ga_3-\ga_1)$.
The unique anti-Stokes direction for~$y_3(x)$  is $\go_3 = \arg(\ga_1-\ga_3)$.

The  Stokes matrices~${I+C_{\ugo_1}}$ and $I+C_{\ugo_3}$  are of  the form
\begin{equation*}
{I+C_{\ugo_1}}=\begin{bmatrix}
 1 & 0& 0\\
 0&1& 0\\
 C^{[3\pv 1]}_{\ugo_1}& C_{\ugo_1}^{[3 \pv 2]}& 1 \\
\end{bmatrix} 
\quad \textrm{and}\quad  
I+C_{\ugo_3} = \begin{bmatrix}
\  1\  & \  0\ & C_{\ugo_3}^{[1 \pv 3]} \ \\
 0&1& C_{\ugo_3}^{[2 \pv 3]}\\
0& 0& 1 \\
\end{bmatrix}.
\end{equation*}
We observe that both matrices are unipotent and  the following relations hold:
\begin{equation*}
 \left\lbrace
\begin{array}{l}
 y_1^{-}(x)- y_1^{+}(x) = y_3^{+}(x) \, C_{\ugo_1}^{[3\pv 1]}\\
 \noalign{\smallskip}
 y_2^{-}(x)-y_2^{+}(x) = y_3^{+}(x) \, C_{\ugo_1}^{[3\pv 2]}\
\end{array} 
 \right.
\end{equation*}
and
\begin{equation*}
y_3^-(x)-y_3^+(x)=y_1^+(x)C_{\ugo_3}^{[1\pv 3]} + y_2^+(x)C_{\ugo_3}^{[3\pv 2]}.
\end{equation*}
Note that  $\go_1$ (\resp~$\go_3$) is not an anti-Stokes direction for~$y_3$ (\resp $y_1$ and~$y_2$); hence the zeroes in the third column of~$C_{\ugo_1}$ (\resp in the third row of~$C_{\ugo_3}$).
\end{exa}
\mmedskip

\subsection{Case of the prepared fundamental solution}  \label{StokesRel}

Let us now focus, in the notation of Sec.~\ref{preparedsol}, p.~\pageref{preparedsol}, on the first block~$\pL_1$ of size~$m_1$ in the matrix of exponents~$\pL$.
Recall that $\pL_1$ is associated with a sub-block of the first block~$Q_1 = q_1(1/x) I_{k_1}$ of the matrix~$Q$ of determining polynomials.
Considering the first~${m_1}$ formal solutions will allow us to determine the first~${m_1}$ columns of the Stokes matrix~${I+C_\ugo}$.
There is no loss of generality in restricting ourselves to~$\pL_1$ since any block of the decomposition can be brought in first position by a permutation of the solutions.

   The first~$m_1$ formal solutions~${\tY_1(x) = 
\begin{bmatrix}
 \ty_1(x) & \ty_2(x)&\dots & \ty_{m_1}(x)\\
\end{bmatrix}}$
 are  of the form
\begin{equation}\label{firstm1}
\hspace{-1mm}\left\lbrace
\begin{array}{l}
\ty_1(x)= \tf_1(x) \, x^{\sL_1}\, \rme^{q_1(1/x)}\\
\ty_2(x) = (\tf_2(x) +\tf_1(x) \ln x)  \, x^{\sL_1}\, \rme^{q_1(1/x)}\\
\cdots\\
\ty_{m_1}(x)= \big(\tf_{m_1}(x) +\cdots + \tf_1(x) \ln^{m_1-1} x\big) \, x^{\sL_1}\, \rme^{q_1(1/x)}\\
\end{array}
\right.
\end{equation}
They are associated with the anti-Stokes directions
\[\arg(\ga_2-\ga_1),\dots,\arg(\ga_N-\ga_1).\]
Given one such direction~$\go$, with principal determination~$\ugo$ we set
\begin{equation}
 {\gb_1=\ga_{i_1}-\ga_1 , \quad\gb_2=\ga_{i_2}-\ga_1 ,\quad \cdots\ , \quad\gb_r=\ga_{i_r}-\ga_1}
\end{equation}
where $\ga_{i_1}, \dots, \ga_{i_r}$ are the Stokes values lying on the ray~$d_\go$ originating from~$0$ in the direction~$\go$, ordered by increasing distance from~$0$.

In restriction to the first~$m_1$ columns, the first~$k_1$ rows in~$C_\ugo(x)$ are zero and only the rows~${i\in {\sI_{i_1}\cup \sI_{i_2} \cup \dots \cup \sI_{i_r}}}$ corresponding to the blocks 
\begin{equation*}
  Q_{i_\ell}(1/x)=q_{i_\ell}(1/x) I_{k_{i_\ell}}= -\ga_{i_\ell}/x I_{k_{i_\ell}}\quad \textrm{for} \ \ 
\ell=1,2,\dots,r
\end{equation*}
may be nonzero.
From the definition~\eqref{StokesFormula1} of the Stokes matrix we derive
\begin{equation}\label{StokesY1}
\displaystyle h_\ell^-(x)- h_\ell^+(x) = \sum_{s=1}^r \rme^{-\gb_s/x} \Big(\sum_{j\in\sI_{i_s}} h_j^+(x) C_\ugo^{[j\pv \ell]}   \Big)\quad\textrm{for} \ \ell=1,2,\dots,m_1
\end{equation}
where~$h_j(x)$ is defined by~$y_j(x) =h_j(x) \rme^{-\ga_i/x}$ when~$j\in\sI_{i}$.

\section{The Stokes phenomenon viewed from the Borel plane} \label{sec:Stokes-Borel}

The idea is now to rewrite the definition~\eqref{StokesFormula1} seen from the Borel plane in order to get a system of  linear equations for the Stokes multipliers.
As in the previous section, we focus on the block
\[
 \begin{bmatrix}\ty_1&\dots&\ty_{m_1}\end{bmatrix} =  \begin{bmatrix}\tth_1&\dots&\tth_{m_1}\end{bmatrix} \rme^{-\ga_1/x}
\]
with~$\begin{bmatrix}\tth_1&\dots&\tth_{m_1}\end{bmatrix} =
\begin{bmatrix}\tf_1&\dots&\tf_{m_1}\end{bmatrix} x^{\frak{L}_1} $ and~$\frak{L}_1=\sL_1 I_{m_1} +J_{m_1}$.

The theory says that the power series~$\tf_\ell(x)$ are 1-summable in all direction but the anti-Stokes directions and that one can obtain these sums by applying a Borel transformation followed by a Laplace transformation.
Given the anti-Stokes direction~$\go$, we are interested in the sums in the directions~$\go\pm\gve$ for~$\gve$ small enough.
Thanks to the condition on the valuation required for prepared solutions we can apply these transformations directly to the~$\tth_\ell(x)$. Indeed, since $\tth_\ell$~has positive valuation, the valuation of~$\hh_\ell=\sB(\tth_\ell)$ is~$> -1$ and the Laplace integrals converge at~$0$.
As a solution of a linear differential equation with no singular point on the rays~$d_{\go\pm\gve}$, $\hh_\ell$ is continuable up to infinity along these rays.
The resulting functions have at most exponential growth at infinity because the operator~${\gD_{[\ga_1]}}$ is of level~$\leq 1$ at infinity, and hence the Laplace integrals converge.

To obtain the left hand-side of formula~\eqref{StokesFormula1}  one should a priori first  analytically continue~$\hh_\ell(\gz)$ to the right and to the left of~$d_\go$, then take a Laplace integral of~${\hh_\ell(\gz)}$ from~$\gz=0$ to infinity in direction~$\go-\gve$, and then a Laplace integral from infinity to~$\gz=0$ in direction~$\go+\gve$. It is more convenient to split the integration path~$\gc$  into~$r$ paths~$\gc_{j}$ going around each singular point~$\gb_1,\gb_2,\dots,\gb_r$ in the ~$\gz$-plane
 (\ie the points~${\ga_{i_1} , \ga_{i_2} ,\cdots,\ga_{i_r}}$ in the~$\xi$-plane) as depicted on Figure~\ref{fig-splitpath}.

\begin{figure}[!ht]
\begin{center}
\includegraphics[width=0.48\textwidth]{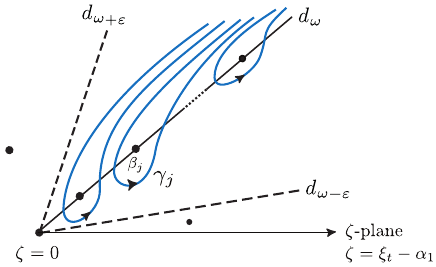}
\vspace{-1mm}
\caption{ 
The path~$\gc_{\go}$ split into finitely many paths $\gc_j$;\hfill\\
\null\qquad\qquad \emph{the big black dots represent the singular points of~$\gD$ }}
\label{fig-splitpath}
\end{center}
\end{figure}

The  solutions~$\hh_\ell(\gz)$, for~${\ell=1,\dots,m_1}$ of the linear differential equation~${\gD_{[\ga_1]} \hy=0}$  may be analytically continued along any path that avoids the singular points of~${\gD_{[\ga_1]}}$. 
The analytic continuation of the~$\hh_\ell(\gz)$'s along the path of integration~$\gc_\go$ should be done as follows: first, one defines~$\hh_\ell(\gz)$ to the right of~$d_\go$ by analytic continuation along a path~${\gc_{0\rightarrow\gb_r}}$ as shown on Figure~\ref{fig-gamma-minus}.
Then, once  the value of~$\hh_\ell(\gz)$ is known on an arc of each~$\gc_{j}$, one considers  the analytic continuations along both branches of the~$\gc_{j}$'s from these values. We denote by~$\hh_\ell^-$ the resulting function. Note that in practice, only the analytic continuation along the path~$\gc_{0\rightarrow\gb_r}$ will need to be computed explicitly, and this can be done in a finite number of steps by means of the Cauchy-Weierstrass method (\cf sec.~\ref{anacont}, p.~\pageref{anacont}).

 \begin{figure}[!ht]
\begin{center}
\includegraphics[width=0.45\textwidth]{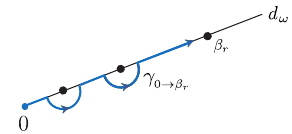}
\vspace{-1mm}\caption{The path ${\gc_{0\rightarrow\gb_r}}$ in the~$\gz$-plane ($\gz=\xi-\ga_1$)}
\label{fig-gamma-minus}
\end{center}
\end{figure}
Integrating along the split path~$\gc$ we obtain, for~$\ell=1,\dots,m_1$,
\begin{equation}\label{StokesLaplace} 
h_\ell^-(x)- h_\ell^+(x) \, = \,
  \sum_{s=1}^r \, \rme^{-\gb_s/x}\,
\int_{\gc_{0} } \hh_{\ell}^{\hspace{1pt }-} \big(\gb_s+\gz\big) \rme^{-\gz/x} d\gz
\end{equation}
where $\gc_{0}$ is the Hankel type path around~$\gz=0$ obtained from~${\gc_{i_s} }$ by the translation $-\gb_s$.

\section{Comparing the two approaches}
\label{sec:compare}

Comparing equations~\eqref{StokesY1}, p.~\pageref{StokesY1}, and~\eqref{StokesLaplace} we obtain, for~$\ell=1,2,\dots,m_1$,
\begin{equation*}
 \sum_{s=1}^r \rme^{-\gb_s/x} \Big(\sum_{j\in\sI_{i_s}} h_j^+(x) C_\ugo^{[j\pv \ell]}   \Big)
 =  \sum_{s=1}^r \, \rme^{-\gb_s/x}\,
\int_{\gc_{0} } \hh_{\ell}^{\hspace{1pt}-} \big(\gb_s+\gz \big) \rme^{-\gz/x} d\gz.
\end{equation*}
Perhaps the most subtle point in the approach is that  one can identify the exponential terms on each side despite the fact that they have functions as coefficients. The reason for this is that these functions are summable-resurgent~\cite[Lemma~4.2]{LR11}. We deduce for all column indices~${\ell = 1,\dots ,m_1}$ and row block indices~${s=1,\dots, r}$ (recall that~$\sI_{i_s}$ is a set of~$k_{i_s}$ consecutive rows indices) the relation
\begin{equation}\label{identite}
 \sum_{j\in\sI_{i_s}} h_j^+(x) \, C_\ugo^{[j\pv \ell]}  =
 \int_{\gc_0 } \hh_{\ell}^{\hspace{1pt}-} \big(\gb_s+\gz \big) \rme^{-\gz/x} d\gz.
\end{equation}

This formula is to be understood as follows.
Given~$s$, the functions~$h^+_j(x)$ for $j \in \sI_{i_s}$ form a basis of a linear space of dimension~$k_{i_s}$ and  we must find the coordinates~$C_\ugo^{[j\pv \ell]} $ of the function of the  right-hand side in this basis.
While the functions~$h^+_j$ and~$\hy^{\hspace{1pt}-}_\ell$ are theoretically well-defined, it is convenient for the effective calculation of the coordinates to replace both sides of the equation~\eqref{identite} by their asymptotic expansions as $x \to 0$.
In other words, \eqref{identite}~is advantageously replaced by
\begin{equation}\label{identite2}
 \boxit(8pt,8pt,3pt,3pt){$\displaystyle%
 \sum_{j\in\sI_{i_s}} \widetilde h_j(x) \, C_\ugo^{[j\pv \ell]}  = T_{(x=0)}
 \int_{\gc_0 } \hh_{\ell}^{\hspace{1pt}-} \big(\gb_s+\gz \big) \rme^{-\gz/x} d\gz
 $}
\end{equation}
where the asymptotic expansion of the integral on the right hand-side is easily calculated with the help  of the following lemma.

\begin{lemma} \label{lemmeconn} 
Let $\gC_{\ugo}$ be  a Hankel contour around the half-line $d_{\go}$, oriented positively 
from~$\arg\gz=\ugo-2\pi$ to  $\arg\gz=\ugo$, 
and let $x\in \C$ be a point of the half-plane $\abs{\arg(x)-\ugo}<\pi/2$.
\begin{enumerate}
 \item\label{item:lemmecon:analytic} Suppose that $\gvf(\gz)\rme^{-\gz/x}$ is integrable at infinity in the direction $\ugo$ and that $\gvf(\gz)$ is analytic on and inside $\gC_{\ugo}$.
 Then,
 \[ \int_{\gC_{\ugo}} \gvf(\gz)\, \rme^{-\gz/x} \, \rmd\gz =0. \]

\item\label{item:lemmecon:nologs} For any $\lambda \in \C$, one has
\[ I_\gl(x):=\int_{\gC_{\ugo}} \gz^\gl\, \rme^{-\gz/x} \, \rmd\gz =  \frac{2\pi\rmi}{\gC(-\gl)} \, \big(x\rme^{-\pi\rmi}\big)^{1+\gl}. \]
In particular, $I_\gl=0$ for all $\gl\in\N$ in accordance with item~(\ref{item:lemmecon:analytic}).

\item\label{item:lemmecon:logs} For any $\lambda \in \C$ and ${p\in\N}$, one has
\[ J_{\gl\pv p}(x):=
 \int_{\gC_{\ugo}} \gz^\gl\, \ln^p(\gz)\,  \rme^{-\gz/x} \, \rmd\gz =\frac{\rmd^p}{\rmd \gl^p} \, I_\gl(x) \]
 and $J_{\gl\pv p}(x)$ is of the form $x^{1+\gl} R\big(\ln (x)\big)$ with~$R\big(\ln( x)\big)\in \C[\ln (x)]$.
 In particular, for~$\gl\in\N$,
\begin{align*}
  J_{\gl\pv 1} &= 2\pi \rmi \,\gC(1+\gl)\, x^{1+\gl}\\
  \noalign{\smallskip}
J_{\gl\pv 2} &= 4\pi \rmi\, \gC(1+\gl)\,\big( -\gc+
\sH_\gl -\pi\rmi+\ln x\big) \, x^{1+\gl}.
\end{align*}
where
 $\sH_\gl=\sum_{j=1}^\gl \frac{1}{j} $ is the harmonic sum of order~$\gl$. 
\end{enumerate}
\end{lemma}
\begin{rem}
 The integrands in the lemma are functions with no singular point other than~$0$ on a sector neighbouring~$d_\go$. Thus, the contour of integration may be deformed so as to look like the contours~$\gc_0$ in formula~\eqref{StokesLaplace}  or~$\gc_{s}$'s on figure~\ref{fig-splitpath}, p.~\pageref{fig-splitpath}.
\end{rem}

\begin{proof}
Item (\ref{item:lemmecon:analytic})   results from  the residue theorem of Cauchy. 

In the next two items the  integrals~$I$ and~$J$ converge  at infinity for any~$x$ such that~${-\pi/2<\ugo-\arg(x)<\pi/2}$.

To prove (\ref{item:lemmecon:nologs}), suppose first that~$\arg(x)=\ugo$ (the case we are interested~in). The  change of variable~$u=\gz \rme^{\rmi\pi}/x$ leads 
to
\[ {\int_{\gC_{\ugo}} \gz^\gl\, \rme^{-\gz/x} \, \rmd\gz =(x \, \rme^{-\rmi\pi})^{\gl+1} \int_{\gC_\pi} u^\gl\, \rme^u \rmd u} \]

The integral to the right is  the one appearing in the integral formula for~$1/\gC(-\gl)$ \cite[(4.8.1)~p.~296]{Dieu} valid for all~$\gl\in\C$.
The case $\arg(x)=\ugo+\gvf$ with~$\abs{\gvf}<\pi/2$ reduces to the previous one by rotation.

The general formula in point~(\ref{item:lemmecon:logs}) results from~(\ref{item:lemmecon:nologs}) and the theorem of derivation of Lebesgue.
The proof of the special cases is similar to that of Lemma~\ref{dico} (3), p.~\pageref{dico}, except that we now use the second-order expansion
\begin{equation*}
\frac{1}{\gC(-\gl - z)} = (-1)^\gl \gC(1+\gl) \, z\, \big(1+(\sH_\gl - \gc)z +O(z^2)\big)
\end{equation*}
(which can be deduced from the series expansion~$1/\gC(z)= z+\gc z^2+O(z^3)$ \cite[5.7.i]{NIST} by applying the relation $\Gamma(z+1) = z \Gamma(z)$ repeatedly).
Substituting this expansion into the expression of $I_{\lambda}$, we get
\begin{align*}
  I_{\lambda + z}(x)
  &= 2 \pi \rmi
     (-1)^\gl \gC(1+\gl) \, z\, \big(1+(\sH_\gl - \gc)z +O(z^2)\big) \\
  & \phantom{=} \cdot (x e^{-\pi\rmi})^{1+\gl}
     \big( 1 + (\ln x - \pi\rmi) z + O(z^2) \big) \\
  &= 2 \pi \rmi \, x^{1+\lambda} \, \Gamma(1+\lambda) \,
     \bigl( z + (\sH_\gl - \gc + \ln x - \pi\rmi) z^2 + O(z^3) \bigr)
\end{align*}
and the result follows.
\end{proof}
\begin{rem} \label{remprepared}
In Sec.~\ref{preparedsol}, p.~\pageref{preparedsol}, while defining the prepared fundamental solution, we multiplied the solutions of the original equation by a power of~$x$ to reduce to solutions of positive valuation.
Doing so ensures that the Laplace transforms of the Borel transforms of the solutions are well defined.
However, it also increases the order of the transformed operator:
suppose that we had to make the change of variable $z = x^{-u} y$ in the original equation $D_1\,z=0$ to get the prepared equation $D \, y = 0$.
Then, the Borel transform~$\gD$ of~$D$ has order $u$~more than the Borel transform~$\gD_1$ of~$D_1$.
Thus, from a practical perspective, it would be more economical to work with~$\gD_1$, and as we show below this is indeed possible.

The Stokes multipliers are provided by formula~\eqref{identite2}, which is relative to the prepared operator~$D$.
Denote by $\widetilde k_j$~and~$\widehat k_\ell^{\hspace{1pt}-}$ the quantities analogous to $\widetilde h_j$~and~$\widehat h_\ell^{\hspace{1pt}-}$ but defined relative to the equation~$D_1 \, z = 0$.
Let us first rewrite the right-hand side of~\eqref{identite2} in terms of~$\widehat k_\ell^{\hspace{1pt}-}$.
For this, split~$\widetilde h_j(x)$ as $\widetilde P_j(x)+ \tH_j(x)$ with a polynomial part~$\widetilde P_j(x)$ of degree~$u$.
Accordingly, write $\widetilde k_j(x) =\widetilde R_j(1/x) + \tK_j(x)$ where $\tK_j(x) =({1}/{x^u} ) \tH_j(x)$ and $\widetilde R_j(1/x)$ is the polar part.
After a Borel transformation, we get the decompositions (compatible with analytic continuation)
\begin{equation*}
 \widehat h_\ell^{\hspace{1pt}-}(\gz)= \widehat P_\ell(\gz)+ \widehat H_\ell^{\hspace{1pt}-}(\gz), \qquad
 \widehat k_\ell^{\hspace{1pt}-}(\gz) =\widehat R_\ell(\gd) + \hK_\ell^{\hspace{1pt}-}(\gz)
\end{equation*}
where $\widehat P_\ell(\gz)$ is a polynomial of degree at most $u-1$ and
$\widehat R_\ell(\gd)$ is a linear combination of derivatives of Dirac masses at  $\gz\equiv \xi-\ga_1=0$.
The relation~$\tK_\ell(x) =({1}/{x^u}) \tH_\ell(x)$ implies~${\hK_\ell(\gz) =\frac{\rmd^u}{\rmd \gz^u} \hH_\ell(\gz)}$. 
Moreover, $\hH_\ell(\gz)$ and its derivatives up to the order~$u$ have at most exponential growth at infinity, hence one has
\begin{align*}
 \int_{\gc_0} \widehat K_\ell^{\hspace{1pt}-}(\gb_s+\gz) \rme^{-\gz/x} \rmd\gz
 &=\int_{\gc_0} \frac{\rmd^u}{\rmd \gz^u}\widehat H_\ell^{\hspace{1pt}-}(\gb_s+\gz) \rme^{-\gz/x} \rmd\gz
\end{align*}
and after $u$ integrations by parts
\begin{align*}
 \int_{\gc_0} \widehat K_\ell^{\hspace{1pt}-}(\gb_s+\gz) \rme^{-\gz/x} \rmd\gz
 &=\frac{1}{x^u} \int_{\gc_0} \widehat H_\ell^{\hspace{1pt}-}(\gb_s+\gz) \rme^{-\gz/x} \rmd\gz\\
 &=\frac{1}{x^u} \int_{\gc_0} \widehat h_\ell^{\hspace{1pt}-}(\gb_s+\gz) \rme^{-\gz/x} \rmd\gz.
\end{align*}
Turning to the left-hand side of~\eqref{identite2},
we have $\widetilde k_j(x) =(1/x^u) \tth_j(x)$, hence
\[\sum_{j\in\sI_{i_s}} \widetilde k_j(x) \, C_\ugo^{[j\pv \ell]}  =\frac{1}{x^u}  \sum_{j\in\sI_{i_s}} \widetilde h_j(x) \, C_\ugo^{[j\pv \ell]} \]
and equation~\eqref{identite2} still holds with the~$h$'s replaced by~$k$'s.
\end{rem}

\label{par:sketch}
In summary, given~$\ell \in \ensemble{1,\dots,m_1}$ (in the notation introduced p.~\pageref{blocQ}) and~$s$, the~$k_{i_s}$ Stokes multipliers~$C_\ugo^{[j\pv\ell]}$ appearing in identity~(\ref{identite}) are located in rows~$j\in\sI_{i_s}$ and column~$\ell$ of the Stokes matrix.
Replacing all functions by their asymptotic expansions and identifying the coefficients of the expansions on both sides of the equality, we obtain an infinite inhomogeneous linear system whose solutions are the Stokes multipliers~$(C_\ugo^{[j\pv\ell]})_{j\in \mathcal{I}_{i_s}}$. 

To be able to compute the asymptotic expansion of the integral on the right-hand side, we need to know the asymptotic expansion  of~$\hy_{\ell}^{\hspace{1pt}-}(\gz)$   at each of its singular points~$\gb_s$.
Computing these expansions is the main task of the effective analytic continuation procedure discussed in the next section.

By repeating these computations for all $\ell \in \ensemble{1,\dots,m_1}$ and $s=1,\dots,r$, we can determine all Stokes multipliers in the first~$m_1$ columns that are associated with~$\sL_1$ and~$\rme^{-\ga_1/x}$. 
Next, the same procedure must be applied to the other sub-blocks~$\sL_i$  associated with~$\rme^{-\ga_1/x}$, and finally to all exponential parts~${\rme^{-\ga_j/x}, j=1, \dots,N}$.

\section{Analytic continuation} \label{anacont}

To conclude this first part, we now recall how the analytic continuation of~$\hh_\ell$ along the path~$\gc =\gc_{0\rightarrow \gb_{r}}$ depicted on Figure~\ref{fig-gamma-minus} can be carried out in practice using the Cauchy-Weierstrass method of series expansion on overlapping discs.

As previously, we consider a solution~${\ty_\ell(x)=\tth_\ell (x)\rme^{-\ga_1/x}, \,1\leq \ell\leq m_1}$.
The Borel transform~$\hh_\ell(\gz)$ of~$\ty_\ell(x)$, with~${\gz=\xi-\ga_1}$, is a sum of products of power series, complex powers of~$\gz$, and  logarithms. The power series are convergent and we interpret the complex powers and logarithms as analytic functions by choosing the principal determination of~$\arg\gz$.
Also recall that $\hh_\ell(\gz)$ is a solution of $\Delta_{[\alpha_1]}$ and that
$${\gb_1=\ga_{i_1} \hspace{-2pt}-\ga_{1}, \ \gb_2=\ga_{i_2} \hspace{-2pt} -\ga_{1},\  \dots,\ \gb_r=\ga_{i_r} \hspace{-2pt} -\ga_{1}}$$
are the singular points of~$\Delta_{[\alpha_1]}$ on the ray~$d_\go$, ordered from from~$\gz=0$ to infinity.

Since we need to know the behaviour of~$\hh(\gz)$ at all~${\gb_{s}, \,s\hspace{-1pt}=\hspace{-1pt}1,2,\dots,r}$, we choose a covering of~$\gc$ by a chain of discs $\sD_0, \sD_1, \dots, \sD_t$ whose centres include all singular points~$\gb_s$.
We denote  by~${\chi_0, \chi_1,\dots,  \chi_t}$ their centres (in the~$\gz$-plane) and by  $1/\rho_0, \dots,  1/\rho_t$ their radii.
The first disc~$\sD_0$ is centred at~$\chi_0=0$ ; the last one~$\sD_t$ is centred at~${\chi_t=\gb_{r}}$; and the discs contain no singular point of~$\gD_{[\ga_1]}$ except possibly at their centres.
\begin{figure}[ht]
\begin{center}
\includegraphics[width=0.8\textwidth]{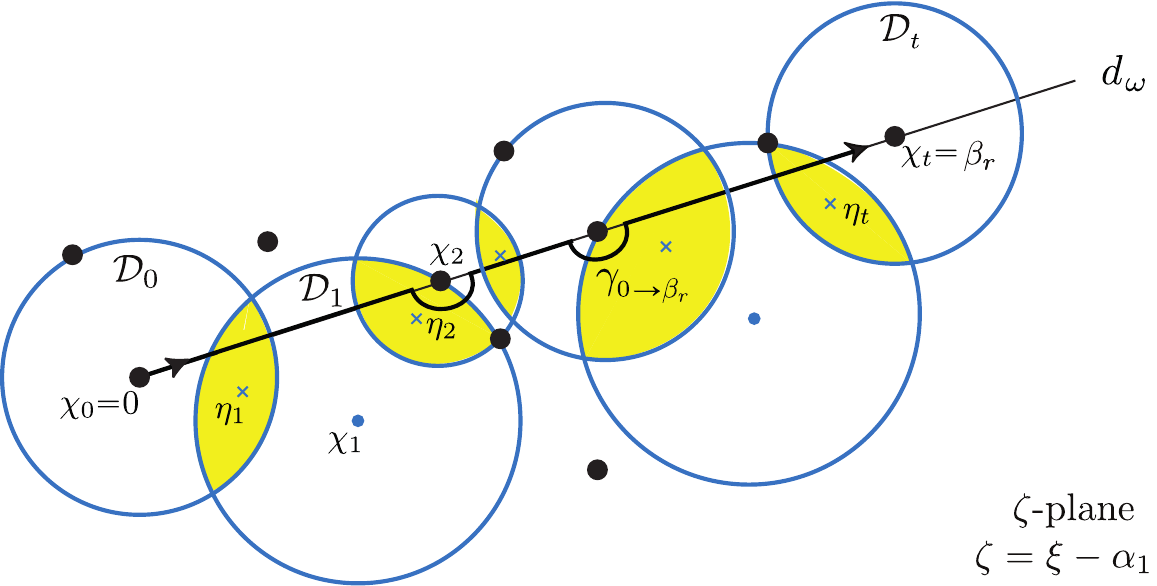}
\caption{Covering $\operatorname{Cov}(\gc_{0\rightarrow\gb_r})$ of the path $\gc_{0\rightarrow\gb_r}$; {\it the big black dots represent the singular points of $\gD_{[\ga_1]}$ limiting the size of the discs, the blue dots indicate the remaining centres, the points $\eta_1,\dots,\eta_t$ are chosen in the successive intersections}}
\label{fig-dessin9}
\end{center}
\end{figure}
We number the discs so that any two consecutive discs overlap and we fix points 
\begin{equation*}
 \eta_1\in \sD_0\cap \sD_1,\ \eta_2\in \sD_1\cap \sD_2,\ \dots,\  \eta_t\in \sD_{t-1}\cap \sD_t.
\end{equation*}
We denote by $\zeta_0,\zeta_1,\dots,
\zeta_t$  the local coordinates at~${\chi_0,\chi_1,\dots,\chi_t}$, obtained by setting 
\begin{equation*}
 \zeta_0=\xi - \ga_1,\  \zeta_1= \zeta_0-\chi_1,\ \dots,\ \zeta_{t-1}=\gz_0-\chi_{t-1}, \ \zeta_t=\gz_0-\chi_t
\end{equation*}
and choosing the principal determinations $-\pi<\arg\gz_p \leq +\pi$ for all $p$. 

For each disc~$\sD_p$, we choose a local fundamental solution~\hspace{-1pt}${ \hY_{[\chi_{p}]}(\zeta_{p})}$ of the equation~${\gD_{[\chi_{p}]} \hy=0}$ at~${\gz_p=0}$ in the classical form
\begin{equation} 
\begin{array}{rcl}
  \hY_{[\chi_{p}]}(\zeta_{p}) &=& \begin{bmatrix}
 \hy_{[\chi_p],1}(\gz_p) & \hy_{[\chi_p],2}(\gz_p)& \cdots&  \hy_{[\chi_p],\nu}(\gz_p)
\end{bmatrix}\\
\noalign{\mmedskip}
&=&  \begin{bmatrix}
 \hh_{[\chi_p],1}(\gz_p) & \hh_{[\chi_p],2}(\gz_p)& \cdots&  \hh_{[\chi_p],\nu}(\gz_p)
\end{bmatrix} \gz_p^{\gL_p}\\
\end{array}
\end{equation}
where the~$\hh$'s are power series. The matrix of exponents~${\gL_p}$ is zero when $\chi_{p}$~is an ordinary point; otherwise, it is in Jordan form with diagonal terms given by proposition~\ref{SingDelta}, p.~\pageref{SingDelta}.
 
Denote by~${\phi_0(\gz_0)}$ the sum of the series~${\hy_\ell(\gz_0)}$ on its disc of convergence~$\sD_0$.
By evaluating this sum and its first~$\nu-1$ derivatives at~$\gz_0=\eta_1$, we express it in the basis~$ \hY_{[\chi_1]}(\zeta_1)$
 at the point~$\gz_1=\eta_1-\chi_1$ (choosing the principal determination of the argument~${-\pi<\arg(\eta_1-\chi_1)\leq +\pi}$).
This yields a constant vector~$T_1$ such that 
\begin{equation}\label{sigma1}
 \phi_0(\eta_1)\,=\, \hY_{[\chi_1]}(\eta_1-\chi_1)\, T_1 
\end{equation}
and a  function
\begin{equation*}
\phi_1(\chi_1+\gz_1)=\hY_{[\chi_1]}(\gz_1)\,T_1
\end{equation*}
which is the analytic continuation of $\phi_0(\gz_0)$ to  $\sD_1$.
Iterating the process, we obtain a function
\begin{equation}\label{phit}
\phi_p(\chi_p+\gz_p)=\hY_{[\chi_p]}(\gz_p)\,T_p
\end{equation}
defined in the neighbourhood of each $\chi_p$
which is the analytic continuation along~$\gamma$ of $\phi_0(\gz_0)$.
Due to the uniqueness of analytic continuation, the result is independent of the choice of the discs~${\sD_0,\dots, \sD_t}$.
In practice, for better performance, it is desirable to have large intersections so that we can choose the points~$\eta$ not too close to the boundary.

The local expansions of $\hh_{\ell}^{\hspace{1pt}-}$ at the singular points~$\beta_p$, $p=1,2, \dots, r$ that enter into the asymptotic expansion on the right-hand side of~\eqref{identite2} can then be read off equation~\eqref{phit}
for the disc~$\sD_p$ with centre~$\chi_p=\gb_s$,
since the basis functions in~$\hY_{[\chi_p]}$ were defined from their local expansions in the first place.

\part{Algorithm and implementation}
\label{part:algorithm}

We now aim to turn the procedure for computing Stokes matrices described in the previous two sections into a detailed algorithm that can be implemented using existing features of computer algebra systems. The detailed algorithm eliminates numerous redundancies in the method, making it much faster in practice. However, we leave for future work a full complexity analysis and confine ourselves to observing that the algorithm can be implemented so as to compute the Stokes multipliers of any \emph{fixed} operator~$D$ within an error of $2^{- p}$ in $\mathrm{O} (p \log (p)^3)$ operations (see Remark~\ref{rk:complexity}, p.~\pageref{rk:complexity}).

Compared to the algorithm sketched at the end of Sec.~\ref{sec:compare}, p.~\pageref{par:sketch}, the complete version works as much as possible with fundamental matrices, or at least with blocks of solutions attached to a given Stokes values and exponent in the decomposition of Sec.~\ref{preparedsol}, p.~\pageref{preparedsol},
instead of considering each solution~$\ty_j(x)$ individually.
In the same spirit, it computes the Stokes matrices in all directions simultaneously. The main benefit of this modified structure is that the number of analytic continuation steps needed is minimized.

We have implemented this algorithm---with some minor changes with respect to the present description---using the SageMath computer algebra system. Our implementation is freely available under the GNU General Public Licence (version~2 or later) as part of the \codestar{ore\_algebra} package, which can be downloaded from
\codestar{\href{https://www.github.com/mkauers/ore\_algebra/}{https://www.github.com/mkauers/ore\_algebra/}}.
The algorithm is implemented in the module
\codestar{ore\_algebra.analytic.stokes}.
Below we give some examples of its use and occasionally comment on implementation details. For more information and additional examples, we refer the reader to the documentation shipped with the package. After installing \codestar{ore\_algebra} and starting SageMath, one can access the documentation of the main function of the \codestar{stokes} module with the commands
\begin{alltt}
sage: from ore_algebra.analytic.stokes import stokes_dict
sage: ?stokes_dict
\end{alltt}
In the following sections, we begin by illustrating the usage of the implementation on an example (Sec.~\ref{sec:running-example}) that will also serve as a running example when detailing the algorithm. Then we discuss the representation of complex numbers (Sec.~\ref{sec:constants}) and generalised power series (Sec.~\ref{sec:local-bases}) on which the algorithm operates. We next explain how, in each Stokes matrix, the block associated with a given pair of Stokes values decomposes as a product of three auxiliary matrices that can be computed independently and reused for several blocks (Sec.~\ref{sec:matrices}). This allows us to state the main procedure (Sec.~\ref{sec:algo-main}), initially in terms of subroutines for computing the auxiliary matrices. In the next two sections, we detail these subroutines, the most important one being that corresponding to the analytic continuation step~(Sec.~\ref{sec:connection}). We conclude with some additional examples~(Sec.~\ref{sec:more-examples}).

\section{An example}\label{sec:running-example}

A.~Duval and C.~Mitschi~\cite{DM89} compute the Stokes matrices of the confluent generalised hypergeometric operator
\begin{equation}
  D_{q, p} = (- 1)^{q - p} z \prod_{j = 1}^p \left( z \ddx{z} + \mu_j \right) - \prod_{j = 1}^q \left( z \ddx{z} + \nu_j - 1 \right) \label{eq:Dqp}
\end{equation}
in terms of values of the gamma function. As a running example we use the special case
\begin{equation}
  p = 1, \quad q = 7, \quad \mu = \left( \frac{1}{2} \right), \quad \nu = \left( 1, 0, 0, 0, 0, \frac{2}{3}, \frac{1}{3} \right) \label{eq:running-params}
\end{equation}
considered in Example~2.45 of Mitschi and Sauzin's lecture notes~\cite{MSau16}. In Example~\ref{ex:hgeom} (p.~\pageref{ex:hgeom}), we discuss several instances of the case $q = 2$, $p = 1$ to illustrate various degenerate situations.

\begin{figure}
  \includegraphics{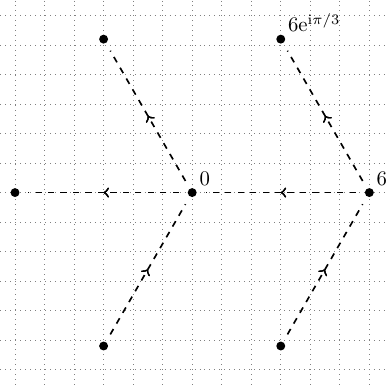}
  \caption{\label{fig:stokes-values}The Stokes values of the operator~$D$ in our running example, with the spanning tree considered in Sec.~\ref{sec:connection}.}
\end{figure}

\begin{example}
  \label{ex:running-intro}
  After the change of variables $\, x = z^{- 1 / 6}\,$ and  $\,{\partial = x^2 \mathrm{d} / \mathrm{d} x}$ as in Part~\ref{part:theory}, the operator~\eqref{eq:Dqp} specialised to the parameters~\eqref{eq:running-params} becomes
  \[ \textstyle D = \frac{1}{x^6} \partial^7 + \frac{9}{x^5} \partial^6 + \frac{58}{x^4} \partial^5 + \frac{272}{x^3} \partial^4 + \frac{897}{x^2} \partial^3 + \frac{1875}{x} \partial^2 + \left( - \frac{46656}{x^6} + 1875 \right) \partial + \frac{139968}{x^5} . \]
  This is an operator of pure level~$1$, with Stokes values $0$ and $6 \mathrm{e}^{j \mathrm{i} \pi / 3}$,\hspace{-2pt} ${- 2 \leqslant j \hspace{-1pt}< \hspace{-1pt}2}$ (Figure~\ref{fig:stokes-values}).
  
  In order to reproduce Mitschi's example numerically, we load the relevant part of \codestar{ore\_algebra} and set up an algebra of differential operators with the commands
\begin{alltt}
sage: from ore_algebra.analytic.stokes import stokes_dict
sage: R.<x> = PolynomialRing(QQ)
sage: A.<Dx> = OreAlgebra(R)
\end{alltt}
  \noindent These last two lines define \codestar{x} as the indeterminate of the ring $R = \mathbb{Q} [x]$ and \codestar{Dx} as the indeterminate of a skew polynomial ring $A = \mathbb{Q} [x] \langle D_x \rangle$ where $D_x$ acts on expressions involving~$x$ as the standard derivation $\mathrm{d} / \mathrm{d} x$. We can thus define the operator~$D$ with
\begin{alltt}
sage: d = x^2*Dx
sage: diffop = (1/x^6*d^7 + 9*1/x^5*d^6 + 58*1/x^4*d^5
....:   + 272*1/x^3*d^4 + 897*1/x^2*d^3 + 1875*1/x*d^2
....:   + (-46656*1/x^6 + 1875)*d + 139968*1/x^5)
\end{alltt}
  Then, to compute the Stokes matrices of~$D$, we run the implementation of the algorithm described in this section with the command
\begin{alltt}
sage: stokes = stokes_dict(diffop, 10^-50)
\end{alltt}
  \noindent The second argument, \codestar{10\^{}-50}, is a rough indication of the accuracy we would like to obtain and is used internally to select the working precision for intermediate computations. We will discuss its role in more detail in what follows. The computation on this example takes around~1\,s on a mid-range laptop purchased in 2020. The \codestar{stokes\_dict} function returns a dictionary \codestar{stokes} mapping~$\mathrm{e}^{\mathrm{i} \dir}$ where $\dir$~is an anti-Stokes direction to the corresponding Stokes matrices. The choice of~$\mathrm{e}^{\mathrm{i} \dir}$ rather than $\dir$~itself as an indexing key is motivated by technical details of the infrastructure provided by SageMath, which make it more convenient to use algebraic numbers as keys.
  
  In particular, for the direction $\dir = 0$, we obtain a matrix of the form (the output has be reformatted and intermediate digits omitted for readability)
\begin{alltt}
sage: stokes[1]
\end{alltt}
\begin{multline*}
  \tiny
  \setlength\arraycolsep{1ex}
  \left[
  \begin{matrix}
      1.0                                                                                 & 0                                                                                        & \cdots \\
      0                                                                                   & 1.0                                                                                      & \cdots \\
      0                                                                                   & 0                                                                                        & \cdots \\
      {}[\pm 2.11 \cdot 10^{- 47}] + [0.14 \ldots 78 \pm 3.86 \cdot 10^{- 47}] \mathrm{i} & 0                                                                                        & \cdots \\
      0                                                                                   & \hspace*{-7em}[- 2.5 \pm 5.88 \cdot 10^{- 35}] + [4.33 \ldots 09 \pm 7.62 \cdot 10^{- 35}] \mathrm{i}  & \cdots \\
      0                                                                                   & 0                                                                                        & \cdots \\
      {}[16.0 \pm 6.61 \cdot 10^{- 42}] +  [\pm 6.61 \cdot 10^{- 42}] \mathrm{i}          & 0                                                                                        & \cdots \\
  \end{matrix}
  \right. \\
  \tiny
  \setlength\arraycolsep{1ex}
  \left.
  \begin{matrix}
      0                                                                                     & 0                                                                                    & 0   & 0   & 0   \\
      0                                                                                     & 0                                                                                    & 0   & 0   & 0   \\
      1.0                                                                                   & 0                                                                                    & 0   & 0   & 0   \\
      0                                                                                     & 1.0                                                                                  & 0   & 0   & 0   \\
      0                                                                                     & 0                                                                                    & 1.0 & 0   & 0   \\
      [2.5 \pm 2.36 \cdot 10^{- 23}] + [4.33 \ldots 86 \pm 3.95 \cdot 10^{- 23}] \mathrm{i} \hspace*{-7em} & 0                                                                                    & 0   & 1.0 & 0   \\
      0                                                                                     & [\pm 4.36 \cdot 10^{- 44}] + [- 221.7 \ldots 25 \pm 5.58 \cdot 10^{- 44}] \mathrm{i} & 0   & 0   & 1.0 \\
    \end{matrix} \right]
\end{multline*}
  An entry $[c \pm r]$ in this output indicates that the corresponding exact value is contained in the interval $[c - r, c + r]$. Observe that some of the radii~$r$ are significantly larger than the tolerance~$10^{- 50}$ specified on input, which is treated as indicative only. However, decreasing the tolerance results in tighter enclosures, as illustrated here on the nontrivial entry of the third column (note that row and column indices in Sage start from zero):
\begin{alltt}
sage: stokes_dict(diffop, 10^-100)[1][5,2]
\end{alltt}
  \[ \tiny{[2.5 \pm 4.41 \cdot 10^{- 69}] + [4.3301270189\dots1744862983 \pm 6.95 \cdot 10^{- 69}] \mathrm{i}} \]
  By further decreasing it one can compute, in principle, arbitrarily tight enclosures of the exact Stokes multipliers.
  
  We note in passing that, exceptionally, this Stokes multiplier is an algebraic number, so that it is easy to guess its minimal polynomial from the approximation, here using the \codestar{algdep} function from \textsc{Pari}~\cite{Pari}, and from there the \ exact value $5\, \mathrm{e}^{\mathrm{i} \pi / 3}$:
\begin{alltt}
sage: algdep(stokes_dict(diffop, 10^-100)[1][5,2].mid(), 2)
x^2 - 5*x + 25
\end{alltt}

  Similarly, for $\dir = \pi / 6$, we have
\begin{alltt}
sage: stokes[exp(I*pi/6)]
\end{alltt}
\begin{multline*}
     \tiny
     \setlength\arraycolsep{1ex}
     \begin{bmatrix}
       1.0                                                                                       & 0                                                                                                    & 0   &         \\
       0                                                                                         & 1.0                                                                                                  & 0   &         \\
       0                                                                                         & 0                                                                                                    & 1.0 &         \\
       0                                                                                         & 0                                                                                                    & 0   & \ddots  \\
       0                                                                                         & 0                                                                                                    & 0   &         \\
       {}[- 5.5 \pm 5.33 \cdot 10^{- 31}] + [9.52 \ldots 82 \pm 8.31 \cdot 10^{- 31}] \mathrm{i} & 0                                                                                                    & 0   &         \\
       0                                                                                         & \hspace*{-9em} [5.5 \pm 1.06 \cdot 10^{- 32}] + [9.52 \ldots 23 \pm 1.26 \cdot 10^{- 32}] \mathrm{i} & 0   &
     \end{bmatrix}
\end{multline*}
  We omit the remaining ten Stokes matrices with $- \pi < \dir \leqslant \pi$, which, due to the symmetry resulting from the deramification step, are easily expressed in terms of the previous two.
  
  Mitschi's results are relative to a formal fundamental solution~$\mathcal{Z} (x)$ specified in Propositions 1.3~and~2.1 of~\cite{DM89} and related to our formal fundamental solution~$\mathcal{Y} (x)$ (detailed in Example~\ref{ex:running-bases} below) by
  \[ \mathcal{Z} (x) =\mathcal{Y} (x)  \left[ \begin{smallmatrix}
       0 & 0 & 0 & 0 & 0 & 0 & c\\
       0 & c & 0 & 0 & 0 & 0 & 0\\
       0 & 0 & 0 & 0 & 0 & c & 0\\
       1 & 0 & 0 & 0 & 0 & 0 & 0\\
       0 & 0 & c & 0 & 0 & 0 & 0\\
       0 & 0 & 0 & 0 & c & 0 & 0\\
       0 & 0 & 0 & c & 0 & 0 & 0
     \end{smallmatrix} \right] \qquad \text{where $c = \frac{4 \pi^{5 / 2}}{\sqrt{3}}$.} \]
  In addition, her Stokes matrix corresponds to the inverse of ours. Taking this into account, our evaluation of the Stokes multipliers agrees with the formulae listed in~\cite[Example~2.43]{MSau16}\footnote{An error appears to have slipped in the specialisation of these formulae for the parameters~\eqref{eq:running-params} in \cite[Example~2.45]{MSau16}.}, except that we find $\eta = 16$ instead of $\eta = - 16$.
\end{example}

\section{Representation of complex numbers}\label{sec:constants}

\paragraph{Computable complex numbers and intervals.}The Stokes multipliers are in general transcendental numbers with no known `exact' representation simpler than their definition. Our algorithm must be able to compute these numbers within an error of no more than some given~$\varepsilon$, and, by necessity, does so by operating on complex numbers that are known approximately only. This means that some mechanism is needed for keeping track of the accumulation of approximation errors along the course of the computation.

Perhaps the first such mechanism that comes to mind is to determine ahead of time how accurately each intermediate operation needs to be carried out in order for the final result to satisfy the error tolerance. This strategy is extremely tedious and error-prone, and we do not attempt to give any bounds of this type.

For theoretical purposes, however, we can mostly ignore the issue if we work with computable complex numbers, in the sense of the following classical definition.

\begin{definition}
  \label{def:computable}A complex number~$x$ is \emph{computable} if there exists an algorithm that takes as input a tolerance~$\varepsilon > 0$ and returns an approximation $\tilde{x} \in \mathbb{Q} [i]$ of~$x$ satisfying $| \tilde{x} - x | \leqslant \varepsilon$. We denote by~$\CComp$ the set of computable complex numbers.
\end{definition}

If $x$ and $y$ are computable complex numbers, then $x + y$ and $xy$ are clearly computable; and, if $x$ is a nonzero computable number, its inverse $1 / x$ is computable. Similarly, the standard functions such as $\log x$ and $1 / \Gamma (x)$ that we use in what follows all map computable numbers from some open domain to computable numbers. However, there is no algorithm to test if two computable numbers are equal.

Formally, the algorithms presented in this article output matrices of computable complex numbers. The fact that the entries of the matrices are computable means that it is possible, given any~$\varepsilon > 0$, to obtain approximations with an error guaranteed not to exceed~$\varepsilon$.

In principle, a computable complex number can be implemented as an object that contains an approximation of a number~$x \in \mathbb{C}$ but also keeps track of the definition of~$x$ in terms of other computable numbers~(\emph{e.g.}, \cite{vanderHoeven2006,Muller2001}). When asked for an $\varepsilon$-approximation of~$x$, the object refines the current approximation by recursively recomputing the quantities on which~$x$ depends to higher and higher precision until the error tolerance is met. In this model, once the basic functions on which our algorithm relies are available as functions on~$\CComp$, all error propagation required for the whole algorithm to produce rigorous arbitrary-precision approximations is automatic. No explicit error tolerances appear as input of the algorithms below since the \emph{output} is, effectively, a program taking as input a tolerance and returning an approximation of the mathematical result.

This model is occasionally used in practice but comes with substantial computational overhead. As already illustrated in Example~\ref{ex:running-intro}, in our implementation, we replace the use of computable complex numbers by interval arithmetic.

\begin{definition}
  We denote by~$\Balls$ the set of complex intervals $[a, b] + [c, d] i$ where the endpoints $a, b, c, d$ are in a fixed dense set of rational numbers. Arithmetic operations and standard elementary and special functions $f : \mathbb{C}^s \rightarrow \mathbb{C}$ are extended to partial functions~$\Balls^s \rightarrow \Balls$ in such a way that, for $\mathbf{x}_1, \ldots, \mathbf{x}_s \in \Balls$, the value $f (\mathbf{x}_1, \ldots, \mathbf{x}_s)$, if defined, satisfies
  \[ f (\mathbf{x}_1, \ldots, \mathbf{x}_s) \supseteq \{ f (x_1, \ldots, x_s) : x_1 \in \mathbf{x}_1, \ldots, x_s \in \mathbf{x}_s \} \]
  and $f (\mathbf{x}_1, \ldots, \mathbf{x}_s)$ is undefined if there is any $(x_1, \ldots, x_s) \in \mathbf{x}_1 \times \cdots \times \mathbf{x}_s$ where $f (x_1, \ldots, x_s)$ is undefined.
\end{definition}

For instance, $\pi$ might be represented by $[3.14, 3.15] + [- 0.1, 0.1] i$ and $[3.14, 3.15] + [0.1, 0.2]$ might be defined to be $[3.2, 3.4]$. In addition to the parameters of the version working over~$\CComp$, the interval versions take as input a \emph{working precision} $p \in \mathbb{N}$ specifying the number of significant digits to be used during the computation\footnote{In Example \ref{ex:running-intro}, the working precision is chosen automatically based on the value of the error tolerance parameter.}. All basic operations on computable complex numbers are replaced by their interval counterpart. The working precision determines the accuracy of the output of basic operations when it is not already limited by that of the input: for example, while $1 / [0.9, 1.0]$ must be an interval of width at least $0.1$ regardless of the working precision, $1 / [3, 3]$ might be defined to be $[0.3, 0.4]$ at precision~$1$ and $[0.33, 0.34]$ at precision~$2$. Note that interval operations may be undefined if their input contains a point where the function is not defined, or simply if the input interval is too wide. In this case, the whole algorithm raises an error.

In general, the output of an interval operation is only guaranteed to contain the corresponding exact result, not to satisfy any preset error tolerance. However, we are only using intervals as approximations of computable complex numbers. This means that we can assume the basic interval operations to be implemented in such a way that none of the interval operations fails when the working precision is large enough, and that each output tends to a point as the working precision tends to infinity. As illustrated in Example~\ref{ex:running-intro}, to compute the Stokes constants with an error bounded by $\varepsilon > 0$, the user\footnote{It tends to be useful to leave it to the user to repeat the computation until the tolerance is met rather than doing it automatically, since, in many applications, one is actually interested in an output that is accurate enough for checking some property or performing some subsequent computation, with no easy way to tell in advance exactly how accurate it needs to be. See also the Arb documentation~\cite{UsingBallArithmetic} for more on this point.} of our implementation can first run the whole algorithm with a working precision~$p$ such that $2^{- p} \approx \varepsilon$, and repeat while the algorithm fails or the accuracy of the result exceeds the tolerance.

\paragraph{Exact complex numbers.}Not all operations we need can be performed using computable complex numbers. Indeed, finding the Stokes values lying on a given ray or deciding if two local exponents differ by an integer requires \emph{exact} equality tests. These tests operate on quantities derived from the coefficients of the operator~$D$ by algebraic operations.

For this reason, we assume from now on that $D$~has coefficients in~$\mathbb{K} (x)$ for some subfield~$\mathbb{K} \subset \mathbb{C}$ where the required operations can be carried out. Specifically, we assume that $\mathbb{K}$ is algebraically closed and closed under complex conjugation, and that we have a way of representing its elements with a finite amount of data. We also need algorithms operating on this representation for
\begin{itemize}
  \item testing if two given elements of~$\mathbb{K}$ are equal,
  
  \item performing the field operations $+, -, \times, /$ in~$\mathbb{K}$,
  
  \item computing the complex conjugate of an element of~$\mathbb{K}$,
  
  \item computing the roots of univariate polynomials with coefficients in~$\mathbb{K}$,
  
  \item given $a \in \mathbb{K}$ and $\varepsilon > 0$, computing an approximation $\tilde{a} \in \mathbb{Q} [i]$ of~$a$ such that $| \tilde{a} - a | \leqslant \varepsilon$.
\end{itemize}
These assumptions imply that we can compute $\ensuremath{\operatorname{Re}} (a)$, $\ensuremath{\operatorname{Im}} (a)$, $| a |$ given $a \in \mathbb{K}$, and that we can decide whether $a < b$ for given $a, b \in \mathbb{K} \cap \mathbb{R}$. The last one means that $\mathbb{K} \subseteq \CComp$.

The prototypal example of a field with these properties is $\mathbb{K}= \QQbar$ with elements represented as roots of univariate polynomials with rational coefficients together with isolating intervals. Johansson~\cite{Johansson2021} describes a practical system for working with such `exact' subfields of~$\mathbb{C}$ covering both~$\bar{\mathbb{Q}}$ and, under suitable number-theoretic assumptions, extensions by transcendental constants. Our implementation is currently limited to the algebraic case.

\section{Local bases of solutions}\label{sec:local-bases}

The symbolic part of our algorithm operates on solutions of~$D$ and their Borel transforms.
To represent these objects using a finite amount of data, we encode them by coordinates in fixed `local' bases of (formal or analytic) solutions of the operators $D_{[\alpha]}$~and~$\Delta_{[\alpha]}$.

\paragraph{Formal solutions in the Laplace plane.}\label{sec:fsol}

For each Stokes value~$\alpha$ of~$D$, we denote by $\ensuremath{\operatorname{FSol}}_0 (D_{[\alpha]})$ the space of formal solutions free of exponentials of~$D_{[\alpha]}$.
In other words, $\ensuremath{\operatorname{FSol}}_0 (D_{[\alpha_i]})$ is the vector space generated by the block~$(\tilde h_j)_{j \in \mathcal I_i}$ of the regular part~$\tilde h$ of the prepared formal fundamental solution introduced in Sec.~\ref{preparedsol}, with $\mathcal I_i$ as in~\eqref{blocQ}.

We denote by~$\mathcal{H}_{[\alpha]}$ a fixed basis of
$\ensuremath{\operatorname{FSol}}_0 (D_{[\alpha]})$.
For most of the discussion, any basis arranged by exponent modulo~$\mathbb Z$ will do, including $(\tilde h_j)_{j \in \mathcal I_i}$ once normalized as explained in Sec.~\ref{preparedsol}.
However, for compatibility with preexisting code, our implementation uses a specific choice of~$\mathcal{H}_{[\alpha]}$ that deviates slightly from the form considered in that section%
\footnote{Specifically, the basis it works with is compatible with the block and sub-block structure of Sec.~\ref{preparedsol}, but it cannot necessarily be factored as the product of a matrix of formal power series by $x^\pL$ with~$\pL$ in Jordan form. The echelon form property that we require simplifies some aspects of the algorithms but the loss of the factorization complicates others, and it is not clear to us at the moment what choice would be best with no consideration for backward compatibility.}.
Some subroutines below are stated to be applicable to this basis, or, more generally, to any basis~$\mathcal{H}_{[\alpha]}$ that is \emph{echelonized} in the following sense, and require minor adaptations otherwise.

To obtain an echelonized basis,
we start from any basis of
$\ensuremath{\operatorname{FSol}}_0 (D_{[\alpha]})$
consisting of series of the form
\begin{equation}\label{eq:local-sol-shape}
  x^{\mathcal{L}}  \sum_{r = 0}^{s - 1} \sum_{m \geqslant 0} c_{r, m} x^m \ln (x)^r
\end{equation}
where $\mathcal L$ is one of the exponents of~$D$ attached to the Stokes value~$\alpha$ and $s$~is the dimension of the corresponding Jordan block in~$\mathfrak L$.
For each class $\mathcal L +\mathbb{Z}$ of exponents differing by integers, we write the corresponding basis elements with a common~$\mathcal L$ of maximal real part (the `group leader' of Definition~\ref{def:structure} below), and enumerate the coefficients of each of these expansions in the order
\begin{equation}
  c_{s - 1, 0}, c_{s - 2, 0}, \ldots, c_{0, 0}, \quad c_{s - 1, 1}, c_{s - 2, 1}, \ldots, c_{0, 1}, \quad \ldots \label{eq:series-coeff-for-echelon-form}
\end{equation}
Then we put these resulting family of sequences in reduced echelon form. One can show (\emph{e.g.}, \cite[Chapter~V]{Poole1936}) that the pivots that appear are the coefficients~$c_{\rho, \mu}$ where $(\mathcal{L}, \mu, \rho)$ belongs to the set~$\mathcal{E} (D_{[\alpha]})$ defined as follows.

\begin{definition}
  \label{def:structure}Given any differential operator~$D$, we say that an exponent~$\mathcal{L}$ of~$D$ at the origin is a \emph{group leader} when none of the numbers $\mathcal{L} - 1, \mathcal{L} - 2, \ldots$ is also an exponent. We denote by $\mathcal{E} (D)$ the set of triples $(\mathcal{L}, \mu, \rho) \in \mathbb{K} \times \mathbb{N} \times \mathbb{N}$ such that $\mathcal{L}$~is a group leader and $\mathcal{L} + \mu$ is an exponent of multiplicity strictly larger than~$\rho$, and we call the set $\mathcal{E} (D)$ the \emph{structure} of the local solutions free of exponentials of~$D$.
\end{definition}

An echelonized basis is naturally indexed by $\mathcal{E} (D_{[\alpha]})$ in the sense that the coefficient of $x^{\mathcal{L} + m} \ln (x)^r$ in a given basis element is~$1$ for exactly one triple $(\mathcal{L}, m, r) \in \mathcal{E} (D_{[\alpha]})$ and $0$~for all others. We will make free use of this correspondence and refer, for instance, to rows and columns of matrices using indices taken from $\mathcal{E} (D_{[\alpha]})$.

We assume from now on that the $\mathcal{H}_{[\alpha]}$ are echelonized,
and we denote by
\[
  \mathcal{Y} (x)
  = \bigl(
      \rme^{-\alpha_1/x} \mathcal{H}_{[\alpha_1]},
      \dots,
      \rme^{-\alpha_N/x} \mathcal{H}_{[\alpha_N]}
    \bigr)
  = (y_1 (x), \ldots, y_n (x))
\]
the formal fundamental solution obtained by collecting the families $\mathcal{H}_{[\alpha]}$ for all Stokes values~$\alpha$.

\begin{remark}
  \label{rk:unroll}When the coefficients $c_{j, m}$ of a solution~\eqref{eq:local-sol-shape} are enumerated in the order~\eqref{eq:series-coeff-for-echelon-form}, each coefficient can be computed from the previous ones using simple recurrence relations deduced from the differential operator~\cite[Chapter~V]{Poole1936}. This means that the series expansions of the solutions~${y}_{n}$ are easily computed to any desired order.
\end{remark}

\begin{remark}\label{rk:sol-ordering}
  The conditions above determine $\mathcal Y$ up to permutations of the Stokes values and classes of exponents.
  In our implementation,
  the families~$\mathcal{H}_{[\alpha]}$ are always sorted in such a way that whenever two solutions are asymptotically comparable, the dominant one as $x \rightarrow 0$ comes first. To be completely specific, we sort the solutions first by increasing $\ensuremath{\operatorname{Re}} (\mathcal{L})$, then, in case of ties, by decreasing degree with respect to $\ln (\zeta)$ of the coefficient of~$\zeta^{\mathcal{L}}$, then by decreasing absolute value of $\ensuremath{\operatorname{Im}} (\mathcal{L})$ (so that purely real exponents come last for a given $\ensuremath{\operatorname{Re}} (\mathcal{L})$), and finally by increasing~$\ensuremath{\operatorname{Im}} (\mathcal{L})$.

  Similarly, we sort the Stokes values firstly by increasing real part, so that the ordering of $y_1, \ldots, y_n$ as a whole also reflects the asymptotic dominance relation as $x \rightarrow 0$ with~$x > 0$. When several Stokes values have the same real part, we order them by decreasing absolute value of the imaginary part, so that, for each real part, purely real values come last. Finally, complex conjugates are sorted by increasing imaginary part.
  In particular, our Stokes matrices in the direction $\dir = 0$ are lower triangular.
\end{remark}

\begin{example}[Continued from Example~\ref{ex:running-intro}]
  \label{ex:running-bases} Denoting $\theta = \mathrm{e}^{\mathrm{i} \pi / 3}$, the formal fundamental solution of the operator~$D$ from Sec.~\ref{sec:running-example} used in our implementation is $\mathcal{Y}= (y_1, \ldots, y_7)$ where
  \allowdisplaybreaks[2]
  \begin{align*}
    y_1 (x) & = \mathrm{e}^{+ 6 / x} x^{- 2}  \left( 1 + \tfrac{2}{9} x + \tfrac{49}{648} x^2 + \tfrac{67}{17496} x^3 + \cdots \right)\\
    y_2 (x) & = \mathrm{e}^{- 6 \theta^{- 2} / x} x^{- 2}  \left( 1 + \tfrac{\theta^{- 2} + 6}{27} x + \tfrac{49}{3888} \theta^{- 2} x^2 + \tfrac{67}{17496} x^3 + \cdots \right)\\
    y_3 (x) & = \mathrm{e}^{- 6 \theta^2 / x} x^{- 2}  \left( 1 + \tfrac{\theta^2 + 6}{27} x + \tfrac{49}{3888} \theta^2 x^2 + \tfrac{67}{17496} x^3 + \cdots \right)\\
    y_4 (x) & = x^3 + \tfrac{315}{128} x^6 + \cdots\\
    y_5 (x) & = \mathrm{e}^{- 6 \theta^{- 1} / x} x^{- 2}  \left( 1 + \tfrac{\theta^{- 1} - 6}{27} x + \tfrac{49}{3888} \theta^{- 2} x^2 - \tfrac{67}{17496} x^3 + \cdots \right)\\
    y_6 (x) & = \mathrm{e}^{- 6 \theta / x} x^{- 2}  \left( 1 + \tfrac{\theta - 6}{27} x + \tfrac{49}{3888} \theta x^2 - \tfrac{67}{17496} x^3 + \cdots \right)\\
    y_7 (x) & = \mathrm{e}^{- 6 / x} x^{- 2}  \left( 1 - \tfrac{2}{9} x + \tfrac{49}{648} x^2 - \tfrac{67}{17496} x^3 + \cdots \right) .
  \end{align*}
  \noindent Each of the spaces $\ensuremath{\operatorname{FSol}}_0 (D_{[\alpha]})$ is one-dimensional, with $\ensuremath{\operatorname{FSol}}_0 (D_{[0]}) =\mathbb{C}y_4$ and $\ensuremath{\operatorname{FSol}}_0 (D_{[6]}) =\mathbb{C} \mathrm{e}^{6 / x} y_7$ for instance.
  
  The Borel transform of~$D$ is equal to
  \begin{multline*}
    \Delta = \Delta_{[0]} = (\zeta^7 - 46656 \zeta)  \diff{\zeta}{6} + (51 \zeta^6 - 139968)  \diff{\zeta}{5} + 958 \zeta^5  \diff{\zeta}{4} \\
    + 8332 \zeta^4  \diff{\zeta}{3} + 34521 \zeta^3  \diff{\zeta}{2} + 62289 \zeta^2  \diff{\zeta}{} + 36015 \zeta .
  \end{multline*}
  Its indicial polynomial is $\tilde{\Pi}_{[0]} (\lambda) = [\lambda]_{5^-}  (\lambda - 2)$, and the corresponding basis structure is
  \[ \mathcal{E} (\Delta_{[0]}) = \{ (0, 0, 0), (0, 1, 0), (0, 2, 1), (0, 2, 0), (0, 3, 0), (0, 4, 0) \} . \]
  The basis $\localbasis{0}$ of the space $\ensuremath{\operatorname{Sol}} (\Delta_{[0]})$, again with the conventions of our implementation, takes the form
  \[ \localbasis{0} (\zeta) = \begin{aligned}[t]
       \left( \vphantom{\frac{}{}} \right. & 1 + \tfrac{2401}{8957952} \zeta^6 + \mathrm{O} (\zeta^{12}), \\
       & \zeta + \tfrac{64}{382725} \zeta^7 + \mathrm{O} (\zeta^{12}), \\
       & \zeta^2 \log (\zeta) + \left( \tfrac{1}{8192} \log (\zeta) - \tfrac{19}{589824} \right) \zeta^8 + \mathrm{O} (\zeta^{12}), \\
       & \zeta^2 + \tfrac{1}{8192} \zeta^8 + \mathrm{O} (\zeta^{12}), \\
       & \zeta^3 + \tfrac{125}{1285956} \zeta^9 + \mathrm{O} (\zeta^{12}), \\
       & \zeta^4 + \tfrac{14641}{179159040} \zeta^{10} + \mathrm{O} (\zeta^{12}) \left. \vphantom{\frac{}{}} \right) .
     \end{aligned} \]
  For all other Stokes values~$\alpha$, one has $\tilde{\Pi}_{[\alpha]} (\lambda) = [\lambda]_{5^-}  (\lambda + 3)$.
\end{example}

\paragraph{Local bases in the Borel plane.}

Similarly, for any $\alpha \in \mathbb{C}$, we denote by $\ensuremath{\operatorname{Sol}}_0 (\Delta_{[\alpha]})$ the space of germs of solutions of~$\Delta_{[\alpha]}$ analytic in a neighborhood of the origin slit along the negative real axis, and we fix a basis
\[ \localbasis{\alpha} = (\hat{y}_{\alpha, 1}, \ldots, \hat{y}_{\alpha, \nu}) \]
of $\ensuremath{\operatorname{Sol}}_0 (\Delta_{[\alpha]})$.
We choose each basis element~$\hat{y}_{\alpha, i}$ in the form
\begin{equation}
  \zeta^{\lambda}  \sum_{r = 0}^{s - 1} \sum_{m \geqslant 0} c_{r, m} \zeta^m \ln (\zeta)^r \label{eq:local-sol-shape-borel}
\end{equation}
where~$\lambda$ is one of the exponents of~$\Delta_{[\alpha]}$ at~$0$ and~$s \in \mathbb{N}$.
As above in the case of formal solutions in the Laplace plane, we assume that within families of series~\eqref{eq:local-sol-shape-borel} that have the same exponent modulo~$\mathbb Z$, the sequences
\begin{equation}
  c_{s - 1, 0}, c_{s - 2, 0}, \ldots, c_{0, 0}, \quad c_{s - 1, 1}, c_{s - 2, 1}, \ldots, c_{0, 1}, \quad \ldots
\end{equation}
are in reduced echelon form.
This way, the basis $\localbasis{\alpha}$ of $\ensuremath{\operatorname{Sol}}_0 (\Delta_{[\alpha]})$ is naturally indexed by $\mathcal{E} (\Delta_{[\alpha]})$.
In the implementation, the exponent classes $\lambda + \mathbb Z$ are ordered as specified in Remark~\ref{rk:sol-ordering}.

\section{Factorization of Stokes blocks}\label{sec:matrices}

As discussed in Sec.~\ref{Stokes matrices}, p.~\pageref{par:trivial-entries}, the Stokes matrix in the direction~$\dir$ has a (potentially) nontrivial block associated with each pair $(\alpha, \beta)$ of Stokes values such that $\arg (\beta - \alpha) = \dir$. Algorithm~\ref{algo:stokes} below computes each such block as a product of a \emph{Borel tranform matrix}, a \emph{connection matrix}, and a \emph{connection-to-Stokes matrix}, which we now define.

\paragraph{Connection matrices.}Let $\alpha, \beta$ be two Stokes values of~$D$. For any path~$\gamma : [0, 1] \rightarrow \mathbb{C}$ avoiding the singular points of~$\Delta$ and such that
\begin{itemize}
  \item $\gamma (0)$ is close enough to~$\alpha$, but not on the local branch cut $\alpha +\mathbb{R}_{\leqslant 0}$,
  
  \item $\gamma (1)$ is close enough to~$\beta$, but not on the local branch cut $\beta +\mathbb{R}_{\leqslant 0}$,
\end{itemize}
analytic continuation along~$\gamma$ provides an isomorphism from $\ensuremath{\operatorname{Sol}}_0 (\Delta_{[\alpha]})$ to $\ensuremath{\operatorname{Sol}}_0 (\Delta_{[\beta]})$. This isomorphism depends on the choice of~$\gamma$. We make the following specific choice: given~$\eta > 0$, we consider a straight-line path~$\gamma$ such that
(see also Figure~\ref{fig:close-triangle}, p.~\pageref{fig:close-triangle})
\begin{itemize}
  \item $\gamma (0) = \alpha + (\eta - i \eta^2)  (\beta - \alpha)$, i.e., $\gamma$ starts close to~$\alpha$, slightly to the right\footnote{The choice of the perturbation factor $\eta - i \eta^2$ instead of the simpler $\eta - i \eta$ one might be tempted to use ensures that the perturbed point is not right on the branch cut of the local argument.} of the ray pointing to~$\beta$,
  
  \item $\gamma (1) = \beta + (\eta - i \eta^2)  (\beta - \alpha)$, i.e., $\gamma$ ends just behind~$\beta$ seen from~$\alpha$, slightly to the right of the same ray.
\end{itemize}
In particular, $\gamma$~passes to the right of any singular point of~$\Delta$ lying on the open line segment~$(\alpha, \beta)$. For small enough~$\eta$, the map from $\ensuremath{\operatorname{Sol}}_0 (\Delta_{[\alpha]})$ to $\ensuremath{\operatorname{Sol}}_0 (\Delta_{[\beta]})$ obtained by analytic continuation along~$\gamma$ does not depend on the choice of~$\eta$.

\begin{definition}
  \label{def:Tmat}We denote by~$\tmat{\alpha}{\beta} \in \mathbb{C}^{\nu \times \nu}$ the matrix of this map in the bases $\localbasis{\alpha}$ and $\localbasis{\beta}$.
\end{definition}

In other words, $\tmat{\alpha}{\beta}$ is the change-of-basis matrix from $\localbasis{\alpha}$ to $\localbasis{\beta}$, both viewed as bases of the space of solutions of~$\Delta$ on a simply connected neighborhood of~$\gamma$.

\paragraph{Borel transform matrices.}

Let $\alpha$ be a Stokes value of~$D$ of multiplicity~$k$.
The formal Borel transformation \eqref{eq:formal-Borel}, p.~\pageref{eq:formal-Borel}, extended by mapping $x^{- j}$ for nonnegative integer~$j$ to zero as explained just after its definition,
defines a map
$ \tilde{\mathcal B}_\alpha :
  \operatorname{FSol}_0 (D_{[\alpha]})
  \to
   \operatorname{Sol}_0 (\Delta_{\alpha})$
between the local solution spaces defined in the previous section.

\begin{definition}
  \label{def:Bmat}We denote by~$\bmat{\alpha} \in \mathbb{C}^{\nu \times k}$ the matrix of this map in the bases~$\mathcal{H}_{[\alpha]}$ and $\localbasis{\alpha}$.
\end{definition}

It results from the discussion in~Sec.~\ref{BorelSol} that, when the exponents of~$D_{[\alpha]}$ all belong to~$\mathbb{C}\backslash\mathbb{Z}_{\leqslant 0}$, one has $k \leqslant \nu$ and the matrix~$\mathbf{B}_{\alpha}$ has full column rank.

\begin{example}[Continued from Example~\ref{ex:running-bases}]
  \label{ex:running-borel} On the problem considered in Example~\ref{ex:running-intro}, the transformed operator~$\Delta$ has order~$\nu = 6$ and all Stokes values have multiplicity~$k = 1$. Since the formal Borel transform of $y_4 (x) = x^3 + \mathrm{O} (x^6)$ is $\zeta^2 / 2 + \mathrm{O} (\zeta^5)$, the Borel transform matrix associated with the Stokes value~$0$ is
  \[ \mathbf{B}_0 = \left( 0, 0, 0, \tfrac{1}{2}, 0, 0 \right)^{\mathrm{T}} . \]
  Similarly, the Borel transform matrix for $\alpha = 6$ is
  \[ \textstyle \mathbf{B}_6 = \left( 0, - \frac{67}{17496}, - \frac{9347}{2519424}, \frac{56135}{22674816}, - \frac{13289119}{29386561536}, - \frac{57551105}{1057916215296} \right)^{\mathrm{T}} \]
  where the entries are the coefficients of~$\zeta^{\lambda}$ for $\lambda = - 3, 0, 1, 2, 3, 4$ (that is, for $\lambda$ ranging over the roots of $\tilde{\Pi}_{[6]}$ enumerated in increasing order) in the expansion of $\tilde{\mathcal{B}}_{\alpha} (y_7)$.
\end{example}

\paragraph{Connection-to-Stokes matrices.}

Fix a Stokes value~$\beta$ of~$D$ of multiplicity~$k'$ and a direction~$\dir \in (-\pi, \pi]$.
We denote by~$\zeta$ the local variable given by $\xi = \beta + \zeta$.
Like in Sec.~\ref{sec:compare}, let~$\gamma_0$ be a Hankel-type path in the $\zeta$-plane, going around the origin positively and with both branches going to infinity just to the left of the direction~$\dir$.

Consider an arbitrary element~$\hat{y} (\zeta)$ of $\ensuremath{\operatorname{Sol}}_0 (\Delta_{[\beta]})$. Initially, $\hat{y} (\zeta)$~is defined in a neighborhood of the origin of the $\zeta$-plane slit along~$\mathbb{R}_{\leqslant 0}$. We still denote by~$\hat{y}$ its analytic continuation starting from small $\zeta>0$ and along~$\gamma_0$. Since, as a solution of~$\Delta_{[\beta]}$, this analytic continuation has at most exponential growth $\rme^{a \Re \zeta}$ at infinity, we have for small enough $x > 0$
\[
\partial \int_{\gamma_0}\mspace{-5mu} \hat{y} (\zeta) \mathrm{e}^{- \zeta / x} \mathrm{d} \zeta = \int_{\gamma_0}\mspace{-5mu} \zeta \hat{y} (\zeta) \mathrm{e}^{- \zeta / x} \mathrm{d} \zeta, \quad
\frac{1}{x}  \int_{\gamma_0}\mspace{-5mu} \hat{y} (\zeta) \mathrm{e}^{- \zeta / x} \mathrm{d} \zeta = \int_{\gamma_0}\mspace{-5mu} \hat{y}' (\zeta) \mathrm{e}^{- \zeta / x} \mathrm{d} \zeta .
\]
Therefore the function
$L_{\beta, \dir} \, \hat y(x) = \int_{\gamma_0} \hat{y} (\zeta) \mathrm{e}^{- \zeta / x} \mathrm{d} \zeta$
is a solution of~$D_{[\beta]}$ defined on the Borel disc with diameter $[0, \rme^{\rmi \omega}/a]$.
Taking its asymptotic expansion as $x \rightarrow 0$ with $\arg x = \dir$, we obtain an element $\tilde L_{\beta, \dir} \, \hat y$ of $\ensuremath{\operatorname{FSol}}_0 (D_{[\beta]})$.

While the function $L_{\beta, \dir} \, \hat y$ depends on~$\dir$, rotating the path~$\gamma_0$ by a small angle only changes $L_{\beta, \dir} \, \hat y$ by an exponentially small amount as $x \to 0$.
It follows that for any two $\dir, \dir'$,
the asymptotic expansion of $L_{\beta, \dir} \, \hat y$ as $x \to 0$ with $\arg x = \dir$
is equal to
the asymptotic expansion of $L_{\beta, \dir'} \, \hat y$ as $x \to 0$ with $\arg x = \dir'$.
In other words, $\tilde L_{\beta, \dir} \, \hat y$ is independent of~$\dir$.
(Another way to see this is that $\tilde L_{\beta, \dir} \, \hat y$ is determined by the initial terms of the local expansion of~$\hat y$ at~$\beta$ using Lemma~\ref{dico}.)
Dropping the unnecessary~$\dir$ from the notation, this construction defines a linear map
$\tilde L_\beta : \operatorname{Sol}_0 (\Delta_{[\beta]}) \to \operatorname{FSol}_0 (D_{[\beta]})$.

\begin{definition}
  \label{def:c2smat}We denote by~$\mathbf{L}_{\beta} {\in \mathbb{C}^{k' \times \nu}} $ the matrix of the map $\tilde L_\beta$ in the bases~$\localbasis{\beta}$ and $\mathcal{H}_{[\beta]}$.
\end{definition}

\begin{example}[Continued from Example~\ref{ex:running-bases}]
  \label{ex:running-c2s} On the problem considered in Example~\ref{ex:running-intro}, one has
  \[ \mathbf{L}_0 = (0, 0, 4 \pi \mathrm{i}, 0, 0, 0), \]
  expressing that the only entry of $\localbasis{0}$ for which the integral in Definition~\ref{def:c2smat} does not vanish is the one starting with $x^2 \log (x)$, and the value of the integral is then~$4 \pi \mathrm{i} (x^3 + \cdots)$. Similarly, one has $\mathbf{L}_6 = (\pi \mathrm{i}, 0, 0, 0, 0, 0)$.
\end{example}

See also Example~\ref{ex:polya}, p.~\pageref{ex:polya} for a slightly more complicated example.
The computation of
$\mathbf B_{\alpha}$~and~$\mathbf L_{\beta}$
is discussed in detail in Sec.~\ref{sec:inimaps}.

\paragraph{Factorization of Stokes blocks.}
With these definitions, the block of the Stokes matrix associated with a given pair of Stokes values decomposes as follows.

\begin{lemma} \label{lem:stokes-block}
  Let $\alpha, \beta$ be two distinct Stokes values of~$D$.
  Let $\mathcal{I}, \mathcal{I}'$ be the corresponding blocks of indices in the decomposition~\eqref{blocQ}, p.~\pageref{blocQ}, of the formal fundamental solution. Then the block $C_{\dir}^{[\mathcal{I}' ; \mathcal{I}]}$ of the Stokes matrix of~$D$ in the direction~$\dir = \arg (\beta - \alpha)$ is equal to $\mathbf{L}_{\beta} \mathbf{T}_{\alpha, \beta} \mathbf{B}_{\alpha}$.
\end{lemma}

\begin{proof}
The result is essentially a rephrasing of equation~\eqref{identite2}, p.~\pageref{identite2}, taking into account that, in virtue of Remark~\ref{remprepared}, it is not necessary to assume the absence of monomials of nonpositive integer degree in the solutions.
More precisely,
for a given $\ell \in \mathcal I$,
consider the element $y_\ell$ of $\mathcal Y$.
Since this $y_\ell$ is a solution of the \emph{unprepared} equation, it corresponds to~$\tilde z_\ell$ in the notation of Remark~\ref{remprepared}.
As in the remark, write
$\tilde z(x) = \tilde k(x) \rme^{-\alpha/x}$
and let $\tilde K(x)$ be the series obtained from~$\tilde k(x)$ by dropping the terms that belong to $\mathbb C[1/x]$.

The matrix~$\mathbf{B}_{\alpha}$ maps the coordinates of~$\tilde k(x)$ in the basis $\mathcal{Y}_{\alpha} (x)$ of ${\operatorname{FSol}}_0 (D_{[\alpha]})$ to those of $\hat{K}_{\ell}$ in the basis $\localbasis{\alpha}  (\zeta)$ of ${\operatorname{Sol}}_0 (\Delta_{[\alpha]})$.

The matrix $\mathbf{T}_{\alpha, \beta}$ maps these coordinates to those of $\hat{K}^-_{\ell} (\beta + \zeta)$ in the basis $\localbasis{\beta}  (\zeta)$ of ${\operatorname{Sol}}_0 (\Delta_{[\beta]})$.

Finally, the matrix
$\mathbf{L}_{\beta}$
maps the coordinates of
$\hat{K}^-_{\ell} (\beta + \zeta)$ in the basis $\localbasis{\beta}  (\zeta)$
of
${\operatorname{Sol}}_0 (\Delta_{[\beta]})$
to the coordinates in the basis
$\mathcal{H}_{[\beta]} (x)$
of the expansion as $x \rightarrow 0$ ($\arg x = \dir$) of
$\int_{\gamma_{0}} \hat{K}^- (\alpha + \zeta) \mathrm{e}^{- \zeta / x} \mathrm{d} \zeta$.

By Remark~\ref{remprepared}, these coordinates are exactly the~$C_{\dir}^{[j ; \ell]}$ for $j \in \mathcal{I}'$.
\end{proof}

\begin{example}[Continued from Examples \ref{ex:running-borel}~and~\ref{ex:running-c2s}]
  \label{ex:running-stokes-block} In the Stokes matrix in the direction $\dir = 0$ computed in Example~\ref{ex:running-intro}, the block corresponding to the pair of Stokes values $(\alpha, \beta) = (0, 6)$ reduces to a single entry $c \approx - 221.7025 \mathrm{i}$ appearing in the fourth column (corresponding to $y_4 (x) = e^{- \alpha / x}  (x^{- 2} + \cdots)$) and last row (corresponding to $y_7 (x) = \mathrm{e}^{- \beta / x}  (x^{- 2} + \cdots)$). Lemma~\ref{lem:stokes-block} expresses this~$c$ as a product $c = \pi \mathrm{i} \times a \times \tfrac{1}{2}$ where $a \approx 141.1420$ is the entry in the fourth column (corresponding to the element $\zeta^2 \log (\zeta) + \cdots$ of~$\localbasis{0}$) and first row (corresponding to the element $\zeta^{- 3} + \cdots$ of~$\localbasis{6}$) of the connection matrix~$\tmat{0}{6}$.
\end{example}

See also Example~\ref{ex:hgeom} (p.~\pageref{ex:hgeom}) for several other examples of the decomposition $\mathbf{L}_{\beta} \mathbf{T}_{\alpha, \beta} \mathbf{B}_{\alpha}$, including a case with nonpositive integer exponents.

\section{The main algorithm}\label{sec:algo-main}

With this notation in place, we can now describe an algorithm that takes as input a differential operator and computes interval enclosures of all its Stokes matrices at the origin. The structure of the algorithm is presented in Algorithm~\ref{algo:stokes} below. It first computes all Borel transform matrices~$\mathbf{B}_{\alpha}$, connection matrices~$\mathbf{T}_{\alpha, \beta}$, and connection-to-Stokes matrices~$\mathbf{L}_{\beta}$ separately, and then fills the Stokes matrices with blocks of the form $\mathbf{L}_{\beta} \mathbf{T}_{\alpha, \beta} \mathbf{B}_{\alpha}$.

For readability, Algorithm~\ref{algo:stokes} refers to directions $\dir$ as elements of $(- \pi, \pi]$. However, all directions it uses are of the form $\dir = \arg \theta$ for some $\theta \in \mathbb{K}$, so that they can be manipulated exactly by encoding them by a suitable~$\theta$.

\begin{algorithm}\label{algo:stokes}
  Stokes matrices
  \begin{description}
    \item[Input] An operator $D = D (x^{- 1}, \partial) \in \mathbb{K} [x^{- 1}] [\partial]$.
    
    \item[Output] A set of pairs $\left( \dir, \mathbf{C}_{\dir} \right)$ where $\dir = \arg \theta$ for some $\theta \in \mathbb{K}$ and $\mathbf{C}_{\dir} \in \CComp^{n \times n}$.
  \end{description}
  \begin{enumerate}
    \item Compute the Borel transform~$\Delta = D ( \ddx{\xi}, \xi ) \in \mathbb{K} [\xi] [\ddx{\xi}]$ of~$D$.
    
    \item Compute the set $\Sigma \subset \mathbb{K}$ of singular points of~$\Delta$ and their multiplicities.
    
    \item \label{step:alignments}Compute the set $\Omega = \{ \arg (\beta - \alpha) : \alpha, \beta \in \Sigma, \alpha \neq \beta \}$ and, for each~$\dir \in \Omega$, the set $\mathcal{A} ( \dir ) = \left\{ (\alpha, \beta) : \arg (\beta - \alpha) = \dir \right\}$.
    
    \item \label{step:connect}Call Algorithm~\ref{algo:connect-all} to compute the matrices $\mathbf{T}_{\alpha, \beta}$ for all $(\alpha, \beta) \in \Sigma^2$.
    
    \item For each $\alpha \in \Sigma$:
    
    \begin{enumerate*}
      \item Compute $\Delta_{[\alpha]} \in \mathbb{K} [\zeta] [ \ddx{\zeta} ]$ by replacing $\xi$ by $\alpha + \zeta$ in $\Delta$.
      
      \item Call Algorithm~\ref{algo:borelmat} to compute the matrix~$\mathbf{B}_{\alpha}$.
      
      \item Call Algorithm~\ref{algo:c2s} to compute the matrix~$\mathbf{L}_{\alpha}$.
    \end{enumerate*}
    
    \item For each $\dir \in \Omega$:
    
    \begin{enumerate*}
      \item Initialize $\mathbf{C}_{\dir} \in \CComp^{n \times n}$ to the zero matrix.
      
      \item For each $(\alpha, \beta) \in \mathcal{A} \left( \dir \right)$,
        set the block $\mathbf{C}_{\dir}^{[\mathcal{I}' ; \mathcal{I}]}$ of the matrix $\mathbf{C}_{\dir}$ to $\bfmat{\beta}  \tmat{\alpha}{\beta}  \bmat{\alpha}$, where $\mathcal{I}$~and~$\mathcal{I}'$ are the sets of row and column indices corresponding respectively to $\alpha$ and~$\beta$ (see Sec.~\ref{preparedsol}).
    \end{enumerate*}
    
    \item Return $\left\{ \left( \dir, \mathbf{C}_{\dir} \right) : \dir \in \Omega \right\}$.
  \end{enumerate}
\end{algorithm}

\begin{remark}
  \label{rk:alignments}Step~\ref{step:alignments} can be implemented by first sorting pairs $(\alpha, \beta)$ according to $\dir = \arg (\beta - \alpha)$, and then grouping together pairs with the same~$\dir$. In fact, in view of the computation of connection matrices (see Remark~\ref{rk:connect-all-details} below), it is convenient to represent each $\mathcal{A} \left( \dir \right)$ as a set of lists each containing all Stokes values lying on a certain line in the direction~$\dir$. For instance, in the case of the operator considered in Sec.~\ref{sec:running-example}, one has for the direction $\dir = 0$ three alignments
  \[ [6 \mathrm{e}^{2 \mathrm{i} \pi / 3}, 6 \mathrm{e}^{\mathrm{i} \pi / 3}], \quad [- 6, 0, 6], \quad [6 \mathrm{e}^{- 2 \mathrm{i} \pi / 3}, 6 \mathrm{e}^{- \mathrm{i} \pi / 3}] \]
  encoding five pairs $(\alpha, \beta)$ corresponding to the five nontrivial blocks (each reduced to a single entry) of the Stokes matrix. Most of this computation can be done numerically using interval arithmetic, with (slower) exact comparisons only needed in the presence of arguments very close to each other or to~$\pi$.
\end{remark}

This algorithm correctly computes the Stokes matrices of operators of single level one. More explicitly, this means that the following result holds.

\begin{theorem}
  \label{thm:algo-stokes}Assume that the operator~$D$ given as input to Algorithm~\ref{algo:stokes} is of single level one. The algorithm returns a set $\left\{ \left( \dir, \mathbf{C}_{\dir} \right) : \dir \in \Omega \right\}$ such that
  \begin{enumerate}
    \item \label{item:all-sing-dir}$\Omega$~is the set of anti-Stokes directions of~$D$;
    
    \item \label{item:equalsstokes}each $\mathbf{C}_{\dir}$ is a matrix of computable numbers such that ${I +\mathbf{C}_{\dir}}$ is the Stokes matrix of~$D$ in the direction~$\dir$, relative to the fundamental matrix~$\mathcal Y$.
  \end{enumerate}
\end{theorem}

\begin{remark}[Interval version and error bounds]
  Using interval arithmetic instead of computable complex numbers, either the algorithm raises an error or its output satisfies item~(\ref{item:all-sing-dir}) above and
  \begin{enumerate}
    \stepcounter{enumi}
    \item \label{item:containment}for all $\dir \in \Omega$ and $1 \leqslant j, \ell \leqslant n$, the entry $\mathbf{C}_{\dir}^{[j ; \ell]} \in \Balls$ of $\mathbf{C}_{\dir}$ contains $C_{\dir}^{[j ; \ell]}$, where $I + C_{\dir}$ is the Stokes matrix of~$D$ in the direction~$\dir$, relative to the fundamental matrix~$\mathcal Y$.
  \end{enumerate}
  Additionally, for a fixed input operator~$D$, the values of the working precision for which the algorithm can fail are bounded, and the diameters of all $\mathbf{C}_{\dir}^{[j ; \ell]}$ tend to zero as the working precision tends to infinity.
\end{remark}

\begin{proof}
The correctness of Algorithms~\ref{algo:borelmat}, \ref{algo:c2s}, and~\ref{algo:connect-all} is proved respectively in Propositions~\ref{prop:algo-borel}, \ref{prop:algo-c2s}, and~\ref{prop:algo-connect-all} below. We have seen in Sec.~\ref{sec:Stokes-Laplace} that $\Omega$ is the set of anti-Stokes directions of~$D$. Now fix~$\dir \in \Omega$. As discussed in Sec.~\ref{Stokes matrices}, the block $\mathbf{C}_{\dir}^{[\mathcal{I}' ; \mathcal{I}]}$ is zero when $\alpha = \beta$ and when $\alpha \neq \beta$ with $(\alpha, \beta) \not\in \mathcal{A} \left( \dir \right)$. Finally, Lemma~\ref{lem:stokes-block} states that the blocks with $(\alpha, \beta) \in \mathcal{A} \left( \dir \right)$ are correctly computed.
\end{proof}

Unlike the number of Borel transform matrices and connection-to-Stokes matrices, the number of connection matrices grows quadratically with the number~$N$ of Stokes values. However, many of these matrices can be deduced from one another, so that the number of connection matrices that need to be computed \emph{by numerical analytic continuation} is only $\mathrm{O} (N)$, as discussed in Sec.~\ref{sec:connection} below. In our experience, the computation of these $\mathrm{O} (N)$ core transition matrices is the most computationally expensive step in practice.

\section{Borel transform and connection-to-Stokes matrices}\label{sec:inimaps}

In this section, we discuss the computation of the matrices $\mathbf{B}_{\alpha}$~and~$\mathbf{L}_{\alpha}$ introduced in Sec.~\ref{sec:matrices}. These parts of the algorithm respectively implement the formal Borel transformation of solutions of~$D$ and the application of connection-to-Stokes formulae to the analytic continuation of solutions of~$\Delta$. The main difference with the sketch of these operations in Part~\ref{part:theory} is that we now work explicitly with coordinates in the bases $\mathcal{H}_{[\alpha]}$~and~$\localbasis{\alpha}$.

\subsection{Borel transform matrices}\label{sec:borel}

Recall that we have defined $\mathbf{B}_{\alpha}$ as the matrix of the Borel transformation, with monomials $x^{- j}$ with $j \in \mathbb{Z}_{\geqslant 0}$ mapped to zero, viewed as a linear map
\[ \tilde{\mathcal{B}}_{\alpha} : \ensuremath{\operatorname{FSol}}_0 (D_{[\alpha]}) \rightarrow \ensuremath{\operatorname{Sol}} (\Delta_{[\alpha]}) . \]
To compute this matrix, we need to determine the coordinates in the destination basis~$\localbasis{\alpha}$ of the Borel transform of each element of the source basis~$\mathcal{H}_{[\alpha]}$. By our choice of $\localbasis{\alpha}$, these coordinates are the coefficients of the monomials ${\zeta^{\lambda + \mu}}' \log (\zeta)^{\rho'}$ for $(\lambda, \mu', \rho') \in \mathcal{E} (\Delta_{[\alpha]})$. We have seen in Sec.~\ref{PrelimBorel}, p.~\pageref{eq:formal-Borel}, that the naive formal Borel transform of a monomial is given by
\begin{equation}
  \tilde{\mathcal{B}} (x^{\lambda + 1} \ln (x)^m) = \frac{\mathrm{d}^m}{\mathrm{d} \lambda^m}  \frac{\zeta^{\lambda}}{\Gamma (\lambda + 1)}, \label{eq:recall-Borel}
\end{equation}
interpreted to give~$0$ when $\lambda \in \mathbb{Z}_{< 0}$ and $m = 0$. Thus, given $y \in \ensuremath{\operatorname{FSol}}_0 (D_{[\alpha]})$, the coefficient of a monomial ${\zeta^{\lambda + \mu}}' \log (\zeta)^{\rho'}$ in the series~$\tilde{\mathcal{B}}_{\alpha} (y)$ only depends on a finite number of terms of~$y$, and we can obtain the coordinates of~$\tilde{\mathcal{B}}_{\alpha} (y)$ in the basis~$\localbasis{\alpha}$ by computing those terms of~$y$ and applying formula~\eqref{eq:recall-Borel}. The procedure, which systematizes the one we applied in Example~\ref{ex:running-borel}, is detailed in Algorithm~\ref{algo:borelmat}.

\begin{algorithm}
  Computation of~$\bmat{\alpha}$
  \begin{description}
    \label{algo:borelmat}\item[Input] Two operators $D \in \mathbb{K} [x^{- 1}] [\partial]$ and $\Delta \in \mathbb{K} [\zeta] [ \ddx{\zeta} ]$.
    
    \item[Output] A matrix $\mathbf{B} \in \CComp^{\nu \times k}$ where $\nu =\ensuremath{\operatorname{ord}} (\Delta)$ and $k = \dim \operatorname{FSol} (D)$.
  \end{description}
  \begin{enumerate}
    \item \label{step:borel:struct-D}Compute the structure $\mathcal{E} (D)$ of the local solutions of~$D$ free of exponentials (\cf Definition~\ref{def:structure}).
    
    \item \label{step:borel:struct-Delta}Compute the structure $\mathcal{E} (\Delta)$ of the local solutions of~$\Delta$.
    
    \item Initialize $\mathbf{B} \in \CComp^{\nu \times k}$ to the zero matrix.
    
    \item \label{step:borel:main-loop}For each triple $(\mathcal{L}, \mu, \rho) \in \mathcal{E} (D)$:
    
    \begin{enumerate*}
      \item \label{step:borel:offset}If $\mathcal{L} \not\in \mathbb{Z}$, set $\lambda =\mathcal{L}- 1$ and $\delta = 0$. If $\mathcal{L} \in \mathbb{Z}$ and there is an integer~$\lambda$ such that $(\lambda, 0, 0) \in \mathcal{E} (\Delta)$, then set~$\lambda$ to the unique integer with that property and set $\delta = \lambda -\mathcal{L}+ 1$. Otherwise, continue to the next value of~$\mathcal{L}$.
      
      \item \label{step:borel:expand-sol}Consider the unique solution
      \[ y (x) = x^{\mathcal{L}}  \sum_{r = 0}^{s - 1} \sum_{m \geqslant 0} c_{r, m} x^m \ln (x)^r \in \ensuremath{\operatorname{FSol}} (D) \]
      such that $c_{\rho, \mu} = 1$ and $c_{r, m} = 0$ for $(\mathcal{L}, m, r) \in \mathcal{E} (D) \backslash \{ (\mathcal{L}, \mu, \rho) \}$. Compute the coefficients~$c_{r, \delta + \mu'} \in \CComp$ for all $0 \leqslant r < s$ and $\mu' \geqslant \mu - \delta$ such that $(\lambda, \mu', 0) \in \mathcal{E} (\Delta)$ (\cf~Sec.~\ref{sec:local-bases}).
      
      \item \label{step:borel:inner-loop}For each $(\mu', \rho')$ such that $(\lambda, \mu', \rho') \in \mathcal{E} (\Delta)$ and $\mu' \geqslant \mu - \delta$:
      
      \begin{enumerate*}
        \item Compute the coefficients $\gamma_0, \ldots, \gamma_{s - 1 - \rho'}$ of the expansion
        \begin{equation}
          \frac{1}{\Gamma (\lambda + \mu' + 1 + X)} = \gamma_0 + \gamma_1 X + \cdots + \gamma_{s - 1 - \rho'} X^{s - 1 - \rho'} + \mathrm{O} (X^{s - \rho'}) . \label{eq:ser-rgamma}
        \end{equation}
        \item Set the entry of $\mathbf{B}$ at row~$(\lambda, \mu', \rho')$ and column~$(\mathcal{L}, \mu, \rho)$ to
        \[ \sum_{r = \rho'}^{s - 1} \frac{r!}{\rho' !} \gamma_{r - \rho'} c_{r, \delta + \mu'} . \]
      \end{enumerate*}
    \end{enumerate*}
    
    \item Return $\mathbf{B}$.
  \end{enumerate}
\end{algorithm}

\begin{remark}
  \label{rk:Kgamma}
  \begin{enumerate}
    \item Steps \ref{step:borel:struct-D}, \ref{step:borel:struct-Delta}, and \ref{step:borel:offset} involve exact computations on elements of~$\mathbb{K}$. We present the remaining steps as computations in~$\CComp$ because they \emph{can} be performed numerically. However, the entries of~$\mathbf{B}_{\alpha}$ belong to the field extension $\mathbb{K} (\mathbf{\Gamma})$ where
    \[ \mathbf{\Gamma}= \{ (1 / \Gamma)^{(i)} (a) : a \in \mathbb{K}, i \in \mathbb{N} \} \]
    and could also be computed exactly in terms of the elements of~$\mathbf{\Gamma}$.
    
    \item When $\lambda + \mu'$ is a negative integer and $\rho' = 0$, one has $\gamma_0 = 0$ in~\eqref{eq:ser-rgamma}. If, additionally, $\delta + \mu'$ is equal to~$\mu$, then by definition of the series~$y (x)$ one has $c_{r, \delta + \mu'} = 0$ for all $r > \rho'$, and hence the corresponding iteration of the loop~\eqref{step:borel:inner-loop} can be skipped.
  \end{enumerate}
\end{remark}

\begin{remark}
  \label{rk:borel:redundancies}As presented here, the algorithm contains several easy-to-avoid redundancies.
  \begin{enumerate}
    
    \item For simplicity, we repeat the computation for each pair $(\mu', \rho')$. When $\lambda + 1 + \mu'$ is a multiple exponent of~$\Delta$, though, one can reuse for all pairs $(\mu', \rho')$ with $\rho' > 0$ the coefficients~$\gamma_i$ computed for the pair $(\mu', 0)$.
    
    \item The series expansion of $\Gamma (\lambda + \mu' + 1 + X)$ for $\mu' \neq 0$ can be deduced from that of $\Gamma (\lambda + 1 + X)$ using the functional equation for the gamma function.
  \end{enumerate}
\end{remark}

\begin{proposition}
  \label{prop:algo-borel}Algorithm~\ref{algo:borelmat}, called with $D$ set to $D_{[\alpha]}$ and $\Delta$ set to $\Delta_{[\alpha]}$, computes the Borel transformation matrix $\mathbf{B}_{\alpha}$ defined in Sec.~\ref{sec:matrices}.
\end{proposition}

\begin{proof}
As noted above, computing~$\mathbf{B}_{\alpha}$ is equivalent to computing the coefficients of ${\zeta^{\lambda + \mu}}' \log (\zeta)^{\rho'}$ in $\tilde{\mathcal{B}}_{\alpha} (y)$ for each $y \in \mathcal{Y}_{\alpha}$ and $(\lambda, \mu', \rho') \in \mathcal{E} (\Delta_{[\alpha]})$. Each iteration of the main loop deals with an element of~$\mathcal{Y}_{\alpha}$ and fills the corresponding column of~$\mathbf{B}_{\alpha}$. Since the indicial polynomial of $\Delta_{[\alpha]}$ has degree~$\nu$, the set~$\mathcal{E} (\Delta_{[\alpha]})$ has cardinality~$\nu$, leading to a $\nu \times k$ matrix.

Fix $(\mathcal{L}, \mu, \rho) \in \mathcal{E} (D_{[\alpha]})$, and let $y \in \mathcal{Y}_{\alpha}$ be the corresponding basis element. Due to the choice of~$\mathcal{Y}_{\alpha}$, the series~$y$ belongs to $x^{\mathcal{L}} \mathbb{C} [[x]] [\ln x]$. Its image~$\tilde{\mathcal{B}}_{\alpha} (y)$ is defined as the sum of the Borel transforms of the terms of~$y$ (minus the polar and constant parts). It follows from~\eqref{eq:recall-Borel} that $\mathcal{B} (x^{\lambda + 1} \ln (x)^m)$ is of the form $\zeta^{\lambda} P (\ln \zeta)$ for some polynomial~$P$ of degree at most $m$, and hence $\tilde{\mathcal{B}}_{\alpha} (y)$ lies in $\zeta^{\mathcal{L}+ \mu - 1} \mathbb{C} [[\zeta]] [\ln \zeta]$. In particular, the only triples $(\lambda, \mu', \rho')$ in $\mathcal{E} (\Delta_{[\alpha]})$ that can correspond to a nonzero coordinate are those with $\lambda -\mathcal{L} \in \mathbb{Z}$. There is at most one\footnote{We note in passing that the case where $\mathcal{L} \in \mathbb{Z}$ but there is no suitable $\lambda \in \mathbb{Z}$, where the algorithm directly continues to the next iteration, may only occur when $k > \nu$ and the element $y (x)$ of $\mathcal{Y}_{\alpha}$ of index $(\mathcal{L}, \mu, \rho)$ considered at step~\ref{step:borel:expand-sol} is a trivial solution $x^{- j}$, $j \in \mathbb{Z}_{\geqslant 0}$. Even when $k > \nu$ and $y (x)$ is a trivial solution, though, it may also happen that $\Delta$~has an integer exponent, in which case the fact that $y (x)$ does not contribute to the Stokes matrix is only detected later in the loop body.}~$\lambda$ with this property. Step~\ref{step:borel:offset} of the algorithm computes this~$\lambda$, using the fact that, when~$\lambda \not\in \mathbb{Z}$, one has $\lambda =\mathcal{L}- 1$ by equation \eqref{eq:ind-keqnu}~or~\eqref{eq:ind-kltnu}. After Step~\ref{step:borel:offset}, one has $\lambda + 1 =\mathcal{L}+ \delta$.

The consequence of equation~\eqref{eq:recall-Borel} noted above also implies that the coefficient of $\zeta^{\lambda + \mu'} \log (\zeta)^{\rho'}$ in $\tilde{\mathcal{B}}_{\alpha} (y)$ only depends on the coefficients of $x^{\mathcal{L}+ m} \ln (x)^{\rho}$ in~$y$ where $\mathcal{L}+ m = \lambda + \mu' + 1$ and $\rho \geqslant \rho'$. The coordinates of $\tilde{\mathcal{B}}_{\alpha} (y)$ in the basis~$\localbasis{\alpha}$ are thus determined by the coefficients of $x^{\mathcal{L}+ \delta + \mu'} \ln (x)^r$ in~$y$ for $(\lambda, \mu', \rho') \in \mathcal{E} (\Delta_{[\alpha]})$ and $r \geqslant \rho'$ such that $\delta + \mu' \geqslant \mu$. These coefficients are among the ones computed at step~\ref{step:borel:expand-sol}, since, if $(\lambda, \mu', \rho')$ belongs to $\mathcal{E} (\Delta_{[\alpha]})$, then $(\lambda, \mu', 0)$ too.

Then, the inner loop iterates over the potentially nonzero coordinates of $\tilde{\mathcal{B}}_{\alpha} (y)$. By~\eqref{eq:recall-Borel}, the contribution of $\mathcal{B} (x^{\mathcal{L}+ \delta + \mu'} \ln (x)^r)$ to the coordinate of index $(\lambda, \mu', \rho') \in \mathcal{E} (\Delta_{[\alpha]})$ is the coefficient of $\zeta^{\lambda + \mu'} \ln (\zeta)^{\rho'}$ in
\begin{align*}
  \frac{\mathrm{d}^r}{\mathrm{d} \lambda^r}  \frac{\zeta^{\lambda + \mu'}}{\Gamma (\lambda + \mu' + 1)} & = \sum_{i = 0}^r \binom{r}{i} \zeta^{\lambda + \mu'} \ln (\zeta)^i  \frac{\mathrm{d}^{r - i}}{\mathrm{d} \lambda^{r - i}}  \frac{1}{\Gamma (\lambda + \mu' + 1)}\\
  & = \zeta^{\lambda + \mu'}  \sum_{i = 0}^r \frac{r!}{i!} \gamma_{r - i} \ln (\zeta)^i .
\end{align*}
The body of the loop computes the sum of these contributions for all~$r \geqslant \rho'$ and stores the result in the corresponding entry of the matrix.
\end{proof}

\subsection{Connection-to-Stokes matrices}\label{sec:c2s}

The computation of the Connection-to-Stokes matrices is similar, with the roles of $D_{[\beta]}$~and~$\Delta_{[\beta]}$ interchanged, and Lemma~\ref{lemmeconn} providing the coefficients. We have seen a simple example in Example~\ref{ex:running-c2s}.

\begin{algorithm}
  Computation of $\bfmat{\beta}$
  \begin{description}
    \label{algo:c2s}\item[Input] Two operators $\Delta \in \mathbb{K} [\zeta] [ \ddx{\zeta} ]$ and $D \in \mathbb{K} [x^{- 1}] [\partial]$.
    
    \item[Output] A matrix $\mathbf{L} \in \CComp^{k \times \nu}$ where $k = \dim \ensuremath{\operatorname{FSol}} (D)$ and $\nu =\ensuremath{\operatorname{ord}} (\Delta)$.
  \end{description}
  \begin{enumerate}
    \item Compute the structure $\mathcal{E} (D)$ of the local solutions of~$D$ free of exponentials (\cf~Sec.~\ref{def:structure}).
    
    \item \label{step:c2s:struct-Borel}Compute the structure $\mathcal{E} (\Delta)$ of the local solutions of~$\Delta$.
    
    \item Initialize $\mathbf{L} \in \CComp^{k \times \nu}$ to the zero matrix.
    
    \item \label{step:c2s:mainloop}For each triple $(\lambda, \mu, \rho) \in \mathcal{E} (\Delta)$:
    
    \begin{enumerate*}
      \item \label{step:c2s:offset}If $\lambda \not\in \mathbb{Z}$, then set $\mathcal{L}= \lambda + 1$ and $\delta = 0$. If $\lambda \in \mathbb{Z}$ and there is an integer~$\mathcal{L}$ such that $(\mathcal{L}, 0, 0) \in \mathcal{E} (D)$, then set~$\mathcal{L}$ to the unique integer with that property and set $\delta =\mathcal{L}- \lambda - 1$. Otherwise, continue to the next value of~$\lambda$.
      
      \item \label{step:c2s:expand-sol}Consider the unique solution
      \[ \hat{y} (\zeta) = \zeta^{\lambda}  \sum_{r = 0}^{s - 1} \sum_{m \geqslant 0} c_{r, m} \zeta^m \ln (\zeta)^r \in \ensuremath{\operatorname{Sol}}_0 (\Delta) \]
      such that $c_{\rho, \mu} = 1$ and $c_{r, m} = 0$ for $(\lambda, m, r) \in \mathcal{E} (\Delta) \backslash \{ (\lambda, \mu, \rho) \}$. Compute the coefficients~$c_{r, \delta + \mu'}$, $0 \leqslant r < s$, for all $\mu' \geqslant \mu - \delta$ such that $(\mathcal{L}, \mu', 0) \in \mathcal{E} (D)$ (\cf~Sec.~\ref{sec:local-bases}).
      
      \item For each $(\mu', \rho')$ such that $(\mathcal{L}, \mu', \rho') \in \mathcal{E} (D)$ and $\mu' \geqslant \mu$:
      
      \begin{enumerate*}
        \item \label{step:c2s:expand}Compute the coefficients $a_0, \ldots, a_{s - 1 - \rho'}$ of the expansion
        \begin{equation}
          \frac{\mathrm{e}^{- i \pi (\mathcal{L}+ \mu' + X)}}{\Gamma (1 -\mathcal{L}- \mu' - X)} = a_0 + a_1 X + \cdots + a_{s - 1 - \rho'} X^{s - 1 - \rho'} + \mathrm{O} (X^{s - \rho'}) . \label{eq:c2s-series}
        \end{equation}
        \item Set the entry of $\mathbf{L}$ at row~$(\mathcal{L}+ \mu', \rho')$ and column $(\lambda + \mu, \rho)$ to
        \[ 2 \pi i \sum_{r = \rho'}^{s - 1} \frac{r!}{\rho' !} a_{r - \rho'} c_{r, \delta + \mu'} . \]
      \end{enumerate*}
    \end{enumerate*}
    
    \item Return $\mathbf{L}$.
  \end{enumerate}
\end{algorithm}

\begin{remark}
  Step~\ref{step:c2s:expand} may deserve some elaboration. In the formula~\eqref{eq:c2s-series}, both $\mathrm{e}^{- i \pi (1 + \lambda + X)}$ and $1 / \Gamma (- \lambda - X)$ are analytic functions of~$X$ for every~$\lambda$. The required series expansion can be computed by computing the expansions to order $s - 1 - \rho'$ of these two functions (as truncated series with coefficients in~$\CComp$) and multiplying them together. We refer to Johansson~\cite{Johansson2023} for a thorough discussion of the rigorous arbitrary-precision computation of the reciprocal gamma function and its derivatives.
  
  As pointed out to us by F.~Johansson, our approach effectively reduces the computation of Hankel integrals of solutions of linear differential equations to evaluations of the gamma function, but some methods used to evaluate the gamma function are in turn based on the numerical computation of Hankel integrals of solutions of differential equations~(\emph{e.g.}, \cite{SchmelzerTrefethen2007})---and others on that of connection constants~\cite[Sec.~6]{Johansson2023}. This may hint at a more direct approach to the approximation of Stokes multipliers that would be at least as efficient as the present one. As it is, though, state-of-the-art algorithms for the direct evaluation of the Hankel integrals we are computing have slightly worse complexity than the algorithm going through the gamma function when all parameters except the precision are fixed (see Remark~\ref{rk:complexity} below). From a practical perspective, being able to take advantage of existing well-optimised implementations of~$\Gamma$ and related functions is highly beneficial.
\end{remark}

\begin{remark}
  Comments analogous to Remark~\ref{rk:borel:redundancies} apply here: $\mathcal{E} (\Delta)$ can be deduced from $\mathcal{E} (D)$; the series expansion~\eqref{eq:c2s-series} only needs to be computed once for each $(\mathcal{L}, \mu')$, and part of that computation can be shared among all iterations corresponding to the same~$\mathcal{L}$ thanks to the functional equation for the gamma function. Additionally, the value of $\Gamma (- \lambda - \mu' - X)$ does not depend on~$\dir$.
\end{remark}

\begin{proposition}
  \label{prop:algo-c2s}Algorithm~\ref{algo:c2s}, called with $\Delta$ set to $\Delta_{[\beta]}$ and $D$ set to $D_{[\beta]}$, computes the connection-to-Stokes matrix~$\mathbf{L}_{\beta}$ defined in Sec.~\ref{sec:matrices}.
\end{proposition}

\begin{proof}
The proof is similar to that of Proposition~\ref{prop:algo-borel}.  The algo\-rithm~fills each column of~$\mathbf{L}$ with the coefficients of $x^{\mathcal{L}+ \mu'} \ln (x)^{\rho'}$ where ${(\mathcal{L}, \mu', \rho') \in \mathcal{E} (D_{[\beta]})}$ in the asymptotic expansion as $x \rightarrow 0$ in the direction $\arg x = \dir$, for a given element~$\hat{y}$ of~$\localbasis{\beta} $, of the integral
\[ \int_{\gamma_{0}} \hat{y} (\zeta) \mathrm{e}^{- \zeta / x} \mathrm{d} \zeta, \]
where $\gamma_{0}$ is the same Hankel type path (depending on~$\dir$) as in equation~\eqref{StokesLaplace}.

Fix $(\lambda, \mu, \rho) \in \mathcal{E} (\Delta_{[\beta]})$ and consider the corresponding~$\hat{y}$.

When $\lambda$~is an integer and $\lambda + \mu +\mathbb{N}$ contains no exponent of~$\Delta_{[\beta]}$ except possibly trivial ones (cf.\ p.~\pageref{par:trivial-exponents}), the function $\hat{y}$ is analytic at~$0$ and the integral vanishes, so that the whole column is zero. This is detected at step~\ref{step:c2s:offset} and in this case the loop iteration is skipped. Otherwise, step~\ref{step:c2s:offset} sets~$\delta \in \mathbb{Z}$ so that $\mathcal{L}= \lambda + 1 + \delta$ (beware that the relation between $\mathcal{L}$, $\lambda$, and $\delta$ is not the same as in the proof of Proposition~\ref{prop:algo-borel}). Integrating the series expansion of~$\hat{y} (\zeta)$ termwise yields the asymptotic expansion
\begin{equation}
  \int_{\gamma_{0}} \hat{y} (\zeta) \mathrm{e}^{- \zeta / x} \mathrm{d} \zeta \sim \sum_{r = 0}^{s - 1} \sum_{m \geqslant \mu} c_{r, m}  \int_{\Gamma_{\dir}} \zeta^{\lambda + m} \ln (\zeta)^r \mathrm{e}^{- \zeta / x} \mathrm{d} \zeta, \label{eq:c2s:expansion}
\end{equation}
where, by Lemma~\ref{lemmeconn}, one has
\begin{align*}
  \int_{\Gamma_{\dir}} \zeta^{\lambda + m} \ln (\zeta)^r \mathrm{e}^{- \zeta / x} \mathrm{d} \zeta & = \diff{\lambda}{r}  \left( \frac{2 \pi \mathrm{i}}{\Gamma (- \lambda - m)} \mathrm{e}^{- \mathrm{i} \pi (\lambda + m + 1)} x^{\lambda + m + 1} \right)\\
  & = 2 \pi \mathrm{i} \sum_{j = 0}^r \binom{r}{j}  \left( \diff{\lambda}{r - j}  \frac{\mathrm{e}^{- \mathrm{i} \pi (\lambda + m + 1)}}{\Gamma (- \lambda - m)} \right) x^{\lambda + m + 1} \ln (x)^j
\end{align*}
for $\arg x = \dir$.

The only terms of the expansion~\eqref{eq:c2s:expansion} contributing to the coefficient of $x^{\mathcal{L}+ \mu'} \ln (x)^{\rho'}$ in the result are those with $\lambda + m + 1 =\mathcal{L}+ \mu'$, that is, $m = \delta + \mu'$, and $r \geqslant \rho'$. The corresponding coefficients~$c_{r, m}$ are computed at step~\ref{step:c2s:expand-sol}. Then the inner loop computes the cofactors
\[ a_{r - j} = \frac{1}{(r - j) !}  \diff{\lambda}{r - j}  \frac{\mathrm{e}^{- \mathrm{i} \pi (\lambda + m + 1)}}{\Gamma (- \lambda - m)} \]
and collects the terms
\[ 2 \pi \mathrm{i} c_{r, m}  \binom{r}{j}  [(r - j) !a_{r - j}] x^{\lambda + m + 1} \ln (x)^j \]
with $m = \delta + \mu'$ and $j = \rho'$.
\end{proof}

See Example~\ref{ex:polya} (p.~\pageref{ex:polya}) for an example of the computation of~$\bfmat{\beta}$ in the presence of logarithms.

\section{Computing all connection matrices}\label{sec:connection}

We turn to the computation of connection matrices, as required by step~\ref{step:connect} of Algorithm~\ref{algo:stokes}. This is the main step of the algorithm from several perspectives. In particular, it is the most computationally expensive step in practice, and the only one that may introduce transcendental constants not expressible in terms of values at points in~$\mathbb{K}$ of the functions $\mathrm{e}^x$, $\mathrm{e}^{\mathrm{i} \pi x}$, and $(1 / \Gamma)^{(i)} (x)$.

Recall that the connection matrix~$\tmat{\alpha}{\beta}$ is defined as the matrix of the analytic continuation map from~$\alpha$ to~$\beta$, along a path~$\gamma$ specified in Sec.~\ref{sec:matrices}, of solutions of the operator~$\Delta$, expressed in the fixed bases of solutions $\localbasis{\alpha}$ and $\localbasis{\beta}$. In principle, $\tmat{\alpha}{\beta}$~can be computed using the procedure sketched in Sec.~\ref{anacont}.
This procedure can also be viewed as a way of realising a rigorous ODE solver supporting generalised initial values at regular singular points (and any other ODE solver satisfying these requirements could also be used in its place).
Numerical methods for differential equations that follow this general scheme are called Taylor methods.
Efficient algorithms are available for computing the sums of series solutions of linear differential equations with rational coefficients (\emph{e.g.}, \cite{BostanChyzakOllivierSalvySchostSedoglavic2007,ChudnovskyChudnovsky1990,vanderHoeven2010}), leading to Taylor methods well suited to arbitrary-precision computations. The error bounds needed to make the computation rigorous can be derived, essentially, from classical proofs of the Cauchy existence theorem for solutions of differential equations; see~\cite{Mezzarobba2019} for details and tighter bounds. Our code uses the ODE solver available as part of \codestar{ore\_algebra}~\cite{Mezzarobba2016}, which implements a method of this family.

Algorithm~\ref{algo:stokes} requires us to compute~$\tmat{\alpha}{\beta}$ for all pairs $(\alpha, \beta)$ of distinct Stokes values. Using the above method separately for each pair leads however to highly redundant computations. Indeed, given three points $p, q, r$, analytic continuation from $p$~to~$r$ is the same as analytic continuation from $p$~to~$q$ and then from $q$~to~$r$, up to correcting factors associated with the singular points crossed when deforming the corresponding paths into one another. When the triangle $p \, q \, r$ does not contain any singular point in its interior or the interior of its edges, the correcting factors are simply \emph{local monodromy matrices} around some of the vertices. We now give details of an algorithm that uses numerical ODE solving to compute the matrices~$\tmat{\alpha}{\beta}$ for $(\alpha, \beta)$ ranging over a \emph{spanning tree} of~$\Sigma$, and recovers all remaining connection matrices by composing known ones and applying correcting factors as appropriate.

\paragraph{Local monodromy matrices.}Let $\alpha$ be a Stokes value of~$D$. The monodromy map sending a local solution of~$\Delta_{[\alpha]}$ at the origin to its analytic continuation along a simple positive loop around the origin is a linear endomorphism of $\ensuremath{\operatorname{Sol}}_0 (\Delta_{[\alpha]})$.

\begin{definition}
  The \emph{local monodromy matrix} at~$\alpha \in \mathbb{C}$ is the matrix $\mmat{\alpha} \in \mathbb{C}^{\nu \times \nu}$ in the basis~$\localbasis{\alpha}$ of the monodromy map around~$\alpha$.
\end{definition}

Since~$\alpha$ is a regular singular point of~$\Delta_{[\alpha]}$, the local monodromy at~$\alpha$ coincides with the formal monodromy.
It can be computed from the exponent structure~$\mathcal{E} (\Delta_{[\alpha]})$ and selected monomials of the series~$\localbasis{\alpha}$ by a reasoning similar to that of Sec.~\ref{sec:inimaps}, but simpler.

\begin{lemma}
  \label{lem:local-monodromy}The entry at row $(\lambda, \mu, \rho) \in \mathcal{E} (\Delta_{[\alpha]})$ of the column of~$\mmat{\alpha}$ associated with an element
  \begin{equation}
    \hat{y}_{\alpha, i} (\zeta) = \zeta^{\lambda}  \sum_{r = 0}^{s - 1} \sum_{m \geqslant 0} c_{r, m} \zeta^m \ln (\zeta)^r \in \localbasis{\alpha} \label{eq:monodromy-sol}
  \end{equation}
  of the local basis at~$\alpha$ is given by
  \begin{equation}
    \mathrm{e}^{2 \pi i \lambda}  \sum_{\delta = 0}^{s - 1 - \rho} c_{\rho + \delta, \mu}  \binom{\rho + \delta}{\delta}  (2 \pi i)^{\delta} . \label{eq:local-monodromy}
  \end{equation}
  The entries corresponding to exponents not in $\lambda +\mathbb{Z}$ are zero.
\end{lemma}

\begin{proof}
The image of $\hat{y}_{\alpha, i}$ by the monodromy map is equal to
\[ \mathrm{e}^{2 \pi i \lambda} \zeta^{\lambda}  \sum_{r = 0}^{s - 1} \sum_{m \geqslant 0} c_{r, m} \zeta^m  (\ln (\zeta) + 2 \pi i)^r . \]
Equation~\eqref{eq:local-monodromy} follows by extracting the coefficient of $\zeta^{\lambda + m} \ln (\zeta)^{\rho}$.
\end{proof}

\paragraph{Composition rules.}Let $\Sigma$ denote the set of singular points of~$\Delta$.

\begin{definition}[Void triangle]
  \label{def:vacant} We say that a triangle $p \, q \, r$ is \emph{void} when its vertices do not lie on a single line and its closure does not intersect $\Sigma \backslash \{ p, q, r \}$.
\end{definition}

Given a void triangle $p \, q \, r$, we need to understand how to compute $\tmat{p}{r}$ from $\tmat{p}{q}$ and $\tmat{q}{r}$. For this it is convenient to introduce the following geometric predicates.

\begin{definition}[Bottom-to-top order]
  For two points $\alpha, \beta \in \mathbb{C}$, we write $\alpha \bttlt \beta$ when
  $\ensuremath{\operatorname{Im}} (\alpha) <\ensuremath{\operatorname{Im}} (\beta)$
  or
  $\ensuremath{\operatorname{Im}} (\alpha) =\ensuremath{\operatorname{Im}} (\beta) \wedge \ensuremath{\operatorname{Re}} (\alpha) >\ensuremath{\operatorname{Re}} (\beta)$.
\end{definition}

Equivalently, we have $\alpha \bttlt \beta$ when $\alpha \neq \beta$ and $0 < \arg (\beta - \alpha) \leqslant \pi$. The definition in the case of $\ensuremath{\operatorname{Im}} (\alpha) =\ensuremath{\operatorname{Im}} (\beta)$ is motivated by the fact that,
with the convention $-\pi < \arg(\xi - \alpha) \leq \pi$,
a point~$\beta$ that lies to the left of~$\alpha$ is `above' the local branch cut.
The relation~$\bttlt$ is a strict total order.

\begin{definition}[Orientation]
  For any three points $p, q, r \in \mathbb{C}$ not all lying on a same line, set $\ensuremath{\operatorname{orient}} (p, q, r) = + 1$ when $p, q, r$ are ordered counterclockwise around a point of the interior of their convex hull, that is, when $r$~lies to the left of the oriented segment~$(p, q)$, and $\ensuremath{\operatorname{orient}} (p, q, r) = - 1$ otherwise.
\end{definition}

Transition matrices associated with the edges of a void triangle can be composed according to the following rule.

\begin{lemma}
  \label{lem:close-triangle}Let $t = p \, q \, r$ be a void triangle. The connection matrices along the edges of~$t$ satisfy
  \begin{equation}
    \tmat{p}{r} =\mathbf{V}_r \mathbf{T}_{q, r} \mathbf{V}_q   \tmat{p}{q} \mathbf{V}_p \label{eq:compose-tmat}
  \end{equation}
  where
  \begin{align}
    \mathbf{V}_p & = \left\{\begin{array}{ll}
      \mmat{p} & \text{if $q \bttlt p \bttlt r$ and $\ensuremath{\operatorname{orient}} (p, q, r) = - 1$,}\\
      \mmat{p}^{- 1} & \text{if $r \bttlt p \bttlt q$ and $\ensuremath{\operatorname{orient}} (p, q, r) = + 1$,}\\
      \mathbf{I} & \text{otherwise},
    \end{array}\right.  \label{eq:Vp}\\
    \mathbf{V}_q & = \left\{\begin{array}{ll}
      \mmat{q}^{- 1} & \text{if $q = \min_{\bttlt} (p, q, r)$}\\
      & \text{\phantom{if }or $\left( p \bttlt q \bttlt r \text{ or } r \bttlt q \bttlt p \right)$}\\
      & \text{\phantom{if or }and $\ensuremath{\operatorname{orient}} (p, q, r) = + 1$},\\
      \mathbf{I} & \text{otherwise},
    \end{array}\right.  \label{eq:Vq}\\
    \mathbf{V}_r & = \left\{\begin{array}{ll}
      \mmat{r} & \text{if $q \bttlt r \bttlt p$ and $\ensuremath{\operatorname{orient}} (p, q, r) = - 1$,}\\
      \mmat{r}^{- 1} & \text{if $p \bttlt r \bttlt q$ and $\ensuremath{\operatorname{orient}} (p, q, r) = + 1$,}\\
      \mathbf{I} & \text{otherwise} .
    \end{array}\right.  \label{eq:Vr}
  \end{align}
\end{lemma}

\begin{figure}
  \begin{tabular}{ccc}
    \includegraphics{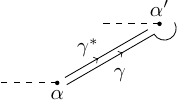} &  & \includegraphics{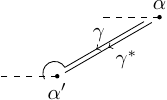}\\
    \includegraphics{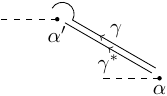} &  & \includegraphics{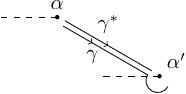}\\
    $\alpha \bttlt \beta$ & \quad & $\beta \bttlt \alpha$
  \end{tabular}
  \caption{\label{fig:close-triangle}The paths $\gamma$ (with hooks added for clarity) and $\gamma^{\ast}$}
\end{figure}

\begin{proof}
In order to obtain the composition rules \eqref{eq:compose-tmat}--\eqref{eq:Vr} for the matrices~$\tmat{\alpha}{\beta}$, we first define variants~$\tmatstar{\alpha}{\beta}$ of these matrices which satisfy simpler relations, and express~$\tmatstar{\alpha}{\beta}$ in terms of~$\tmat{\alpha}{\beta}$. Then we prove analogues of \eqref{eq:compose-tmat}--\eqref{eq:Vr} in the case of~$\tmatstar{\alpha}{\beta}$, and deduce \eqref{eq:compose-tmat}--\eqref{eq:Vr} from these analogues.

To define~$\tmatstar{\alpha}{\beta}$, let $\alpha, \beta \in \mathbb{C}$ be two points such that $\ensuremath{\operatorname{Im}} (\alpha) \neq \ensuremath{\operatorname{Im}} (\beta)$ and the open segment~$(\alpha, \beta)$ does not contain any singular point of~$\Delta$. Then~$\tmatstar{\alpha}{\beta}$ is the connection matrix defined in a similar way to~$\tmat{\alpha}{\beta}$ (\cf~Sec.~\ref{sec:matrices}) with the analytic continuation path~$\gamma$ replaced by
\[ \gamma^{\ast} = [(1 - \eta) \alpha + \eta \beta + i \eta^2, (1 - \eta) \beta + \eta \alpha + i \eta^2], \]
that is, by a segment contained in $(\alpha, \beta)$, shifted slightly up to avoid any ambiguity in the definition when $\ensuremath{\operatorname{Im}} (\alpha) =\ensuremath{\operatorname{Im}} (\beta)$. Observe that we have $\tmatstar{\beta}{\alpha} = ( \tmatstar{\alpha}{\beta} )^{- 1}$.

Extending~$\gamma^{\ast}$ by a small circular arc connecting $\gamma^{\ast} (0)$ to $\gamma (0)$ contained in the analyticity domain of the local basis~$\localbasis{\alpha}$ does not change the matrix~$\tmatstar{\alpha}{\beta}$. Similarly, one may connect $\gamma^{\ast} (1)$ to $\gamma (1)$ without leaving the analyticity domain of~$\localbasis{\beta}$ and this does not change~$\tmatstar{\alpha}{\beta}$. These extensions result in a path $\gamma^{\ast}_{\ensuremath{\operatorname{ext}}}$ with the same endpoints as~$\gamma$. When $\alpha \bttlt \beta$, the paths $\gamma$~and $\gamma^{\ast}_{\ensuremath{\operatorname{ext}}}$ are homotopic in $\mathbb{C}\backslash \Sigma$ (Figure~\ref{fig:close-triangle}, left), so that $\tmatstar{\alpha}{\beta}$ is equal to~$\tmat{\alpha}{\beta}$. In contrast, when $\beta \bttlt \alpha$, the loop $(\gamma^{\ast}_{\ensuremath{\operatorname{ext}}})^{- 1} \gamma$ goes around~$\beta$ once in the positive direction (Figure~\ref{fig:close-triangle}, right), and hence one has $( \tmatstar{\alpha}{\beta} )^{- 1}  \tmat{\alpha}{\beta} =\mathbf{M}_{\beta}$. We conclude that
\begin{equation}
  \tmatstar{\alpha}{\beta} = \left\{\begin{array}{ll}
    \tmat{\alpha}{\beta}, & \alpha \bttlt \beta,\\
    \mmat{\beta}^{- 1}  \tmat{\alpha}{\beta}, & \beta \bttlt \alpha .
  \end{array}\right. \label{eq:tmattotmat*}
\end{equation}

Now consider the composition of the matrices $\tmatstar{p}{q}$, $\tmat{q}{r}$, and $\tmat{r}{p}$. Again, we can connect the corresponding paths~$\gamma^{\ast}$ by small arcs contained in the analyticity domains of local solution bases at $p, q, r$ to get a closed loop. When the vertex~$v \in \{ p, q, r \}$ lying vertically between the other two is located to the left of the opposite side of~$t$ (\emph{i.e.}, when the local branch cut $v + (- \infty, 0]$ points outside of~$t$), since $t \cap \Sigma \subseteq \{ p, q, r \}$, the loop is contractible in $\mathbb{C}\backslash \Sigma$. The same conclusion holds in the limiting case where the bottom edge of~$t$ is horizontal, so that the middle vertex according to $\bttlt$ is the left endpoint of that edge. Otherwise, that is, when the middle vertex in $\bttlt$ order lies to the right of the opposite edge, including when the top edge is horizontal, the loop goes around that vertex once, positively or negatively depending on $\ensuremath{\operatorname{orient}} (p, q, r)$. Letting $u, v, w$ denote the points $p, q, r$ ordered so that $u \bttlt v \bttlt w$, we thus have
\[ \tmatstar{u}{v}  \tmatstar{w}{u}  \tmatstar{v}{w} = \mmat{v}^{\tau}, \quad \tau := \left\{\begin{array}{ll}
     0, & \ensuremath{\operatorname{orient}} (u, v, w) = - 1,\\
     1, & \ensuremath{\operatorname{orient}} (u, v, w) = + 1.
   \end{array}\right. \]

As $\tmatstar{\beta}{\alpha} = ( \tmatstar{\alpha}{\beta})^{- 1}$, this implies
\begin{align*}
     \tmatstar{u}{v} & = \mmat{v}^{\tau}  \tmatstar{w}{v}  \tmatstar{u}{w},  & \tmatstar{v}{u} & = \tmatstar{w}{u}  \tmatstar{v}{w}  \mmat{v}^{- \tau},\\
     \tmatstar{w}{u} & = \tmatstar{v}{u}  \mmat{v}^{\tau}  \tmatstar{w}{v},  & \tmatstar{u}{w} & = \tmatstar{v}{w}  \mmat{v}^{- \tau}  \tmatstar{u}{v},\\
     \tmatstar{v}{w} & = \tmatstar{u}{w}  \tmatstar{v}{u}  \mmat{v}^{\tau},  & \tmatstar{w}{v} & = \mmat{v}^{- \tau}  \tmatstar{u}{v}  \tmatstar{w}{u},
\end{align*}
and therefore, using~\eqref{eq:tmattotmat*},
\begin{align*}
     \tmat{u}{v} & = \mmat{v}^{\tau - 1}  \tmat{w}{v}  \tmat{u}{w}, & \mmat{u}^{- 1}  \tmat{v}{u} & = \mmat{u}^{- 1}  \tmat{w}{u}  \tmat{v}{w}  \mmat{v}^{- \tau},\\
     \mmat{u}^{- 1}  \tmat{w}{u} & = \mmat{u}^{- 1}  \tmat{v}{u}  \mmat{v}^{\tau - 1}  \tmat{w}{v},  & \tmat{u}{w} & = \mmat{w}^{- 1}  \tmat{v}{w}  \mmat{v}^{- \tau}  \tmat{u}{v},\\
     \tmat{v}{w} & = \tmat{u}{w}  \mmat{u}^{- 1}  \tmat{v}{u}  \mmat{v}^{\tau},  & \mmat{v}^{- 1}  \tmat{w}{v} & = \mmat{v}^{- \tau}  \tmat{u}{v}  \mmat{u}^{- 1}  \tmat{w}{u} .
\end{align*}
The expressions \eqref{eq:Vp}--\eqref{eq:Vr} follow by discussing according to the matching between $(u, v, w)$ and $(p, q, r)$. For example, in the case of~$\mathbf{V}_p$:
\begin{itemize}
  \item when $(p, r) = (v, w)$, that is, $q \bttlt p \bttlt r$, one has $\mathbf{V}_p = \mmat{v}^{\tau} = \mmat{p}^{\tau}$ and $\ensuremath{\operatorname{orient}} (u, v, w) = -\ensuremath{\operatorname{orient}} (p, q, r)$;
  
  \item when $(p, r) = (v, u)$, i.e., $r \bttlt p \bttlt q$, one has $\mathbf{V}_p = \mmat{p}^{- \tau}$, and in this case $\ensuremath{\operatorname{orient}} (u, v, w) =\ensuremath{\operatorname{orient}} (p, q, r)$;
  
  \item in all other cases, one has $\mathbf{V}_p =\mathbf{I}$. \qedhere
\end{itemize}
\end{proof}

We can deduce $\tmat{q}{p}$ from $\tmat{p}{q}$ using similar arguments.

\begin{lemma}
  \label{lem:inverse-connection}For two points $p \bttlt q$ such that $(p, q) \cap \Sigma = \varnothing$, one has $\tmat{q}{p} = \mmat{p}  \tmat{p}{q}^{- 1}$.
\end{lemma}

\begin{proof}
The loop obtained by connecting the paths~$\gamma$ associated with $\tmat{p}{q}$ and $\tmat{q}{p}$ (in this order) by small arcs around at $p$~and~$q$ not passing to the left of the corresponding points goes around~$p$ once, counterclockwise.
\end{proof}

Finally, it will be useful to handle the case where~$q$ lies between $p$~and~$r$. The corresponding rule is a direct consequence of the definition of $\tmat{\alpha}{\beta}$.

\begin{lemma}
  \label{lem:compose-flat}Let $p \, q \, r$ be a flat triangle with $q \in (p, r)$. Then one has $\tmat{p}{r} = \tmat{q}{r}  \tmat{p}{q}$.
\end{lemma}

\paragraph{Algorithm.}The above discussion leads us to Algorithm~\ref{algo:connect-all} for computing all connection matrices simultaneously.

\begin{algorithm}
  Connection matrices
  \begin{description}
    \label{algo:connect-all}\item[Input] An operator $\Delta \in \mathbb{K} [\zeta] [ \ddx{\zeta} ]$ with regular singular points.
    
    \item[Output] A family $(\mathbf{T}_{\alpha, \beta})_{\alpha, \beta \in \Sigma}$ of matrices $\mathbf{T}_{\alpha, \beta} \in \CComp^{\nu \times \nu}$, where $\Sigma$~is the set of singular points of~$\Delta$ and
    $\nu =\ensuremath{\operatorname{ord}} (\Delta)$.
  \end{description}
  \begin{enumerate}
    \item \label{step:spanning-tree}Compute a Euclidean minimum spanning tree $\mathcal{S}$ of the singular points of~$\Delta$.
    
    \item Initialize a mapping $\left( \tmat{\alpha}{\beta} \right)_{\alpha, \beta \in \Sigma}$ by setting $\tmat{\alpha}{\beta} = \bot$ for $\alpha \neq \beta$ and $\tmat{\alpha}{\alpha} =\mathbf{I}$ for all $\alpha$.
    
    \item For each $\alpha \in \Sigma$, compute the matrix $\mmat{\alpha}$ using~\eqref{eq:local-monodromy}.
    
    \item \label{step:numeric-tmat}For each edge $\{ \alpha, \beta \} \in \mathcal{S}$, with $\alpha \bttlt \beta$:
    
    \begin{enumerate*}
      \item \label{step:numeric-tmat-direct}Compute $\tmat{\alpha}{\beta}$ by solving the equation $\Delta (\hat{y}) = 0$ numerically.
      
      \item \label{step:invert}Set $\tmat{\beta}{\alpha} = \mmat{\alpha}  \tmat{\alpha}{\beta}^{- 1}$.
    \end{enumerate*}
    
    \item \label{step:connect-all:nonempty-triangles}Compute the set $\mathcal{V}$ of triples $\{ p, q, r \} \subseteq \Sigma$ such that the triangle $p \, q \, r$ is void (Definition~\ref{def:vacant}) and the set $\mathcal{F}$ of triples $\{ p, q, r \} \subseteq \Sigma$ such that $p, q, r$ are aligned and all distinct. Set $\mathcal{T}=\mathcal{V} \cup \mathcal{F}$.
    
    \item \label{step:loop}While there exist $\alpha, \beta \in \Sigma$ such that $\tmat{\alpha}{\beta} = \bot$:
    
    \begin{enumerate*}
      \item \label{step:choose-triangle}Choose a triple $t = \{ p, q, r \} \in \mathcal{T}$ such that $\tmat{p}{q} \neq \bot$ and $\tmat{q}{r} \neq \bot$.
      
      \item If $t \in \mathcal{F}$:
      
      \begin{enumerate*}
        \item \label{step:compose-flat}If $q \in (p, r)$, set $\tmat{p}{r} = \tmat{q}{r}  \tmat{p}{q}$ and $\tmat{r}{p} = \tmat{q}{p}  \tmat{r}{q}$.
        
        \item If $p \in (q, r)$, set $\tmat{p}{r} = \tmat{q}{r}  \tmat{q}{p}^{- 1}$ and $\tmat{r}{p} = \tmat{p}{q}^{- 1}  \tmat{r}{q}$.
        
        \item \label{step:compose-flat-end}If $r \in (p, q)$, set $\tmat{r}{p} = \tmat{q}{p}  \tmat{q}{r}^{- 1}$ and $\tmat{p}{r} = \tmat{r}{q}^{- 1}  \tmat{p}{q}$.
      \end{enumerate*}
      
      \item Else:
      
      \begin{enumerate*}
        \item Compare $p, q, r$ according to $\bttlt$.
        
        \item Compute $\ensuremath{\operatorname{orient}} (p, q, r)$.
        
        \item \label{step:close-triangle}Compute $\tmat{p}{r}$ and $\tmat{r}{p}$ using~\eqref{eq:compose-tmat}.
      \end{enumerate*}
      
      \item \label{step:update-queue}Remove from $\mathcal{T}$ all triples $\{ p, q', r \}$ such that $\tmat{p}{q'} \neq \bot$ and $\tmat{q'}{r} \neq \bot$.
    \end{enumerate*}
    
    \item Return $(\mathbf{T}_{\alpha, \beta})_{\alpha, \beta \in \Sigma}$.
  \end{enumerate}
\end{algorithm}

\begin{example}[Continued from Example~\ref{ex:running-intro}]
  In the running example introduced in Sec.~\ref{sec:running-example}, our implementation chooses for~$\mathcal{S}$ the spanning tree depicted on Figure~\ref{fig:stokes-values}. The six core connection matrices associated with the edges of the spanning tree oriented in the bottom-to-top order are computed by calling the ODE solver from \codestar{ore\_algebra} with a tolerance slightly lower than the one given on input to \codestar{stokes\_dict}. For instance, the computation of $\tmat{6}{0}$ can be reproduced separately with the command
\begin{alltt}
sage: diffop.borel_transform().numerical_transition_matrix(
....:                                         [6,0], 1e-50/26)
\end{alltt}
  \noindent and yields a $6 \times 6$ matrix whose first two columns are
  \[
    \tiny
    \begin{matrix}
       {}[0.207 \ldots 79 \pm 9.31 \cdot 10^{- 56}] + [- 0.007 \ldots 24 \pm 8.03 \cdot 10^{- 57}] \mathrm{i}                & [126.0 \ldots 22 \pm 2.50 \cdot 10^{- 52}]   \\
       {}[- 0.455 \ldots 77 \pm 3.68 \cdot 10^{- 56}] + [\pm 9.06 \cdot 10^{- 58}] \mathrm{i}                                & [- 200.0 \ldots 33 \pm 3.34 \cdot 10^{- 52}] \\
       {}[- 0.192 \ldots 97 \pm 1.99 \cdot 10^{- 55}] + [- 0.018 \ldots 49 \pm 2.45 \cdot 10^{- 56}] \mathrm{i}              & [- 62.47 \ldots 78 \pm 2.65 \cdot 10^{- 53}] \\
       {}[0.243 \ldots 32 \pm 3.94 \cdot 10^{- 55}] + [0.024 \ldots 86 \pm 3.26 \cdot 10^{- 56}] \mathrm{i}                  & [96.43 \ldots 32 \pm 4.14 \cdot 10^{- 53}]   \\
       {}[0.034 \ldots 24 \pm 1.02 \cdot 10^{- 56}] + [\pm 3.38 \cdot 10^{- 60}] \mathrm{i}                                  & [9.263 \ldots 28 \pm 3.96 \cdot 10^{- 54}]   \\
       {}[- 0.001 \ldots 52 \pm 4.45 \cdot 10^{- 57}] + [4.037 \ldots 94 \cdot 10^{- 5} \pm 5.46 \cdot 10^{- 60}] \mathrm{i} & [- 0.336 \ldots 15 \pm 1.74 \cdot 10^{- 55}] \\
     \end{matrix} \]
  We note that, in this special case, we could have used the fact that the operator~$\Delta$ is invariant by $\xi \mapsto \mathrm{e}^{\mathrm{i} \pi / 3} \xi$ to obtain a full set of core connection matrices from a single one instead of computing them independently. As a similar but more widely applicable optimisation, since $\Delta$ has coefficients in~$\mathbb{R} [\xi]$, we could have deduced the connection matrices associated with the edges of~$\mathcal{S}$ leading to Stokes values in the open upper half plane from the others by complex conjugation. Our code currently performs neither of these reductions automatically.
  
  The connection matrices along the edges of~$\mathcal{S}$ with the reverse orientation are computed using Lemma~\ref{lem:inverse-connection}. For instance, $\tmat{0}{6}$ is computed as $\mmat{6}  \tmat{6}{0}^{- 1}$ where $\mmat{6}$ is identity except for the first column which reads $(1, - 67 \mathrm{i} \pi / 17496, - 9347 \mathrm{i} \pi / 2519424, \ldots)^{\mathrm{T}}$, and in the first row of the product we find the coefficient $a = [- 141.140 \ldots 09 \pm 9.36 \cdot 10^{- 41}]$ which made an appearance in Example~\ref{ex:running-stokes-block}. We note in passing a significant loss of accuracy due to the inversion of~$\tmat{6}{0}$ in interval arithmetic. Similarly, $\tmat{- 6}{0}$ is computed as $\mmat{0}  \tmat{0}{- 6}^{- 1}$ where $\mmat{0}$ is identity except for a subdiagonal entry equal to~$2 \pi \mathrm{i}$ in the third column.
  
  The remaining connection matrices are computed in the main loop, using Lemmas \ref{lem:close-triangle} and~\ref{lem:compose-flat}.
  Let $\theta = \rme^{\rmi \pi / 3}$.
  At first one can compute, among others,
  \begin{itemize}
    \item $\tmat{- 6}{6} = \tmat{0}{6}  \tmat{- 6, 0}{}$ and $\tmat{6}{- 6} = \tmat{0}{- 6}  \tmat{6}{0}$ by Lemma~\ref{lem:compose-flat},
    
    \item $\tmat{0}{6 \theta} = \tmat{6}{6 \theta^{-1}}  \tmat{0}{6}$ by Lemma~\ref{lem:close-triangle} in the case $r \bttlt q \bttlt p$ and $\ensuremath{\operatorname{orient}} (p, q, r) = - 1$,
    
    \item $\tmat{6 \theta^{-1}}{0} = \tmat{6}{0}  \mmat{6}^{- 1}  \tmat{6 \theta^{-1}}{6}$ by Lemma~\ref{lem:close-triangle} with $p \bttlt q \bttlt r$ and $\ensuremath{\operatorname{orient}} (p, q, r) = + 1$,
    
    \item $\tmat{6 \theta}{6 \theta^{-1}} = \tmat{6}{6 \theta^{-1}}  \tmat{6 \theta}{6}$ and $\tmat{6 \theta^{-1}}{6 \theta} = \tmat{6}{6 \theta}  \mmat{6}^{- 1}  \tmat{6 \theta^{-1}}{6}$ by the same two rules.
  \end{itemize}
  Once $\tmat{6 \theta^{-1}}{0}$ is known, one can use it to compute, for instance
  \[ \tmat{6 \theta^{-1}}{6 \theta^{-2}} = \mmat{6 \theta^{-2}}^{- 1}  \tmat{0}{6 \theta^{-2}}  \tmat{6 \theta^{-1}}{0} \]
  by the case $p \bttlt r \bttlt q$ and $\ensuremath{\operatorname{orient}} (p, q, r) = + 1$ of Lemma~\ref{lem:close-triangle}. The algorithm continues in this way until all $42$ possible connection matrices are known.
\end{example}

\begin{proposition}
  \label{prop:algo-connect-all}Assume that the algorithm used at step~\ref{step:numeric-tmat-direct} computes~$\tmat{\alpha}{\beta} \in \CComp^{\nu \times \nu}$ given $\alpha$, $\beta$, and~$\Delta$. Then Algorithm~\ref{algo:connect-all}, when called on the Borel transform~$\Delta$ of the operator~$D$, computes the family $(\mathbf{T}_{\alpha, \beta})_{\alpha, \beta \in \Sigma}$ of connection matrices defined in Sec.~\ref{sec:matrices}.
\end{proposition}

\begin{proof}
Step~\ref{step:numeric-tmat-direct} makes sense and can be implemented as sketched above because $\Delta$~has regular singular points. The fact that the matrices $\tmat{\alpha}{\beta}$ computed at steps~\ref{step:numeric-tmat-direct}, \ref{step:invert}, \ref{step:compose-flat}--\ref{step:compose-flat-end}, and~\ref{step:close-triangle} are correct follows, respectively, from the assumption, Lemma~\ref{lem:inverse-connection}, Lemma~\ref{lem:compose-flat}, and Lemma~\ref{lem:close-triangle}. It remains to prove that the algorithm computes~$\tmat{\alpha}{\beta}$ for \emph{all}~($\alpha, \beta$) .

Let $\mathcal{T}_0 =\mathcal{V} \cup \mathcal{F}$ be the initial contents of the set~$\mathcal{T}$. If $| \Sigma | \leqslant 2$, then~$\mathcal{T}_0 = \varnothing$ and the algorithm terminates without going through the main loop. In this case all~$\tmat{\alpha}{\beta}$ have been computed. Assume from now on that $| \Sigma | \geqslant 3$. To start with, observe that every time the algorithm tests the loop condition at step~\ref{step:loop}:
\begin{itemize}
  \item for all $\alpha, \beta$, one has $\tmat{\alpha}{\beta} = \bot$ if and only if $\tmat{\beta}{\alpha} = \bot$, and the pairs $\{ \alpha, \beta \}$ such that $\tmat{\alpha}{\beta} \neq \bot$ are the edges of an undirected connected graph~$\mathcal{G}$;
  
  \item $\mathcal{T}$ is equal to the set of triples $\{ p, q, r \} \in \mathcal{T}_0$ such that at least one of $\tmat{p}{q}$, $\tmat{q}{r}$, and $\tmat{p}{r}$ is~$\bot$.
\end{itemize}
If the loop terminates without error, then $\mathcal{G}$~is the complete graph on~$\Sigma$ and the algorithm's output is correct. Thus it suffices to prove that every time step~\ref{step:choose-triangle} is reached, there exists a triangle~$\{ p, q, r \} \in \mathcal{T}_0$ with exactly two edges in~$\mathcal{G}$.
Since the edges of the Euclidean spanning tree~$\mathcal{S}$ do not cross, we can complete~$\mathcal{S}$ into a full triangulation\footnote{Here by \emph{full triangulation} of~$\Sigma$ we mean a simplicial complex whose vertices are exactly the elements of~$\Sigma$ and whose reunion is the convex hull of~$\Sigma$~\cite[Definition~2.2.1, 1.0.2]{DeLoeraRambauSantos2010}. For example, the only full triangulation of three distinct colinear points consists of the two segments joining the middle one to the other two, the three points themselves, and the empty face.}~$\mathcal{C}$ of~$\Sigma$ \cite[Lemma 3.1.2]{DeLoeraRambauSantos2010}.

If $\mathcal{C}$~contains an edge that is not yet in~$\mathcal{G}$, let~$\mathcal{U}$ be the graph dual to $\mathcal{C}\backslash\mathcal{S}$, that is, the graph whose vertices are the faces of the triangulation~$\mathcal{C}$ and where two vertices are adjacent when the corresponding faces are separated by an edge of the triangulation that is not part of the spanning tree. The graph~$\mathcal{U}$ is connected since $\mathcal{S}$~is acyclic and acyclic since $\mathcal{S}$~is connected, so $\mathcal{U}$~is a tree. Now consider the faces of~$\mathcal{C}$ in a postfix traversal of~$\mathcal{U}$ rooted on the face at infinity. The first face whose boundary is not contained in~$\mathcal{G}$ is a void triangle with exactly one edge not in~$\mathcal{G}$.

If now all edges of~$\mathcal{C}$ are in~$\mathcal{G}$ but there exists an element of~$\mathcal{V}$ with an edge not in~$\mathcal{G}$, consider a triangulation~$\mathcal{C}'$ obtained from~$\mathcal{C}$ by a \emph{flip}, that is, by replacing the common edge of two adjacent triangles by the other diagonal of the convex quadrilateral consisting of the reunion of the two triangles. The new faces created by the flip are elements of~$\mathcal{T}_0$ with at least two edges in~$\mathcal{G}$. If the remaining edge is not in~$\mathcal{G}$, we are done; otherwise, we can iterate the argument starting from~$\mathcal{C}'$. Since any element of~$\mathcal{V}$ can be completed into a triangulation and any triangulation can be connected to any other by a sequence of flips~\cite[Sec.~3.4.1]{DeLoeraRambauSantos2010}, we will eventually reach a triangle with exactly two edges in~$\mathcal{G}$.

Finally, if the edges of elements of~$\mathcal{V}$ are all in~$\mathcal{G}$, then the only ${\{ \alpha, \beta \} \not\in \mathcal{G}}$ with $\alpha \neq \beta$ in~$\Sigma$ are such that the segment $[\alpha, \beta]$ contains a third point of~$\Sigma$. Indeed, consider a segment $[\alpha, \beta]$ on which there is no other point of~$\Sigma$. Either the elements of~$\Sigma$ are all colinear, in which case $\{ \alpha, \beta \} \in \mathcal{S} \subseteq \mathcal{G}$, or there exists a point $\beta' \in \Sigma$ closest to the line through $\alpha$~and~$\beta$ but not on this line, and $\{ \alpha, \beta, \beta' \} \in \mathcal{T}$. Then for any closest $\alpha, \beta$ with $\{ \alpha, \beta \} \not\in \mathcal{G}$, there is a point $\beta' \in [\alpha, \beta]$ such that $\{ \alpha, \beta, \beta' \} \in \mathcal{F}$ has exactly two edges in~$\mathcal{G}$.
\end{proof}

\begin{remark}
  \label{rk:connect-all-details}
  \begin{enumerate}
    \item Any plane spanning tree of~$\Sigma$ can be used as the spine~$\mathcal{S}$. Taking a Euclidean minimum spanning tree helps limit the cost of the computation, but it would be even better to weight the edges by a more precise estimate of the cost of computing the corresponding $\tmat{\alpha}{\beta}$, possibly also taking into account its estimated condition number.
    
    \item At steps~\ref{step:compose-flat} and~\ref{step:close-triangle}, we compute $\tmat{p}{r}$ and $\tmat{r}{p}$ separately rather than deducing the latter from the former using Lemma~\ref{lem:inverse-connection} in order to limit the loss of numerical accuracy from repeated matrix inverses. Similary, adapting the algorithm used at step~\ref{step:numeric-tmat-direct} to compute~$\tmat{\alpha}{\beta}^{- 1}$ as the same time as~$\tmat{\alpha}{\beta}$ may yield more accurate results. Processing the triples $\{ p, q, r \}$ by decreasing accuracy of $\left( \tmat{p}{q}, \tmat{q}{r} \right)$ may also help.
    
    \item At step~\ref{step:connect-all:nonempty-triangles}, we can restrict ourselves to triangles that \emph{can be checked} not to contain any other singular points using interval arithmetic, provided that we verify that the number of matrices computed is equal to $| \Sigma |  (| \Sigma | - 1) / 2$ before exiting the algorithm and restart the whole computation with a higher working precision if not.
    
    \item As noted in Remark~\ref{rk:alignments}, the computation of~$\mathcal{F}$ (which involves exact operations in~$\mathbb{K}$) can be combined with that of the sets $\mathcal{A} \left( \dir \right)$ from Algorithm~\ref{algo:stokes}.
    
    \item While we only have $\mathrm{O} (n^2)$ matrices to compute, the cardinality of~$\mathcal{V}$ can reach $\Theta (n^3)$, computing it (step~\ref{step:connect-all:nonempty-triangles}) in the obvious way can take up to~$\Theta (n^4)$ operations, and keeping~$\mathcal{T}$ up to date (step~\ref{step:update-queue}) up to $\Theta (n^3)$ operations in total. (It is not too hard to arrange that step~\ref{step:choose-triangle} executes in constant time by modifying step~\ref{step:update-queue} to maintain a partition of~$\mathcal{T}$ according to the number of `known' edges.) These naïve estimates can likely be improved using of geometric data structures; however, estimating the complexity of the algorithm by assigning a unit cost to operations in~$\mathbb{K}$ or in $\CComp$ is highly unrealistic.
  \end{enumerate}
\end{remark}

\begin{remark}
  (Complexity with respect to the precision) \label{rk:complexity}When the operator~$D$ is fixed, the output of Algorithm~\ref{algo:stokes}, viewed as a family of matrices over~$\CComp$, is entirely determined, but it makes sense to ask about the cost of computing $2^{- p}$-approximations of the entries of these matrices as~$p$ tends to infinity. In this setting, the algorithm reduces to a fixed sequence of operations on computable complex numbers. In addition, for each intermediate result~$x_i$, there exists a constant~$c_i > 0$ such that, when $p$~is large enough, $x_i$ can be evaluated within an error bounded by~$2^{- p}$ starting from $2^{- p - c_i}$-approximations of $x_0, \ldots, x_{i - 1}$.
  
  Assume $\mathbb{K}= \bar{\mathbb{Q}}$, that is, the coefficients of the polynomial coefficients of the differential equation $D y = 0$ are algebraic numbers. Step~\ref{step:numeric-tmat-direct} of Algorithm~\ref{algo:connect-all} can then be performed in $\mathrm{O} (\mathrm{M} (p \log (p)^2))$ bit operations, where $\mathrm{M} (p)$ denotes the complexity of multiplying $p$-bit integers~\cite{vdH2001}. The remaining steps of Algorithm~\ref{algo:stokes} that depend on~$p$ boil down to additions and multiplications of complex numbers, evaluations of exponentials and logarithms, and evaluations of derivatives of $1 / \Gamma$ at points in~$\mathbb{K}$, all of which can be performed within the same complexity bound~\cite{Brent1978}. The complexity of Algorithm~\ref{algo:stokes} with respect to the target precision alone is thus $\mathrm{O} (\mathrm{M} (p \log (p)^2))$, or $\mathrm{O} (p \log (p)^3)$ using the bound $\mathrm{M} (p) = \mathrm{O} (p \log p)$ \cite{HarveyHoeven2021}. This is better by a factor $\asymp \log (p)$ than the best published complexity bound in this case, due to van der Hoeven~\cite{vdH2007}. However, van der Hoeven's algorithm applies to singularities of arbitrary multilevel.
  
  For a more general computable~$\mathbb{K}$, classical algorithms for the required basic operations~\cite{Borwein1987,ChudnovskyChudnovsky1988} yield a bound of $\mathrm{O} (p^{3 / 2} \log (p)^{\mathrm{O} (1)})$, not counting the time needed to compute numerical approximations of the coefficients of~$D$. This holds both for our method and for van der Hoeven's.
\end{remark}

\section{Further examples}\label{sec:more-examples}

We conclude with a few additional examples highlighting various aspects of the algorithm and implementation: one illustrating some degeneracies that can happen and the role of exponents on a simple hypergeometric equation, one where the formal solutions in the Laplace plane already involve logarithms, and one where the numerical analytic continuation step is more challenging. We also provide some data on the performance of our implementation on some of these examples. This is intended to give the reader an impression of the problems that are within reach with no serious optimisation effort, not as a thorough analysis. All measurements were performed on a laptop equipped with an Intel i7-10810U CPU, running SageMath version~10.7 and \codestar{ore\_algebra} revision \codestar{07a1adbb}.

\begin{example}[Variations on the confluent hypergeometric equation $D_{2, 1}$]
  \label{ex:hgeom} Let us consider again the operator~\eqref{eq:Dqp} (p.~\pageref{eq:Dqp}), now in the simple case $q = 2$, $p = 1$. In terms of $x = z^{- 1}$ and~$\partial$, the operator reads
  \[ D = x^{- 2} \partial^2 - ((\nu_1 + \nu_2 - 1) x^{- 1} + x^{- 2}) \partial + ((\nu_1 - 1)  (\nu_2 - 1) + \mu x^{- 1}) . \]
  The Stokes values are $0$~and~$1$ and our standard formal fundamental solution (Sec.~\ref{sec:local-bases}) is given by
  \begin{align*}
    y_1 & = x^{\mu}  (1 + (\mu - \nu_1 + 1)  (\mu - \nu_2 + 1) x + \cdots),\\
    y_2 & = \mathrm{e}^{- 1 / x} x^{\nu_1 + \nu_2 - \mu - 1}  (1 - (\mu - \nu_1)  (\mu - \nu_2) x + \cdots) .
  \end{align*}
  It is known that the associated Stokes matrix in the direction $\dir = 0$ is equal to
  \[ \left(\begin{array}{cc}
       1 & 0\\
       c & 1
     \end{array}\right), \quad c = \frac{2 \pi \mathrm{i}}{\Gamma (1 + \mu - \nu_1) \Gamma (1 + \mu - \nu_2)} . \]
  The Borel transform of $D$ for general parameters is
  \begin{equation}
    \Delta = \xi (\xi - 1)  \diff{\xi}{2} + ((- \nu_1 - \nu_2 + 5) \xi + \mu - 2)  \diff{\xi}{} + (\nu_1 - 2)  (\nu_2 - 2), \label{eq:D21-borel}
  \end{equation}
  with exponents $(0, \mu - 1)$ at~$0$ and $(0, \nu_1 + \nu_2 - \mu - 2)$ at~$1$. When $\nu_1 = 1$, one can replace~$D$ by $xD$ while preserving both the solution space and the fact that coefficients are polynomials in~$x^{- 1}$. The Borel transform then simplifies to
  \begin{equation}
    \Delta = \xi (\xi - 1)  \ddx{\xi} + (- \nu_2 + 2) \xi + \mu - 1 \label{eq:D21-borel-1}
  \end{equation}
  and only the free exponents~$\lambda_{[0]} = \mu - 1$ at~$0$ and $\lambda_{[1]} = \nu_2 - \mu - 1$ at~$1$ remain. Our implementation automatically performs this simplification.

    First consider the case $\mu = 1$, $\nu_1 = 1$, $\nu_2 = 3 / 2$. After simplification, one has $\Delta = x (x - 1)  \diff{\xi}{} + \frac{1}{2} x$, $\lambda_{[0]} = 0$, $\lambda_{[1]} = - 1 / 2$. It turns out that $\localbasis{1}$ coincides with the analytic continuation of~$- \localbasis{0}$, though the algorithm only detects the coincidence numerically, by computing $\tmat{0}{1} = [- 1.00 \ldots]$. Moreover, $\localbasis{0}$~is just $[\mathcal{B} (y_0)]$, that is, $\bmat{0} =\mathbf{I}_1$, and one has
    \[ \bfmat{1} = [2 \pi \mathrm{i} \mathrm{e}^{\mathrm{i} \pi (\lambda_{[1]} + 1)} \Gamma (\lambda_{[1]})^{- 1}] = \left[ 2 \sqrt{\pi} \right], \]
    so that the Stokes multiplier in the direction $\dir = 0$ is obtained as the product $1 \times (- 1.00 \ldots) \times 2 \sqrt{\pi} = - 3.54 \ldots$. In the direction $\dir = \pi$, the Stokes matrix is trivial since all solutions of~$\Delta$ are analytic at the origin, leading in the algorithm to $\bfmat{0} = 0$.
    
    Now suppose that $\mu = 1$, $\nu_1 = 4 / 3$, $\nu_2 = 5 / 3$. Then, both at~$0$ and at~$1$, the Borel transform~\eqref{eq:D21-borel} has a double exponent equal to~$0$. The associated bases of solutions are
    \begin{align*}
      \localbasis{0} (\zeta) & = \left( 1 + \frac{2}{9} \zeta \log \zeta + \frac{5}{9} \zeta + \cdots, \hspace{1em} 1 + \frac{2}{9} \zeta + \cdots \right),\\
      \localbasis{1} (\zeta) & = \left( 1 - \frac{2}{9}  (\zeta - 1) \log (\zeta - 1) - \frac{5}{9}  (\zeta - 1) + \cdots, \hspace{1em} 1 - \frac{2}{9}  (\zeta - 1) + \cdots \right) .
    \end{align*}
    The Borel transform matrices both send the relevant solution of~$D$ to the element of~$\localbasis{\alpha}$ free of logarithms, that is, $\bmat{0} = \bmat{1} = [0, 1]^{\mathrm{T}}$. As discussed in Sec.~\ref{BorelSol}, the operator~$\Delta$ has `extra' solutions that can be intepreted as the Borel transforms of solutions of $D (y) = a$ for $a \in \mathbb{C}$. The algorithm computes the connection matrix
    \[ \tmat{0}{1} = \left( \begin{array}{rr}
         0.90854 \ldots + 0.00000 \ldots \mathrm{i} & 0.27566 \ldots + 0.00000 \ldots \mathrm{i}\\
         0.63318 \ldots - 2.85427 \ldots \mathrm{i} & - 0.90854 \ldots - 0.86602 \ldots \mathrm{i}
       \end{array} \right) \]
    and deduces $\tmat{1}{0} = \tmat{0}{1}^{- 1} = \tmat{0}{1}^{- 1} + 0.86602 \ldots \mathrm{i}$. One has $\bfmat{0} = \bfmat{1} = (2 \pi \mathrm{i}, 0)$, leading to $c = - 1.73205 \ldots \mathrm{i}$, and in this case the Stokes matrix in the direction~$\dir = \pi$ is the transpose of that in the direction~$\dir = 0$.
    
    Finally set $\mu = 0$, $\nu_1 = 0$, $\nu_2 = 1 / 2$. This leads to an `unprepared' equation where neither of the exponents $\mu = 0$ of $y_1 (x) = 1 + \frac{1}{2} x + \cdots$ and $\nu_1 + \nu_2 - \mu - 1 = - 1 / 2$ of~$y_2 = \mathrm{e}^{- 1 / x} x^{- 1 / 2}$ lies in the open right-hand plane. The formal Borel transform
    \[ \mathcal{B} (y_1) (\zeta) = \delta_0 (\zeta) + \frac{1}{2} + \frac{3}{4} \zeta + \cdots \]
    contains a Dirac term, while
    \[ \mathcal{B} (\mathrm{e}^{1 / x} y_2) (\zeta) = - \frac{\zeta^{- 3 / 2}}{2 \sqrt{\pi}} \]
    is not integrable at the origin. One has
    \[ \Delta = \zeta (\zeta - 1)  \diff{\zeta}{2} + \left( \frac{9}{2} \zeta - 2 \right)  \diff{\zeta}{} + 3, \]
    and
    \begin{align*}
      \localbasis{0} (\zeta) & = \left( \zeta^{- 1} + \frac{1}{2} \log \zeta + \frac{3}{4} \zeta \log \zeta + \cdots, 1 + \frac{3}{2} \zeta + \cdots \right),\\
      \localbasis{1} (\zeta) & = \left( (\zeta - 1)^{- 3 / 2}, 1 - \frac{6}{5}  (\zeta - 1) + \cdots \right) .
    \end{align*}
    From these truncated expansions one computes $\bmat{0} = [0, 1 / 2]^{\mathrm{T}}$, mapping~$y_1$ to its Borel transform minus the Dirac terms. The connection matrix is
    \[ \tmat{0}{1} = \left( \begin{array}{rr}
         0.69314 \ldots \mathrm{i} & - 1.00000 \ldots\\
         2.00000 \ldots & 0.00000 \ldots
       \end{array} \right) \]
    and one has $\bfmat{0} = (2 \pi \mathrm{i} \mathrm{e}^{\mathrm{i} \pi (\lambda_{[0]} + 1)} \Gamma (\lambda_{[0]})^{- 1}, 0) = \left( - 4 \sqrt{\pi}, 0 \right)$, leading to the value $c = - 3.54490 \ldots \mathrm{i}$. In the other direction, one has $\bmat{1} = \left[ - 1 / \left( 2 \sqrt{\pi} \right), 0 \right]^{\mathrm{T}}$ and $\bfmat{0} = (2 \pi \mathrm{i}, 0)$, so the Stokes matrix is trivial.
\end{example}

\begin{example}[Pólya walks]
  \label{ex:polya} Denote by $w^{(d)}_{\ell}$ the number of nearest-neighbor walks on~$\mathbb{Z}^d$ starting from the origin and returning to it after $\ell$~steps. For $3 \leqslant d \leqslant 15$, consider `the' differential operator~$\Delta^{(d)}$ of minimum order that annihilates the generating series $\sum_{\ell = 0}^{\infty} w_{\ell}^{(d)}$ and let~$D^{(d)}$ be its inverse Borel transform.
  
  In the case $d = 3$, one has
  \begin{align*}
    \Delta^{(d)} & = \xi^2  (4 \xi^2 - 1)  (36 \xi^2 - 1)  \diff{\xi}{3} + (1296 \xi^5 - 240 \xi^3 + 3 \xi) \diff{\xi}{2}\\
    &   + (2592 \xi^4 - 288 \xi^2 + 1) \ddx{\xi} + 864 \xi^3 - 48 \xi,\\
    D^{(d)} & = 144 x^{- 3} \partial^6 - 1296 x^{- 2} \partial^5 + (- 40 x^{- 3} + 2592 x^{- 1}) \partial^4 + (240 x^{- 2} - 864) \partial^3\\
    &   + (x^{- 3} - 288 x^{- 1}) \partial^2 + (- 3 x^{- 2} + 48) \partial + x^{- 1},
  \end{align*}
  and the components of the standard fundamental solution of~$D^{(d)}$ are
  \begin{align*}
    y_1 (x), y_6 (x) & = \mathrm{e}^{\pm 1 / (6 x)} x^{3 / 2}  \left( 1 \pm \tfrac{27}{4} x + \cdots \right),\\
    y_2 (x), y_5 (x) & = \mathrm{e}^{\pm 1 / (2 x)} x^{3 / 2}  \left( 1 \pm \tfrac{17}{4} x + \cdots \right),\\
    y_3 (x) & = x \ln x + 12 x^3 \ln x + 32 x^3 + \cdots,\\
    y_4 (x) & = x + 12 x^3 + \cdots .
  \end{align*}
  The operator~$\Delta^{(d)}$ has exponents $(0, 0, 0)$ at~$0$ and $(0, 1 / 2, 1)$ at each of the other four singular points. At~$0$, our implementation uses the local basis
  \begin{align*}
    \localbasis{0} (\zeta) & = \left( \tfrac{1}{2} \ln^2 \zeta + 3 \zeta^2 \ln^2 \zeta + 7 \zeta^2 \ln \zeta + \mathrm{O} (\zeta^4 \ln^2 \zeta) \right.,\\
    &  \ln \zeta + 6 \zeta^2 \ln \zeta + 7 \zeta^2 + \mathrm{O} (\zeta^4 \ln^2 \zeta),\\
    &   1 + 6 \zeta^2 + \mathrm{O} (\zeta^4)) .
  \end{align*}
  Since the solutions $y_3(x)$ and $y_4(x)$ of~$D$ attached to the Stokes value~$0$ only involve logarithms linearly and, by Lemma~\ref{dico}, the Borel transform cannot increase the degree in $\log(x)$, we already know that the first coordinate in this basis of their Borel transforms will be zero.
  Writing
  \begin{align*}
    \mathcal{B} (y_3) (\zeta) & = \ln \zeta + \gamma + \frac{\zeta^2}{\Gamma (3)}  (12 \ln \zeta + 12 \gamma - 18 + 32) + \cdots,\\
    \mathcal{B} (y_4) (\zeta) & = \zeta + \frac{12 \zeta^2}{\Gamma (3)} + \cdots
  \end{align*}
  we obtain
  \[ \bmat{0} = \left[\begin{array}{cc}
       0 & 0\\
       1 & 0\\
       \gamma & 1
     \end{array}\right] . \]
  Similarly, the expansion of the integrals
  \begin{align*}
    \int_{\Gamma_{0}} (\ln \zeta + 6 \zeta^2 \ln \zeta + 7 \zeta^2 + \cdots) \mathrm{e}^{- \zeta / x} \mathrm{d} \zeta & = 2 \pi \mathrm{i} x + 24 \pi \mathrm{i} x^3 + \cdots\\
    \int_{\Gamma_{0}} \left( \dfrac{1}{2} \ln^2 \zeta + \cdots \right) \mathrm{e}^{- \zeta / x} \mathrm{d} \zeta & = \frac{1}{2}  [2 \pi \mathrm{i} x (2 \ln x - 2 \gamma - 2 \pi \mathrm{i})] + \cdots
  \end{align*}
  yields
  \[ \bfmat{0} = \left[\begin{array}{ccc}
       2 \pi \mathrm{i} & 0 & 0\\
       2 \pi^2 - 2 \pi \mathrm{i} \gamma & 2 \pi \mathrm{i} & 0
     \end{array}\right] . \]
  The calculation of $\bmat{\alpha}$~and~$\bfmat{\alpha}$ for the other Stokes values~$\alpha$ is closer to the previous examples and we omit it. After computing the connection matrices over each of the four segments delimited by two consecutive Stokes values and forming the appropriate combinations, we eventually obtain the following Stokes matrix in the direction~$\dir = 0$:
  \[ \begin{split}
     \small{\left[\begin{array}{cccccc}
       1                                & 0                             & 0                               & \cdots \\
       0.000 \ldots - 18.000 \ldots i   & 1                             & 0                               & \cdots \\
       13.540 \ldots + 0.000 \ldots i   & 0.000 \ldots + 1.504 \ldots i & 1                               & \cdots \\
       - 7.815 \ldots - 14.179 \ldots i & 0.000 \ldots - 0.868 \ldots i & 0                               & \cdots \\
       36.000 \ldots + 0.000 \ldots i   & 0.000 \ldots + 8.000 \ldots i & 0.000 \ldots - 2.930 \ldots i   & \cdots \\
       0.000 \ldots - 4.000 \ldots i    & 1.333 \ldots - 0.000 \ldots i & - 0.976 \ldots + 1.772 \ldots i & \cdots 
     \end{array}\right.} \\
     \small{\left.\begin{array}{ccccc}
       0                               & 0                             & 0 \\
       0                               & 0                             & 0 \\
       0                               & 0                             & 0 \\
       1                               & 0                             & 0 \\
       0.000 \ldots - 5.077 \ldots i   & 1.                            & 0 \\
       - 1.692 \ldots + 0.000 \ldots i & 0.000 \ldots - 0.666 \ldots i & 1
     \end{array}\right]}
     \end{split} \]
  (along with a second, upper triangular one in the direction $\dir = \pi$).

  \begin{table}
    \begin{tabular}{rrrrr}
      \toprule
      $\nu$ & $n$ & $\varepsilon = 10^{- 10}$ & $\varepsilon = 10^{- 100}$ & $\varepsilon = 10^{- 1000}$\\
      \midrule
      3 & 6 & 0.3 & 0.9 & 4.3\\
      4 & 7 & 0.6 & 0.5 & 6.8\\
      5 & 10 & 1.3 & 2.0 & 7.8\\
      6 & 11 & 1.4 & 2.2 & 9.3\\
      7 & 14 & 3.5 & 5.6 & 23.5\\
      8 & 15 & 7.7 & 6.9 & 26.5\\
      9 & 18 & 23.9 & 18.8 & 62.5\\
      10 & 19 & 29.7 & 24.0 & 81.9\\
      11 & 22 & 53.8 & 38.0 & 146.0\\
      12 & 23 & --- & 55.2 & 199.3\\
      13 & 26 & 129.2 & 70.0 & 306.2\\
      14 & 27 & --- & 96.7 & 404.1\\
      15 & 30 & 261.0 & 115.3 & 558.8\\
      \bottomrule
    \end{tabular}
    \medskip
    \caption{\label{table:polya}Time in seconds to compute all the Stokes matrices of~$D^{(d)}$
    (which has order~$n$ and degree $\nu=d$ in $1/x$)
    with an absolute error on each entry not exceeding~$\varepsilon$. Missing entries correspond to cases where some intermediate computation failed due to insufficient precision.}
  \end{table}

  More generally, $D^{(d)}$ is an operator of order $n = 2 d$ when $d$~is odd and $n = 2 d - 1$ when $d$~is even, with coefficients of degree~$\nu = d$ with our usual conventions. Its Stokes values include $0$~with multiplicity~$d - 1$ and $\lfloor (d + 3) / 2 \rfloor$ pairs of opposite rational numbers; in particular, they all lie on the real axis. The Stokes value~$0$ has multiplicity~$d - 1$ and a unique associated exponent equal to~$0$, so that the corresponding solutions of~$D^{(d)}$ involve logarithms up to the power~$d - 2$. All other Stokes values have multiplicity one and an associated exponent equal to $d / 2 - 1$. Correspondingly, the solutions of~$\Delta^{(d)}$ at~$0$ involve logarithms up to the power~$d - 1$, while, at other singular points, there is a basis of solutions consisting of~$d - 1$ power series and one solution with either half-integer exponents (for odd~$d$) or logarithmic terms (for even~$d$).
  
  Table~\ref{table:polya} shows the time taken by our implementation to compute the Stokes matrices of~$D^{(d)}$ to various accuracies. In each case, we first call the \codestar{stokes\_dict} function with the indicated~$\varepsilon$. If the maximum radius~$\delta$ of the computed intervals exceed~$\varepsilon$, we repeat the computation with $\varepsilon$ replaced by $\varepsilon^{2.1} / \delta$, and so on until all Stokes multipliers are known with an error of less than the initial~$\varepsilon$. The reported running time corresponds to the total of all runs. Low-precision computations are slower than intermediate-precision ones for large instances because this procedure can require more iterations when the initial precision is small. More generally, we observe a very slow (less than linear in the number of digits!) growth of running times as the error tolerance decreases, suggesting a significant loss of accuracy over the course of the computation and weaknesses in automatic working precision management.
\end{example}

\begin{example}[A tunnel]
  \label{ex:tunnel} For $\eta > 0$, we consider the operator
  \[ D = x^{- 2} \partial^4 + (1 - 4 x^{- 2}) \partial^3 + (\eta^2 x^{- 1} + (5 + \eta^2) x^{- 2}) \partial + 2 (1 + \eta^2) x^{- 2} . \]
  The Stokes values are $0$, $2$, and $1 \pm \eta \mathrm{i}$. The transformed operator $\Delta$ has exponents $(0, 0)$ at the origin, $(- 4, 1)$ at~$2$, and $(1, \mp (2 / \eta - \eta / 2) \mathrm{i})$ at $1 \pm \eta \mathrm{i}$. In the direction $\dir = 0$, there is a single nontrivial Stokes multiplier, corresponding to the pair of Stokes values $(0, 2)$. 
  The numerical computation of this Stokes multiplier is challenging for small~$\eta$ because the line segment connecting these Stokes values passes between $1\pm\eta\rmi$, where large imaginary exponents are present.
  
  To test the behaviour of our implementation in this situation, we use the following SageMath code, which computes the Stokes matrices of~$D$ and extracts the Stokes multiplier of interest:
\begin{alltt}
sage: d = x^2*Dx
sage: dop = (d^4 + (x^2 - 4)*d^3 + (eta^2*x + 5 + eta^2)*d^2
....:         + (2*x - 2 - 2*eta^2)*d + (2 + 2*eta^2)*x)
sage: stokes = stokes_dict(dop, 1e-50)
sage: stokes[1][3,0]
\end{alltt}
  \noindent Table~\ref{table:tunnel} shows the output for various values of~$\eta$ along with the corresponding running times. We observe that the code forgoes computing the result with an absolute error matching the tolerance of $10^{- 50}$ passed as input but does return a result with a \emph{relative} error of this order of magnitude. This outcome is somewhat sensitive to the choice of the tolerance. With $\eta = 10^{- 4}$, it becomes necessary to decrease the tolerance for the computation to succeed at all.

  \begin{table}
    \footnotesize
    \begin{tabular}{cr@{$\,+\,$}rc}
      \toprule
      $\eta$ & \multicolumn{2}{c}{computed Stokes constant}& time (s)\\
      \midrule
      $10^{- 1}$ & $[\pm 1.75 \cdot 10^{11}]  $ & $[- 3.34009\dots{}269 \cdot 10^{52} \pm 3.58 \cdot 10^{11}] \,\mathrm{i}$ & 0.4\\
      $10^{- 2}$ & $[\pm 1.32 \cdot 10^{490}] $ & $[- 6.39027\dots{}429 \cdot 10^{543} \pm 3.01 \cdot 10^{489}] \,\mathrm{i}$ & 3.5\\
      $10^{- 3}$ & $[\pm 4.92 \cdot 10^{5403}]$ & $[- 3.73685\dots{}673 \cdot 10^{5455} \pm 6.22 \cdot 10^{5403}] \,\mathrm{i}$ & 100 \\
      \bottomrule
    \end{tabular}
    \medskip
    \caption{\label{table:tunnel}Computed Stokes constants and timings for Example~\ref{ex:tunnel}.}
  \end{table}
\end{example}

\bibliography{LMReq}

\end{document}